\documentclass[a4paper]{article}
\usepackage{a4wide}
\usepackage{graphicx}
\usepackage{amsmath}
\usepackage{amsthm}
\usepackage{amssymb, latexsym, mathrsfs}
\usepackage{dsfont}
\usepackage{enumerate}
\usepackage{subfigure}
\usepackage{algorithm}
\usepackage{color}
\usepackage{transparent}
\usepackage[affil-it]{authblk}
\theoremstyle{plain}
\newtheorem{thm}{Theorem}
\newtheorem{assum}{Assumption}

\newtheorem{defi}{Definition}
\newtheorem{rmk}{Remark}
\newtheorem{prop}{Proposition}
\newtheorem{lem}{Lemma}

\def\Rset{\mathbb{R}}
\def\Nset{\mathbb{N}}
\def\Zset{\mathbb{Z}}

\newcommand{\AAA}{{\mathcal A}}
\newcommand{\BB}{{\mathcal B}}
\newcommand{\CC}{{\mathcal C}}

\newcommand{\II}{{\mathcal I}}
\newcommand{\KK}{{\mathcal K}}

\newcommand{\MM}{{\mathcal M}}
\newcommand{\NN}{{\mathcal N}}

\newcommand{\PP}{{\mathcal P}}

\newcommand{\SSS}{{\mathcal S}}

\newcommand{\ZZ}{{\mathcal Z}}

\newcommand{\hx}{{\hat x}}

\newcommand{\bu}{{\bar u}}

\newcommand{\bv}{{\bar v}}
\newcommand{\bz}{{\bar z}}

\newcommand{\mbf}[1]{\mathbf{#1}}
\newcommand{\px}{{x^+}}
\newcommand{\subss}[2]{#1{_{[#2]}}}
\newcommand{\tx}{{\tilde x}}
\newcommand{\tz}{{\tilde z}}

\newcommand{\hpx}{{\hat{x}}^+}

\newcommand{\pz}{{{z}^+}}
\newcommand{\diag}{\mbox{diag}}

\newcommand{\Xset}{\mathbb{X}}

\newcommand{\bUset}{{\bar{\mathbb{U}}}}
\newcommand{\bZset}{{\bar{\mathbb{Z}}}}
\newcommand{\Uset}{\mathbb{U}}

\newcommand{\hXset}{\mathbb{{\hat{X}}}}

\newcommand{\Wset}{\mathbb{{{W}}}}

\newcommand{\Vset}{\mathbb{V}}
\newcommand{\ball}[1]{{B_{#1}}}
\newcommand{\abs}[1]{{|{#1}|}}
\newcommand{\norme}[2]{{||{#1}||_{#2}}}

\newcommand{\Pset}{\mathbb{P}}

\newcommand{\One}{\textbf{1}}

\newcommand{\convh}{\mbox{convh}}

\newcommand{\hbx}{{\bar{\hat x}}}
\newcommand{\bkappa}{{\bar{\kappa}}}

\newcommand{\ba}[1]{\begin{array}{#1}}
\newcommand{\ea}{\end{array}}

\newcommand{\matr}[1]{
\begin{bmatrix}
    #1
\end{bmatrix}
}

\begin{document}

     \title{Plug-and-Play Model Predictive Control based on robust control invariant sets\thanks{The research leading to these results has received funding from the European Union Seventh Framework Programme [FP7/2007-2013]  under grant agreement n$^\circ$ 257462 HYCON2 Network of excellence.}}
     \author{Stefano Riverso%
       \thanks{Electronic address: \texttt{stefano.riverso@unipv.it}; Corresponding author}}

     \author{Marcello Farina%
       \thanks{Electronic address: \texttt{farina@elet.polimi.it}}}

     \author{Giancarlo Ferrari-Trecate%
       \thanks{Electronic address: \texttt{giancarlo.ferrari@unipv.it}\\S. Riverso and G. Ferrari-Trecate are with Dipartimento di Ingegneria Industriale e dell'Informazione, Universit\`a degli Studi di Pavia, via Ferrata 1, 27100 Pavia, Italy\\M. Farina is with Dipartimento di Elettronica e Informazione, Politecnico di Milano, via Ponzio 34/5, 20133 Milan, Italy}}

     \affil{Dipartimento di Ingegneria Industriale e dell'Informazione\\Universit\`a degli Studi di Pavia}
     \date{\textbf{Technical Report}\\ December, 2013}

     \maketitle

     \begin{abstract}
          In this paper we consider a linear system represented by a coupling graph between subsystems and propose a distributed control scheme capable to guarantee asymptotic stability and satisfaction of constraints on system inputs and states. Most importantly, as in \cite{Riverso2012a,Riverso2013c} our design procedure enables plug-and-play (PnP) operations, meaning that (i) the addition or removal of subsystems triggers the design of local controllers associated to successors to the subsystem only and (ii) the synthesis of a local controller for a subsystem requires information only from predecessors of the subsystem and it can be performed using only local computational resources. Our method hinges on local tube MPC controllers based on robust control invariant sets and it advances the PnP design procedure proposed in \cite{Riverso2012a,Riverso2013c} in several directions. Quite notably, using recent results in the computation of robust control invariant sets, we show how critical steps in the design of a local controller can be solved through linear programming. Finally, an application of the proposed control design procedure to frequency control in power networks is presented.
     \end{abstract}

     \newpage

     \section{Introduction}
          \label{sec:intro}
          Centralized advanced control systems are nowadays ubiquitous for guaranteeing optimality, stability, robustness and reliability in a wide range of applications, spanning from process industry to chemical plants and automotive/transportation systems. Their implementation requires measurements to be conveyed instantaneously and simultaneously to a central station, where the control law is computed. Then, the control variables must be transmitted from the central station to each actuator. As a consequence, when actuators and sensors are sparse over a wide geographical area, communication and computational requirements can be very demanding. Moreover, centralized design often hinges on the detailed knowledge of the whole plant and the availability of a dynamical model that must be amenable to designing control laws through advanced tools including, e.g. optimization steps on solution to matrix inequalities.

          The ever-increasing complexity and size of process plants, manufacturing systems, transportation systems, and power distribution networks call for the development of distributed control architectures. Novel control approaches should enable decisions to be taken by a number of distributed regulators with low computational burden and collocated with sensors and actuators. Furthermore, synthesis problems should be decomposed into independent (or almost independent) sub-problems.

          Decentralized and distributed control methods and architectures have been studied in the last decades to provide a solution to these issues. They are based on specific system representations requiring to split of the overall system into interacting subsystems. Interestingly, the interactions between subsystems can be described using directed graph representations which can be extremely useful to unveil the so-called system's structural properties \cite{Siljak1991,Lunze1992}. Consistently with this graph-theoretical characterization, in this work the concepts of predecessor and successor will be used. Namely, a subsystem $j$ is said to be a predecessor (respectively, a successor) of subsystem $i$ if its state variables directly influence the dynamics of $i$ (respectively, if its dynamics is directly influenced by the input/state variables of $i$).

          In the last years many decentralized and distributed control schemes based on Model Predictive Control (MPC) \cite{Scattolini2009} have been proposed, in view of the possibility of coping with constraints on system's variables \cite{Mayne2000} and the fact that predictions of the input and state trajectories can be exchanged among distributed regulators to guarantee stability, robustness, and global optimality. Available distributed MPC schemes span from cooperative \cite{Sanchez2008,Stewart2010} (guaranteeing system-wide optimality), to non-cooperative ones, which require limited computational load, memory, and transmission of information \cite{Trodden2010,Camponogara2002,Farina2012,Riverso2012e}.

        One the main problems of the above MPC-based controllers is the need of a centralized off-line design phase. While simplified and even decentralized stability analysis tools for distributed systems are available, based on aggregated models (see, e.g., the vector Lyapunov function method \cite{Lunze1992}), the design phase can in general be carried out in a decentralized fashion only in special cases \cite{Bakule1988}. To overcome this problem, in \cite{Riverso2012a,Riverso2013c} a novel solution based on the PnP design paradigm \cite{Stoustrup2009} has been proposed. PnP design besides synthesis decentralization, imposes the following constraints: when a subsystem is added to an existing plant (i) local controllers have to be designed only for the subsystem and its successors; (ii) the design of a local controller uses only information from the subsystem and its predecessors. Quite remarkably, these requirements allow controllers to be computed using computational resources collocated with subsystems. Furthermore, the complexity of controller design and implementation, for a given subsystem, scales with the number of its predecessors only.

          As in \cite{Riverso2012a,Riverso2013c} we propose a PnP design procedure hinging on the notion of tube MPC \cite{Rakovic2005b} for handling coupling among subsystems, and aim at stabilizing the origin of the whole closed-loop system while guaranteeing satisfaction of constraints on local inputs and states. However, we advance the design procedure in \cite{Riverso2012a,Riverso2013c} in several directions. First, in \cite{Riverso2012a,Riverso2013c} the most critical step in the design of local MPC controllers required the solution to nonlinear optimization problems. In this paper, using local tube MPC regulators based on Robust Control Invariant (RCI) sets, we guarantee overall stability and constraints satisfaction solving Linear Programming (LP) problems only. Second, in \cite{Riverso2012a,Riverso2013c} stability requirements where fulfilled imposing an aggregate sufficient small-gain condition for networks. In the present paper, we resort instead to set-based conditions that are usually less conservative. Third, while methods in \cite{Riverso2012a,Riverso2013c} were tailored to decentralized control only, the new PnP-DeMPC can also admit a distributed implementation tacking advantage of pieces of information received online by predecessors. As for any decentralized synthesis procedure our method involves some degree of conservativeness \cite{Bakule1988} and its potential application to real-world systems will be discussed through examples. In particular, as in \cite{Riverso2012a,Riverso2013c} we present an application of PnP-DeMPC to frequency control in power networks. Furthermore we highlight computational advantages brought about by our method by considering the control of a large array of masses connected by springs and dampers.

          The paper is structured as follows. The design of decentralized controllers is introduced in Section \ref{sec:modelcontroller} with a focus on the assumptions needed for guaranteeing asymptotic stability of the origin and constraint satisfaction. In Section \ref{sec:main} we discuss how to design local controllers in a distributed fashion by solving LP problems and in Section \ref{sec:plugplay} we describe PnP operations. In Section \ref{sec:distrcontr} we show how to enhance the control scheme taking advantage of pieces of information received from predecessors. In Section \ref{sec:examples} we present applicative examples and Section \ref{sec:conclusions} is devoted to concluding remarks.

          \textbf{Notation.} We use $a:b$ for the set of integers $\{a,a+1,\ldots,b\}$. The column vector with $s$ components $v_1,\dots,v_s$ is $\mbf v=(v_1,\dots,v_s)$. The function $\diag(G_1,\ldots,G_s)$ denotes the block-diagonal matrix composed by $s$ block $G_i$, $i\in 1:s$. The symbols $\oplus$ and $\ominus$ denote the Minkowski sum and difference, respectively, i.e. $A=B\oplus C$ if $A=\{a:a=b+c,\mbox{ for all }b\in B \mbox{ and }c\in C\}$ and $A=B\ominus C$ if $a\oplus C\subseteq B,~\forall a\in A$. Moreover, $\bigoplus_{i=1}^sG_i=G_1\oplus\ldots\oplus G_s$. $\ball{\rho}(z)\subset\Rset^n$ denotes the open ball of radius $\rho$ centered in $z\in\Rset^n$. Given a set $\Xset\subset\Rset^n$, $\convh(\Xset)$ denotes its convex hull. The symbol $\One$ denotes a column vector of suitable dimension with all elements equal to $1$.
          \begin{defi}[Robust Control Invariant (RCI) set]
            Consider the discrete-time Linear Time-Invariant (LTI) system $x(t+1)=Ax(t)+Bu(t)+w(t)$, with $x(t)\in\Rset^n$, $u(t)\in\Rset^m$, $w(t)\in\Wset^n$ and subject to constraints $u(t)\in\Uset\subseteq\Rset^m$ and $w(t)\in\Wset\subset\Rset^n$. The set $\Xset\subseteq\Rset^n$ is an RCI set with respect to $w(t)\in\Wset$, if $\forall x(t)\in\Xset$ then there exist $u(t)\in\Uset$ such that $x(t+1)\in\Xset$, $\forall w(t)\in\Wset$.
          \end{defi}

     \section{Decentralized tube-based MPC of linear systems}
          \label{sec:modelcontroller}
          We consider a discrete-time LTI system
          \begin{equation}
            \label{eq:system}
            \mbf{\px}=\mbf{A}\mbf{x}+\mbf{B}\mbf{u}
          \end{equation}
          where $\mbf x\in\Rset^n$ and $\mbf u\in\Rset^m$ are the state and the input, respectively, at time $t$ and $\mbf{\px}$ stands for $\mbf{x}$ at time $t+1$. We will use the notation $\mbf x(t)$, $\mbf u(t)$ only when necessary. The state is composed by $M$ state vectors $\subss x i\in\Rset^{n_i}$, $i\in\MM=1:M$ such that $\mbf{x}=(\subss x 1,\dots,\subss x M)$, and $n=\sum_{i\in\MM}n_i$. Similarly, the input is composed by into $M$ vectors $\subss u i\in\Rset^{m_i}$, $i\in\MM$ such that $\mbf{u}=(\subss u 1,\dots,\subss u M)$ and $m=\sum_{i\in\MM}m_i$.

          We assume the dynamics of the $i-th$ subsystem is given by
          \begin{equation}
            \label{eq:subsystem}
            \subss\Sigma i:\quad\subss \px i=A_{ii}\subss x i+B_i\subss u i+\subss w i
          \end{equation}
          \begin{equation}
            \label{eq:couplingW}
            \subss w i = \sum_{j\in\NN_i}A_{ij}\subss x j
          \end{equation}
          where $A_{ij}\in\Rset^{n_i\times n_j}$, $i,j\in\MM$, $B_i\in\Rset^{n_i\times m_i}$ and $\NN_i$ is the set of predecessors to subsystem $i$ defined as $\NN_i=\{j\in\MM:A_{ij}\neq 0, i\neq j\}$.

          According to \eqref{eq:subsystem}, the matrix $\mbf A$ in \eqref{eq:system} is decomposed into blocks $A_{ij}$, $i,j\in\MM$. We also define $\mbf{A_D}=\diag(A_{11},\dots,A_{MM})$ and $\mbf{A_C}=\mbf{A}-\mbf{A_D}$, i.e. $\mbf{A_D}$ collects the state transition matrices of every subsystem and $\mbf{A_C}$ collects coupling terms between subsystems. From \eqref{eq:subsystem} one also obtains $\mbf{B}=\diag(B_1,\dots,B_M)$ because submodels \eqref{eq:subsystem} are input decoupled.\\

          \begin{assum}
            \label{ass:controllable}
            The matrix pair $(A_{ii},B_i)$ is controllable, $\forall i\in\MM$.
            \begin{flushright}$\blacksquare$\end{flushright}
          \end{assum}

          In this Section we propose a decentralized controller for \eqref{eq:system} guaranteeing asymptotic stability of the origin of the closed-loop system and constraints satisfaction.

          In the spirit of optimized tube-based control \cite{Rakovic2005b}, we treat $\subss w i$ as a disturbance and define the nominal system $\subss{\hat\Sigma} i$ as
          \begin{equation}
            \label{eq:nominalsubsystem}
            \subss{\hat\Sigma} i:\quad\subss \hpx i=A_{ii}\subss\hx i+B_i\subss v i
          \end{equation}
          where $\subss v i$ is the input. Note that \eqref{eq:nominalsubsystem} has been obtained from \eqref{eq:subsystem} by neglecting the disturbance term $\subss w i$.

          We equip subsystems $\subss\Sigma i$, $i\in\MM$ with the constraints $\subss x i \in \Xset_i,~\subss u i \in \Uset_i$, define the sets $\Xset=\prod_{i\in\MM}\Xset_i$, $\Uset=\prod_{i\in\MM}\Uset_i$ and consider the collective constrained system \eqref{eq:system} with
          \begin{equation}
            \label{eq:constraints}
            \mbf x \in \Xset, ~
            \mbf u \in \Uset.
          \end{equation}
          The next Assumption characterizes the shape of constraints $\Xset_i$ and $\Uset_i$, $i\in\MM$.
          \begin{assum}
            \label{ass:shapesets}
            Constraints $\Xset_i$ and $\Uset_i$ are compact and convex polytopes containing the origin in their nonempty interior.
            \begin{flushright}$\blacksquare$\end{flushright}
          \end{assum}

          As in \cite{Rakovic2005b} our goal is to relate inputs $\subss v i$ in \eqref{eq:nominalsubsystem} to $\subss u i$ in \eqref{eq:subsystem} and compute sets $\Zset_i\subseteq\Rset^n$, $i\in\MM$ such that
          \begin{equation*}
            \subss x i(0)\in\subss\hx i(0)\oplus\Zset_i\Rightarrow\subss x i(t)\in\subss\hx i(t)\oplus\Zset_i,~\forall t\geq 0.
          \end{equation*}
          In other words, we want to confine $\subss x i(t)$ in a tube around $\subss\hx i(t)$ of section $\Zset_i$.

          To achieve our aim, we define the set $\Zset_i$, $\forall i\in\MM$ as an RCI set for the constrained system \eqref{eq:subsystem}, with respect to the disturbance $w_i\in\Wset_i=\bigoplus_{j\in\NN_i}A_{ij}\Xset_j$. From the definition of RCI set, we have that if $\subss x i\in\Zset_i$, then there exist $\subss u i=\bkappa_i(\subss x i):\Zset_i\rightarrow\Uset_i$ such that $\subss\px i\in\Zset_i$, $\forall\subss w i\in\Wset_i$. Note that, by construction, one has $\Zset_i\subseteq\Xset_i$ and therefore the RCI set $\Zset_i$ could not exist if $\Wset_i$ is ``too big'', i.e. $\Wset_i\supseteq\Xset_i$. Moreover if $\subss x i\in\subss\hx i\oplus\Zset_i$ and one uses the controller
          \begin{equation}
            \label{eq:tubecontrol}
            \subss\CC i:\quad\subss u i=\subss v i+\bkappa_i(\subss x i-\subss\hx i)
          \end{equation}
          then for all $\subss v i\in\Rset^{m_i}$, one has $\subss\px i\in\subss\hpx i\oplus\Zset_i$.

          The next goal is to compute tightened constraints $\hXset_i\subseteq\Xset_i$ and input constraints $\Vset_i\subseteq\Uset_i$ guaranteeing that
          \begin{align}
            \label{eq:cnstr_satisfied}
            &\subss \hx i (t)\in\hXset_i,~\subss v i (t)\in\Vset_i,~\forall i\in\MM\\
            &\quad \Rightarrow \mbf{x}(t +1)\in \Xset,~\mbf{u}(t)\in\Uset.\nonumber
          \end{align}
          To this purpose, we introduce the following assumption.
          \begin{assum}
            \label{ass:constr_satisf}
            There exist $\rho_{i,1}>0$, $\rho_{i,2}>0$ such that $\Zset_i\oplus\ball{\rho_{i,1}}(0)\subseteq\Xset_i$ and $\Uset_{z_i}\oplus\ball{\rho_{i,2}}(0)\subseteq\Uset_i$, where $\ball{\rho_{i,1}}(0) \subset\Rset^{n_i}$ and $\ball{\rho_{i,2}}(0)\subset\Rset^{m_i}$ and $\Uset_{z_i}=\bkappa_i(\Zset_i)$.
            \begin{flushright}$\blacksquare$\end{flushright}
          \end{assum}
          Assumption~\ref{ass:constr_satisf} implies that the coupling of subsystems connected in a cyclic fashion must be sufficiently small. As an example, for two subsystems $\Sigma_1$ and $\Sigma_2$ where each one is parent of the other one, Assumption~\ref{ass:constr_satisf} implies that $\Zset_1\subseteq\Xset_1$ and $\Zset_1\subseteq\Xset_1$. Since, by construction, $\Zset_i\supseteq\Wset_i$, one has $A_{21}\Xset_1\subseteq \Xset_2$ and $A_{12}\Xset_2\subseteq\Xset_1$ that implies $A_{12}A_{21}\Xset_1\subseteq\Xset_1$. These conditions are similar to the ones arising in the small gain theorem for networks \cite{Dashkovskiy}.

          If Assumption \ref{ass:constr_satisf} is verified, there are constraint sets $\hXset_i$ and $\Vset_i$, $i\in\MM$, that verify
          \begin{equation}
            \label{eq:stateinclu}
            \hXset_i\oplus\Zset_i\subseteq\Xset_i
          \end{equation}
          \begin{equation}
            \label{eq:inputinclu}
            \Vset_i\oplus\Uset_{z_i}\subseteq\Uset_i.
          \end{equation}

          Under Assumptions \ref{ass:controllable}-\ref{ass:constr_satisf}, we set in \eqref{eq:tubecontrol}
          \begin{equation}
            \label{eq:def_kappa_eta}
            \subss v i(t) = \kappa_i(\subss x i(t))=\subss v i(0|t),\qquad\subss\hx i (t)=\eta_i(\subss x i(t))=\subss\hx i(0|t)
          \end{equation}
          where $\subss v i(0|t)$ and $\subss\hx i(0|t)$ are optimal values of variables $\subss v i(0)$ and $\subss\hx i(0)$, respectively, appearing in the following MPC-$i$ problem, to be solved at time $t$
          \begin{subequations}
            \label{eq:decMPCProblem}
            \begin{align}
              &\label{eq:costMPCProblem}\Pset_i^N(\subss x i(t)) = \min_{\substack{\subss\hx i(0)\\\subss v i(0:N_i-1)}}\sum_{k=0}^{N_i-1}\ell_i(\subss\hx i(k),\subss v i(k))+V_{f_i}(\subss\hx i(N_i))&\\
              &\label{eq:inZproblem}\subss x i(t)-\subss \hx i(0)\in\Zset_i&\\
              &\label{eq:dynproblem}\subss \hx i(k+1)=A_{ii}\subss \hx i(k)+B_i\subss v i(k)&k\in0:N_i-1 \\
              &\label{eq:inhXproblem}\subss \hx i(k)\in\hXset_i&k\in0:N_i-1 \\
              &\label{eq:inVproblem}\subss v i(k)\in\Vset_i&k\in0:N_i-1 \\
              &\label{eq:inTerminalSet}\subss \hx i(N_i)\in\hXset_{f_i}&
            \end{align}
          \end{subequations}
          In \eqref{eq:decMPCProblem}, $N_i\in\Nset$ is the prediction horizon and $\ell_i(\subss\hx i(k),\subss v i(k)):\Rset^{n_i\times m_i}\rightarrow\Rset_+$ is the stage cost. Moreover $V_{f_i}(\subss\hx i(N_i)):\Rset^{n_i}\rightarrow\Rset_+$ is the final cost and $\hXset_{f_i}$ is the terminal set fulfilling the following assumption.
          \begin{assum}
            \label{ass:axiomsMPC}
            For all $i\in\MM$, there exist an auxiliary control law $\kappa_i^{aux}(\subss\hx i)$ and a $\KK_\infty$ function $\BB_i$ such that:
            \begin{enumerate}[(i)]
            \item\label{enu:boundstagecost}$\ell_i(\subss \hx i,\subss v i)\geq \BB_i(\norme{(\subss \hx i,\subss v i)}{})$, for all $\subss \hx i\in\Rset^{n_i}$,  $\subss v i\in\Rset^{m_i}$ and $\ell_i(0,0)=0$;
            \item\label{enu:invariantAux}$\hXset_{f_i}\subseteq\hXset_i$ is an invariant set for $\subss\hpx i=A_{ii}\subss\hx i+B_i\kappa_i^{aux}(\subss\hx i)$;
            \item\label{enu:terminalSetAux}$\forall \subss \hx i\in\hXset_{f_i}$, $\kappa_i^{aux}(\subss\hx i)\in\Vset_i$;
            \item\label{enu:decterminal}$\forall \subss \hx i\in\hXset_{f_i}$, $V_{f_i}(\subss\hpx i)-V_{f_i}(\subss\hx i)\leq-\ell_i(\subss\hx i,\kappa_i^{aux}(\subss\hx i))$.
            \end{enumerate}
            \begin{flushright}$\blacksquare$\end{flushright}
          \end{assum}
          We highlight that there are several methods, discussed e.g. in \cite{Rawlings2009}, for computing $\ell_i(\cdot)$, $V_{f_i}(\cdot)$ and $\Xset_{f_i}$ verifying Assumption \ref{ass:axiomsMPC}. In summary, the controller $\subss\CC i$ is given by \eqref{eq:tubecontrol}, \eqref{eq:def_kappa_eta} and \eqref{eq:decMPCProblem} and it is completely decentralized since it depends upon quantities of system $\subss\Sigma i$ only.

          The main problem that still has to be solved is the following one.
          \paragraph{Problem $\PP$:} Compute RCIs $\Zset_i$, $i\in\MM_i$ for \eqref{eq:subsystem}, if they exist, verifying Assumption \ref{ass:constr_satisf}.
               \begin{flushright}$\blacksquare$\end{flushright}

          In the next section, we show how to solve it in a distributed fashion with efficient computations under Assumption \ref{ass:controllable} and \ref{ass:shapesets}. In this case, we will also show how sets $\hXset_i$ and $\Vset_i$ verifying \eqref{eq:stateinclu} and \eqref{eq:inputinclu} can be readily computed.

          \begin{rmk}
            \label{rmk:oldvsnew1}
            In this section, as in \cite{Riverso2012a,Riverso2013c}, we introduced a DeMPC scheme based on tube-based control. In \cite{Riverso2012a,Riverso2013c}, using the robust control scheme proposed in \cite{Mayne2005}, we set $\bkappa_i(\cdot)$ as a linear function, i.e. $\bkappa_i(\subss x i-\subss\hx i)=K_i(\subss x i-\subss\hx i)$. This choice has the disadvantage of requiring the computation of matrices $K_i$, $i\in\MM$, fulfilling a global stability assumption. Differently, in the next section, using the control scheme proposed in \cite{Rakovic2005b}, we will guarantee the overall stability for the closed-loop system through a suitable local computation of sets $\Zset_i$.
          \end{rmk}

     \section{Decentralized synthesis of DeMPC}
          \label{sec:main}
          From Assumption \ref{ass:shapesets} constraints $\Xset_i$ and $\Uset_i$ are polytopes given by
          \begin{equation}
            \label{eq:setsXpoly}
            \Xset_i=\{\subss x i\in\Rset^{n_i}: c_{x_{i,r}}^T\subss x i\leq d_{x_{i,r}},~\forall r\in 1:g_i \}
          \end{equation}
          \begin{equation}
            \label{eq:setsUpoly}
            \Uset_i=\{\subss u i\in\Rset^{m_i}: c_{u_{i,r}}^T\subss u i\leq d_{u_{i,r}},~\forall r\in 1:l_i \},
          \end{equation}
          where $c_{x_{i,r}}\in\Rset^{n_i}$, $d_{x_{i,r}}\in\Rset$, $c_{u_{i,r}}\in\Rset^{m_i}$ and $d_{u_{i,r}}\in\Rset$.
          Using the procedure proposed in \cite{Rakovic2010} we can compute an RCI set $\Zset_i\subset\Xset_i$ using an appropriate parametrization. As in Section VI of \cite{Rakovic2010}, we define $\forall i\in\MM$ the set of variables $\theta_i$ as
          \begin{subequations}
            \label{eq:thetai}
            \begin{align}
              \theta_i =\{&\subss\bz i^{(s,f)}\in\Rset^{n_i}&\forall s\in 1:k_i,~\forall f\in 1:q_i;\\
                                &\subss\bu i^{(s,f)}\in\Rset^{m_i}&\forall s\in 0:k_i-1,~\forall f\in 1:q_i;\\
                                &\rho_i^{(f_1,f_2)}\in\Rset&\forall f_1\in 1:q_i,~\forall f_2\in 1:q_i;\\
                                &\psi_i^{(r,s)}\in\Rset&\forall r\in 1:l_i,~\forall s\in 0:k_i-1;\\
                                &\gamma_i^{(r,s)}\in\Rset&\forall r\in 1:g_i,~\forall s\in 0:k_i-1;\\
                                &\alpha_i\in\Rset
                                \}
            \end{align}
          \end{subequations}
          where $k_i$, $q_i\in\Nset$ are parameters of the procedure that can be chosen by the user as well as the set
          \begin{equation}
            \label{eq:Zio}
            \bZset_i^0=\convh(\{\subss\bz i^{(0,f)}\in\Rset^{n_i},~\forall f\in 1:q_i\}),~\mbox{with}~\subss\bz i^{(0,1)}=0.
          \end{equation}
          Let us define the sets
          \begin{equation}
            \label{eq:Zis}
            \bZset_i^s=\convh(\{\subss\bz i^{(s,f)}\in\Rset^{n_i},~\forall f\in 1:q_i\}),~\forall s\in 1:k_i,~\mbox{with}~\subss\bz i^{(s,1)}=0,
          \end{equation}
          and
          \begin{equation}
            \label{eq:Uzis}
            \bUset_{z_i}^s=\convh(\{\subss\bu i^{(s,f)}\in\Rset^{m_i},~\forall f\in 1:q_i\}),~\forall s\in 0:k_i-1,~\mbox{with}~\subss\bu i^{(s,1)}=0.
          \end{equation}
          and consider the following set of affine constraints on the decision variable $\theta_i$
          \begin{subequations}
            \label{eq:Thetai}
            \begin{align}
              \Theta_i=\{\theta_i:~&\alpha_i<1;~-\alpha_i\leq 0;\\
                                               &\label{eq:SkS01}\subss z i^{(k_i,f_1)}-\sum_{f_2=1}^{q_i}\rho_i^{(f_1,f_2)}\subss z i^{(0,f_2)}=0&\forall f_1\in 1:q_i;\\
                                               &\label{eq:SkS02}-\alpha_i+\sum_{f_2=1}^{q_i}\rho_i^{(f_1,f_2)}\leq 0&\forall f_1\in1:q_i;\\
                                               &\label{eq:SkS03}-\rho_i^{(f_1,f_2)}\leq 0&\forall f_1\in 1:q_i,~\forall f_2\in 1:q_i;\\
                                               &\sum_{s=0}^{k_i-1}\psi_i^{(r,s)}+d_{u_{i,r}}\alpha_i< d_{u_{i,r}}&\forall r\in1:l_i\\
                                               &c_{u_{i,r}}^T\subss\bu i^{(s,f)}-\psi_i^{(r,s)}\leq 0&\forall r\in1:l_i,\forall s\in0:k_i-1,\forall f\in1:q_i\\
                                               &\sum_{s=0}^{k_i-1}\gamma_i^{(r,s)}+d_{x_{i,r}}\alpha_i< d_{x_{i,r}}&\forall r\in1:g_i\\
                                               &c_{x_{i,r}}^T\subss\bz i^{(s,f)}-\gamma_i^{(r,s)}\leq 0&\forall r\in1:g_i,\forall s\in0:k_i-1,\forall f\in1:q_i\\
                                               &\subss\bz i^{(s+1,f)}-A_{ii}\subss\bz i^{(s,f)}-B_i\subss\bu i^{(s,f)}=0&\forall s\in0:k_i-1,\forall f\in1:q_i~
                                               \}.
            \end{align}
          \end{subequations}
          The following assumption is needed to compute the RCI set $\Zset_i$.
          \begin{assum}
            \label{ass:Z0} The set $\bZset_i^0$ is such that there is $\omega_i>0$ verifying $\Wset_i\oplus\ball{\omega_i}(0)\subseteq\bZset_i^0$.
            \begin{flushright}$\blacksquare$\end{flushright}
          \end{assum}
          We highlight that, in view of Assumption \ref{ass:shapesets}, the set $\Wset_i$ contains the origin in its nonempty relative interior. Hence, under Assumption \ref{ass:Z0}, the set $\bZset_i^0$ also contains the origin in its nonempty interior.\\
          The relation between elements of $\Theta_i$ and the RCI sets in Problem $\PP$ is established in the next proposition.
          \begin{prop}
            \label{prop:feasytheta}
            Let Assumptions \ref{ass:controllable} and \ref{ass:Z0} hold and sets $\Xset_i$ and $\Uset_i$ be defined as in \eqref{eq:setsXpoly} and \eqref{eq:setsUpoly} respectively. Let $k_i\geq\CC\II(A_{ii},B_i)$. If there exist an admissible $\theta_i\in\Theta_i$, then
            \begin{equation}
              \label{eq:Zrcilambda}
              \Zset_i=(1-\alpha_i)^{-1}\bigoplus_{s=0}^{k_i-1}\bZset_i^s\subset\Xset_i\\
            \end{equation}
            is an RCI set and the corresponding set $\Uset_{z_i}$ is given by
            \begin{equation}
              \label{eq:Uzrcilambda}
              \Uset_{z_i}=(1-\alpha_i)^{-1}\bigoplus_{s=0}^{k_i-1}\bUset_{z_i}^s\subset\Uset_i.\\
            \end{equation}
            \begin{flushright}$\blacksquare$\end{flushright}
          \end{prop}
          \begin{proof}
            In Section 3 of \cite{Rakovic2010}, the authors prove that set $\Zset_i$ defined as in \eqref{eq:Zrcilambda} is an RCI set and that the inclusions in \eqref{eq:stateinclu} and \eqref{eq:inputinclu} hold.
          \end{proof}
          \begin{rmk}
            Under Assumption \ref{ass:shapesets} the feasibility problem \eqref{eq:Thetai} is an LP problem, since the constraints in $\Theta_i$ are affine. In \cite{Rakovic2010} the authors propose to find $\theta\in\Theta_i$ while minimizing different cost functions under constraints $\Theta_i$ in order to achieve different aims. In our context the most important goal is the minimization of $\alpha_i$ that corresponds to the minimization of the size of the set $\Zset_i$. We also note that the inclusion of $0$ in the definition of sets $\bZset_i^s$, $\forall s\in 0:k_i$, ensures that $\bZset_i^s$ contains the origin and hence, under Assumption \ref{ass:Z0}, $\Zset_i$ contains the origin in its nonempty interior.
          \end{rmk}

          We highlight that the set of constraints $\Theta_i$ depends only upon local fixed parameters $\{A_{ii},B_{i},\Xset_i,\Uset_i\}$, fixed parameters $\{A_{ij},\Xset_j\}_{j\in\NN_i}$ of predecessors of $\subss{\hat\Sigma} i$ (because from Assumption \ref{ass:Z0} the set $\bZset_i^0 $ must be defined in such a way that $\bZset_i^0\supseteq\Wset_i=\bigoplus_{j\in\NN_i}A_{ij}\Xset_j$) and local tunable parameters $\theta_i$ (the decision variables \eqref{eq:thetai}). Moreover $\Theta_i$ does not depend on tunable parameters of predecessors. This implies that the computation of sets $\Zset_i$ and $\Uset_{z_i}$ in \eqref{eq:Zrcilambda} and \eqref{eq:Uzrcilambda} does not influence the choice of $\Zset_j$ and $\Uset_{z_j}$, $j\neq i$ and therefore Problem $\PP$ is decomposed in the following independent LP problems for $i\in\MM$.
          \paragraph{Problem $\PP_i$:} Solve the feasibility LP problem $\theta_i\in\Theta_i$.
               \begin{flushright}$\blacksquare$\end{flushright}

          If Problem $\PP_i$ is solved, then $\forall i\in\MM$ we can compute sets $\hXset_i$ and $\Vset_i$ in \eqref{eq:inhXproblem} and \eqref{eq:inVproblem} as
          \begin{equation}
            \label{eq:tightset}
            \hXset_i=\Xset_i\circleddash\Zset_i,\qquad\Vset_i=\Uset_i\circleddash\Uset_{z_i}.
          \end{equation}
          The overall procedure for the decentralized synthesis of local controllers $\subss\CC i$, $i\in\MM$ is summarized in Algorithm \ref{alg:distrisynt}.
          \begin{prop}
            \label{prop:asshold}
            Under Assumption \ref{ass:controllable} and \ref{ass:shapesets} if, for all $i\in\MM$, controllers $\subss\CC i$ are designed according to Algorithm \ref{alg:distrisynt}, then also Assumptions \ref{ass:constr_satisf}, \ref{ass:axiomsMPC} and \ref{ass:Z0} are verified.
          \end{prop}

          \begin{algorithm}
            \caption{Design of controller $\subss\CC i$ for system $\subss \Sigma i$}
            \label{alg:distrisynt}
            \textbf{Input}: $A_{ii}$, $B_i$, $\Xset_i$, $\Uset_i$, $\NN_i$, $\{A_{ij}\}_{j\in\NN_i}$, $\{\Xset_{j}\}_{j\in\NN_i}$, $k_i\geq\CC\II(A_{ii},B_i)$\\
            \textbf{Output}: controller $\subss \CC i$ given by \eqref{eq:tubecontrol}, \eqref{eq:def_kappa_eta} and \eqref{eq:decMPCProblem}\\
            \begin{enumerate}[1)]
            \item Compute sets $\Wset_i=\bigoplus_{j\in\NN_i}A_{ij}\Xset_j$ and set $\bZset_i^0\supseteq\Wset_i\oplus\ball{\omega_i}(0)$ for a sufficiently small $\omega_i>0$ guaranteeing $\bZset_i^0\subset\Xset_i$.
            \item\label{enu:feasProb} Solve the feasibility LP problem $\theta_i\in\Theta_i$. If it is unfeasible \textbf{stop} (the controller $\subss\CC i$ cannot be designed).
            \item\label{enu:rciAlg} Compute $\Zset_i$ as in \eqref{eq:Zrcilambda} and $\Uset_{z_i}$ as in \eqref{eq:Uzrcilambda}.
            \item\label{enu:hXsetAlg} Compute $\hXset_i=\Xset_i\circleddash\Zset_i$.
            \item\label{enu:VsetAlg} Compute $\Vset_i=\Uset_i\circleddash\Uset_{z_i}$.
            \item\label{enu:VfXfAlg}Choose $\ell_i(\cdot)$, $V_{f_i}(\cdot)$ and $\Xset_{f_i}$ verifying Assumption \ref{ass:axiomsMPC}.
            \end{enumerate}
          \end{algorithm}

          If in Step \ref{enu:feasProb} of Algorithm \ref{alg:distrisynt} the LP problem is infeasible, we can restart the Algorithm with a different $k_i$. However the existence of a parameter $k_i$ such that the LP problem is feasible is not guaranteed \cite{Rakovic2010}. Steps \ref{enu:rciAlg}, \ref{enu:hXsetAlg} and \ref{enu:VsetAlg} of Algorithm \ref{alg:distrisynt}, that provide constraints appearing in the MPC-$i$ problem \eqref{eq:decMPCProblem}, are the most computationally expensive ones because they involve Minkowski sums and differences of polytopic sets. Next, we show how to avoid burdensome computations.

          \subsection{Implicit representation of sets $\Zset_i$ and $\Uset_{z_i}$}
               \label{sec:computZi}
               In this section we show how to rewrite constraint \eqref{eq:inZproblem} by exploiting the implicit representation of set $\Zset_i$ proposed in Section VI.B of \cite{Rakovic2010}. Recalling that $\Zset_i$ is the Minkowski sum of $k_i$ polytopes and that, for all $s\in 0:k_i-1$, polytope $\bZset_i^s$ is described by the convex combination of points ${\subss \bz i}^{(s,f)}$, we have
               \begin{align*}
                 {\subss \tz i}^s\in\bZset_i^s\quad\mbox{ if }\forall f\in1:q_i,\exists\beta_i^{(s,f)}\geq 0\quad\mbox{ such that }\sum_{f=1}^{q_i}\beta_i^{(s,f)}=1 , {\subss \tz i}^s=\sum_{f=1}^{q_i}\beta_i^{(s,f)}{\subss \bz i}^{(s,f)}.
               \end{align*}
               Hence we have that $\subss x i(t)-\subss\hx i(0|t)\in\Zset_i$ if and only if $\forall f\in1:q_i, \forall s\in 0:k_i-1$ there exist $\beta_i^{(s,f)}\in\Rset$ such that\\
               \begin{subequations}
                 \label{eq:implicitZi}
                 \begin{align}
                   &\label{eq:boundBetaimplicitZi}\beta_i^{(s,f)}\ge 0\\
                   &\label{eq:sumBetaimplicitZi}\sum_{f=1}^{q_i}\beta_i^{(s,f)}=1\\
                   &\label{eq:inZimplicitZi}\subss x i(t)-\subss\hx i(0|t)=(1-\alpha_i)^{-1}\sum_{s=0}^{k_i-1}\sum_{f=1}^{q_i}\beta_i^{(s,f)}\subss \bz i^{(s,f)}.
                \end{align}
               \end{subequations}
               In other words we add to the optimization problem \eqref{eq:decMPCProblem} the variables $\beta_i^{(s,f)}$ and replace \eqref{eq:inZproblem} with constraints \eqref{eq:boundBetaimplicitZi}-\eqref{eq:inZimplicitZi}.

               With similar arguments, we can also provide an implicit representation of sets $\Uset_{z_i}$. In particular, we have that $\subss{u_z}{i}\in\Uset_{z_i}$ if and only if $\forall f\in1:q_i, \forall s\in 0:k_i-1$ there exist $\phi_i^{(s,f)}\in\Rset$ such that\\
               \begin{subequations}
                 \label{eq:implicitUzi}
                 \begin{align}
                   &\label{eq:boundBetaimplicitUzi}\phi_i^{(s,f)}\ge 0\\
                   &\label{eq:sumBetaimplicitUzi}\sum_{f=1}^{q_i}\phi_i^{(s,f)}=1\\
                   &\label{eq:inZimplicitUzi}\subss{u_z}{i}=(1-\alpha_i)^{-1}\sum_{s=0}^{k_i-1}\sum_{f=1}^{q_i}\phi_i^{(s,f)}\subss \bu i^{(s,f)}.
                \end{align}
               \end{subequations}

          \subsection{Computation of sets $\hXset_i$ and $\Vset_i$}
               In this section we show how to compute sets $\hXset_i$ and $\Vset_i$ in \eqref{eq:tightset} using the implicit representation of $\Zset_i$ and $\Uset_{z_i}$.

               Using \eqref{eq:Zrcilambda} we can rewrite $\hXset_i=\Xset_i\circleddash (1-\alpha_i)^{-1}\bigoplus_{s=0}^{k_i-1}\bZset_i^s$. Recalling that $\bZset_i^s, \forall s\in 0:k_i-1$ are defined as the convex hull of points $\subss \bz i^{(s,f)}$, $f\in1:q_i$, we can compute the set $\hXset_i$ using Algorithm \ref{alg:hXset}.
               \begin{algorithm}
                 \caption{}
                 \label{alg:hXset}
                 \textbf{Input}: set $\Xset_i$ defined as in \eqref{eq:setsXpoly}, points $\subss \bz i^{(s,f)}, \forall s\in 0:k_i-1, \forall f\in1:q_i$ and scalar $\alpha_i$.\\
                 \textbf{Output}: set $\hXset_i$.\\
                 \begin{enumerate}[(I)]
                 \item $\bar C_i=(c_{x_{i,1}}^T,\ldots,c_{x_{i,g_i}}^T)\in\Rset^{{g_i}\times{n_i}}$ and $\bar D_i = (d_{x_{i,1}},\ldots,d_{x_{i,g_i}})\in\Rset^{g_i}$
                 \item \textbf{For each $s\in 0:k_i-1$}
                   \begin{enumerate}[(i)]
                   \item \textbf{For each $f\in 1:q_i$}\\
                     $\tilde C_i = (\bar C_i,\bar C_i)$ and $\tilde D_i = ( \bar D_i,\bar D_i-(1-\alpha_i)^{-1}\bar C_i\subss \bz i^{(s,f)} )$
                   \item\label{enu:setsXhateffic} Remove redundant constraints from $\tilde C_i\subss\hx i\leq\tilde D_i$ so obtaining the inequalities $\bar C_i\subss\hx i\leq\bar D_i$
                   \end{enumerate}
                 \item Set $\hXset_i=\{\subss\hx i:\bar C_i\subss\hx i\leq\bar D_i\}$ where $\bar C_i\in\Rset^{\hat g_i\times n_i}$ and $\bar D_i\in\Rset^{\hat g_i}$
                 \end{enumerate}
               \end{algorithm}

               In particular, the operation in Step (\ref{enu:setsXhateffic}) amounts to solve suitable LP problems. We can compute $\Vset_i$ using the implicit representation of $\Uset_{z_i}$ in a similar way. Indeed it suffices to use Algorithm \ref{alg:hXset} replacing $\Xset_i$ with $\Uset_i$ defined in \eqref{eq:setsUpoly} and points $\subss \bz i^{(s,f)}$ with points $\subss \bu i^{(s,f)}, \forall s\in 0:k_i-1, \forall f\in1:q_i$.

          \subsection{Computation of control law $\bkappa_i(\cdot)$}
               In Section \ref{sec:modelcontroller}, we introduced local controllers $\subss\CC i$. Note that in \eqref{eq:tubecontrol} the control law $\subss u i$ is composed by the term $\subss v i$, that is computed by solving the local MPC-$i$ problem \eqref{eq:decMPCProblem}, and the term $\bkappa(\subss z i)$ with $\subss z i=\subss x i-\subss\hx i$. Since $\bkappa_i(\cdot)$ depends on $\subss\hx i$, we need to solve the MPC-$i$ problem \eqref{eq:decMPCProblem} and then compute $\bkappa_i(\subss z i)$. The control law $\bkappa(\subss z i)\in\Uset_{z_i}$ guarantees that if $\subss x i(t)-\subss\hx i(0|t)\in\Zset_i$ (i.e. MPC-$i$ problem \eqref{eq:decMPCProblem} is feasible) then there is a $\lambda_i>0$ such that $\subss x i(t+1)-\subss\hx i(1|t)\in\lambda_i\Zset_i$. To compute the control law $\bkappa_i(\subss z i)$ one can use the methods proposed in \cite{Blanchini1991} or in \cite{Rakovic2010}. In \cite{Blanchini1991} the authors propose to solve an LP problem in order to maximize the contractivity parameter $\lambda_i$, i.e. for a given $\subss z i$ we compute $\bkappa_i(\subss z i)\in\Uset_{z_i}$ by minimizing the scalar $\lambda_i$ such that $A_{ii}\subss z i+B_i\bkappa_i(\subss z i)\in\lambda\Zset_i\ominus\Wset_i$. In \cite{Rakovic2010} the authors propose an implicit representation of controller $\bkappa_i(\subss z i)$ based on the implicit representation \eqref{eq:implicitZi} of set $\Zset_i$. In our framework we want to take advantages of both approaches and compute the control law $\bkappa_i(\cdot)$ solving the following LP problem
               \begin{subequations}
                 \label{eq:kappainv}
                 \begin{align}
                   \bar\Pset_i(\subss z i):&\min_{\substack{  \mu , \beta_i^{(s,f)} }}\mu\\
                   \label{eq:kappainvBeta}&\beta_i^{(s,f)}\ge 0&\forall f\in1:q_i, \forall s\in 0:k_i-1\\
                   \label{eq:kappainvSum1}&\sum_{f=1}^{q_i}\beta_i^{(s,f)}=\mu&\forall s\in 0:k_i-1\\
                   \label{eq:kappainvmu}&\mu\ge 0\\
                   \label{eq:pzInvproblem}&\subss z i=(1-\alpha_i)^{-1}\sum_{s=0}^{k_i-1}\sum_{f=1}^{q_i}\beta_i^{(s,f)}\subss \bz i^{(s,f)}
                 \end{align}
               \end{subequations}
               and setting
               \begin{equation}
                 \label{eq:kappainvcomputed}
                 \bkappa_i(\subss z i)=(1-\alpha_i)^{-1}\sum_{s=0}^{k_i-1}\bar\kappa^s_i(\subss z i),~\bar\kappa^s_i(\subss z i)=\sum_{f=1}^{q_i}\bar\beta_i^{(s,f)}\subss \bu i^{(s,f)}
               \end{equation}
               where $\bar\beta_i^{(s,f)}$ are the optimizers to \eqref{eq:kappainv}. Solving LP problem \eqref{eq:kappainv}, we compute a control law $\bkappa_i(\cdot)$ that tries to keep the state $\subss x i$ and the nominal state $\subss \hx i$ as close as possible. According to \cite{Gal1995} we can assume without loss of generality that $\bkappa_i(\cdot)$ is a continuous piecewise affine map. Note that since $\bZset_i^0\subseteq\Zset_i$, if $\subss z i=0$ no control action is needed in order to guarantee robust invariance. Indeed, in this case an optimal solution to \eqref{eq:kappainv} is $\mu=0$ and $\beta_i^{(s,f)}=0$, $\forall f\in1:q_i, \forall s\in 0:k_i-1$ and therefore $\bkappa_i(\subss z i)=0$.

          We are now in a position to analyze the stability properties of the closed-loop system. Defining the collective variables $\mbf \hx = (\subss \hx 1,\ldots,\subss \hx M)\in\Rset^n$, $\mbf v = (\subss v 1,\ldots,\subss v M)\in\Rset^m$ and the function $\mbf\bkappa(x) = (\bkappa_1(\subss x 1),\ldots,\bkappa_M(\subss x M)):\Rset^n\rightarrow\Rset^m$, from \eqref{eq:subsystem} and \eqref{eq:tubecontrol} one obtains the collective model
          \begin{equation}
            \label{eq:controlled_model}
            \mbf \px = \mbf{Ax}+\mbf{Bv}+\mbf{B\bkappa(x-\hx)}.
          \end{equation}
          \begin{defi}
            The feasibility region for the MPC-$i$ problem is
            $$
            \Xset_i^N=\{\subss s i\in\Xset_i:~\mbox{~\eqref{eq:decMPCProblem} is feasible for}~\subss x i(t)=\subss s i\}
            $$
            and the collective feasibility region is $\Xset^N=\prod_{i\in\MM}\Xset_i^N$.
          \end{defi}
          \begin{thm}
            \label{thm:mainclosedloop}
            Let Assumptions \ref{ass:controllable} and \ref{ass:shapesets} hold and assume controllers $\subss\CC i$ in \eqref{eq:tubecontrol} are computed using Algorithm \ref{alg:distrisynt}. Then the origin of the closed-loop system \eqref{eq:controlled_model} is asymptotically stable, $\Xset^N$ is a region of attraction and $\mbf x (0)\in\Xset^N$ guarantees constraints \eqref{eq:constraints} are fulfilled at all time instants.
            \begin{flushright}$\blacksquare$\end{flushright}
          \end{thm}
          The proof of Theorem \ref{thm:mainclosedloop} is provided in Appendix \ref{sec:proofthmmain}.

          \begin{rmk}
            \label{rmk:oldvsnew2}
            In Remark \ref{rmk:oldvsnew1}, we highlighted that the main difference with the PnP scheme proposed in \cite{Riverso2012a,Riverso2013c} is the computation of sets $\Zset_i$ and functions $\bar\kappa_i(\cdot)$, $\forall i\in\MM$. We note that in \cite{Riverso2012a,Riverso2013c}, the computation of $K_i$ and $\Zset_i$ requires the solution to a nonlinear optimization problem. In this section, we have shown that for the PnP scheme proposed in Section \ref{sec:modelcontroller}, using results from \cite{Rakovic2010}, we can compute set $\Zset_i$ and function $\bar\kappa_i(\cdot)$ solving LP problems only.
          \end{rmk}

     \section{Plug-and-play operations}
          \label{sec:plugplay}
          In this Section we discuss the synthesis of new controllers and the redesign of existing ones when subsystems are added to or removed from system \eqref{eq:subsystem}. The goal will be to preserve stability of the origin and constraint satisfaction for the new closed-loop system. Note that plugging in and unplugging of subsystems are here considered as off-line operations. As a starting point, we consider a plant composed by subsystems $\subss \Sigma i$, $i\in\MM$ equipped with local controllers $\subss \CC i$, $i\in\MM$ produced by Algorithm \ref{alg:distrisynt}.

          \subsection{Plugging in operation}
               \label{sec:plugin}
               We start considering the plugging in of subsystem $\subss{\Sigma}{M+1}$, characterized by parameters $A_{M+1~M+1}$, $B_{M+1}$, $\Xset_{M+1}$, $\Uset_{M+1}$, $\NN_{M+1}$ and $\{A_{ij}\}_{j\in\NN_{M+1}}$. In particular, $\NN_{M+1}$ identifies the subsystems that will be physically coupled to $\subss{\Sigma}{M+1}$. For building the controller $\subss{\CC}{M+1}$ we execute Algorithm \ref{alg:distrisynt} that needs information only from systems $\subss{\Sigma}{j}$, $j\in\NN_{M+1}$. If there is no solution to the feasibility LP problem in Step \ref{enu:feasProb} of Algorithm \ref{alg:distrisynt}, we declare that $\subss{\Sigma}{M+1}$ cannot be plugged in. Let $\SSS_{M+1}=\{j:M+1\in\NN_j\}$ be the set of successors to system $M+1$. Since each system $\subss{\Sigma}{j}$, $j\in\SSS_{M+1}$ has the new predecessor $\subss{\Sigma}{M+1}$, we have that the set $\bZset_j^0$ already computed verifies $\bZset_j^0\subset\Wset_j$ and hence not all assumptions of Proposition \ref{prop:feasytheta} could be satisfied. Therefore, when $\NN_j$ gets larger, for each $j\in\SSS_{M+1}$ the controllers $\subss\CC j$ must be redesigned according to Algorithm \ref{alg:distrisynt}. Again, if Algorithm \ref{alg:distrisynt} stops in Step \ref{enu:feasProb} for some $j\in\SSS_{M+1}$, we declare that $\subss{\Sigma}{M+1}$ cannot be plugged in.\\
               Note that redesign of controllers that are farther in the network is not needed, i.e. even without changing controllers $\subss\CC i$, $i\notin\{M+1\}\bigcup\SSS_{M+1}$ convergence to zero of the origin and constraint satisfaction are guaranteed for the new closed-loop system.

          \subsection{Unplugging operation}
               \label{sec:unplug}
               We consider the unplugging of  system $\subss{\Sigma}{k}$, $k\in\MM$ and define the set $\SSS_{k}=\{k:i\in\NN_k\}$ of successors to system $k$. Since for each $i\in\SSS_k$ the set $\NN_i$ gets smaller, we have that the set $\bZset_i^0$ already computed verifies $\bZset_i^0\supset\Wset_i$ and hence assumptions of Proposition \ref{prop:feasytheta} are still satisfied. This means that for each $i\in\SSS_k$ the LP problem in Step \ref{enu:feasProb} of Algorithm \ref{alg:distrisynt} is still feasible and hence the controller $\subss\CC i$ does not have to be redesigned. Moreover since for each system $\subss\Sigma j$, $j\notin \{k\}\bigcup\SSS_k$ the set $\NN_j$ does not change, the redesign of controller $\subss\CC j$ is not required.\\
               In conclusion, removal of system $\subss\Sigma k$ does not require the redesign of any controller, in order to guarantee convergence to zero of the origin and constraints satisfaction for the new closed-loop system. However, we highlight that since systems $\subss\Sigma i$, $i\in\SSS_k$ have one predecessor less, the redesign of controllers $\subss\CC i$ through Algorithm \ref{alg:distrisynt} could improve the performance.

     \section{Distributed on-line implementation of controllers $\subss\CC i$}
          \label{sec:distrcontr}
          In Section \ref{sec:modelcontroller}, we introduced decentralized local controllers $\subss\CC i$ that, using the nominal model \eqref{eq:nominalsubsystem} and local information only, can control system $i$ without the knowledge of the state of the predecessors. However, our framework allows one to take advantage of information from predecessors systems without redesigning controllers $\subss\CC i$.

          If at time $t$ the controller of system $\subss\Sigma i$ can receive the value of states $\subss x j(t), \forall j\in\NN_i$ from predecessors, we can define the new controller ${\subss\CC i}^{dis}$ as
          \begin{equation}
            \label{eq:controllerDis}
            {\subss\CC i}^{dis}:\quad\subss u i=\subss v i+\bkappa_i^{dis}(\subss x i-\subss\hx i,\{\subss x j\}_{j\in\NN_i}).
          \end{equation}
          In \eqref{eq:controllerDis} the term $\subss v i$ is the same appearing in the controller $\subss\CC i$ and is obtained by solving the MPC-$i$ problem \eqref{eq:decMPCProblem}. Similarly to the control law $\bkappa_i(\cdot)$ in \eqref{eq:tubecontrol}, the second term in \eqref{eq:controllerDis} must guarantee robust invariance of the set $\Zset_i$ and it can be computed by solving \eqref{eq:kappainv} with constraint \eqref{eq:pzInvproblem} replaced by
          \begin{equation}
            \label{eq:pzInvproblemDis}
            A_{ii}\subss z i+B_i \subss{u_z}{i}+\sum_{j\in\NN_i}A_{ij}\subss x j=(1-\alpha_i)^{-1}\sum_{s=0}^{k_i-1}\sum_{f=1}^{q_i}\beta_i^{(s,f)}\subss \bz i^{(s,f)}.
          \end{equation}
          where $\subss{u_z}i\in\Rset^{m_i}$ are additional optimization variables.
          The desired control term is then given by $\bkappa_i^{dis}(\subss x i-\subss\hx i,\{\subss x j\}_{j\in\NN_i})=\subss{u_z}{i}$. Note that constraint \eqref{eq:pzInvproblemDis} allows us to compute $\bkappa_i^{dis}(\cdot)$ taking into account the real state of predecessors at time $t$. Using \eqref{eq:inputinclu} and \eqref{eq:controllerDis}, we can still guarantee input constraints \eqref{eq:setsUpoly} adding the following constraints in the LP problem $\bar\Pset_i$ in \eqref{eq:kappainv}
          $$
          c_{u_{i,r}}^T\subss{u_z}{i}\leq d_{u_{i,r}}-c_{u_{i,r}}^T\subss v i,~\forall r\in 1:l_i.
          $$
          We highlight that the LP problem \eqref{eq:kappainv} is feasible if and only if the new LP problem is feasible. In fact, using the definition of robust control invariance, the LP problem \eqref{eq:kappainv} is feasible if there exist $\subss{u_z}{i}\in\Uset_{z_i}$ such that $\subss\pz i=A_{ii}\subss z i+B_i \subss{u_z}{i}+\subss w i\in\Zset_i,  \forall \subss w i\in\Wset_i$. The fact that $\sum_{j\in\NN_i}A_{ij}\subss x j\in\Wset_i$ guarantees the feasibility of both LP problems.

          We show advantages of including information from predecessors through an example. Consider two dynamically coupled systems equipped with controllers synthesized using Algorithm \ref{alg:distrisynt} and assume $\subss x 1(0)=0$ and $\subss x 2(0)\neq 0\notin\Zset_2$. Without exchange of information, the solution to the MPC-$i$ problem  \eqref{eq:decMPCProblem} is $\subss v 1(0)=0$ and $\subss\hx 1(0)=0$ for the first system and $\subss v 2(0)\neq0$ and $\subss\hx 2(0)\neq0$ for the second system, hence the solution of the LP problem \eqref{eq:kappainv} will be $\bkappa_1(\subss z 1)=0$ and $\bkappa_2(\subss z 2)\neq0$. This means we apply a control action to system $2$ only. However, $\subss x 1(1)\neq 0$ because of coupling. Differently, solving the LP problem \eqref{eq:kappainv} with constraint \eqref{eq:pzInvproblem} replaced by \eqref{eq:pzInvproblemDis}, we obtain $\bkappa_1(\subss z 1)\neq0$ and $\bkappa_2(\subss z 2)\neq0$. Therefore, we apply a control action on both systems because system $1$ tries to counteract in advance coupling with system $2$.

     \section{Examples}
          \label{sec:examples}
          In this section, we illustrate three examples.
          \begin{enumerate}
          \item A low-order system composed by the interconnection of two mass-spring-damper systems, allowing decentralized and distributed implementations of local controllers to be compared.
          \item The Power Network Systems (PNS) previously introduced in \cite{Riverso2012f} and \cite{Riverso2012a,Riverso2013c} where we compare the performance of the proposed controllers with centralized MPC and with the plug-and-play controllers proposed in \cite{Riverso2012a,Riverso2013c}. Furthermore, we discuss plug-and-play operations corresponding to the addition and removal of power generation areas;
          \item A large-scale system composed by an array of 1024 mass-spring-damper systems.
          \end{enumerate}
          All examples and simulations are implemented using the \emph{PnPMPC-toolbox} for Matlab \cite{Riverso2012g} dedicated to the modeling of large-scale systems and the design of plug-and-play controllers.

          \subsection{Comparison of Decentralized and Distributed controllers}
               In this section, we compare the performance of controllers $\subss\CC i$  and ${\subss\CC i}^{dis}$. We consider the example illustrated in Figure \ref{fig:exampleCart1D}.
               \begin{figure}[htb]
                 \centering
                 \def\svgwidth{250pt}
                 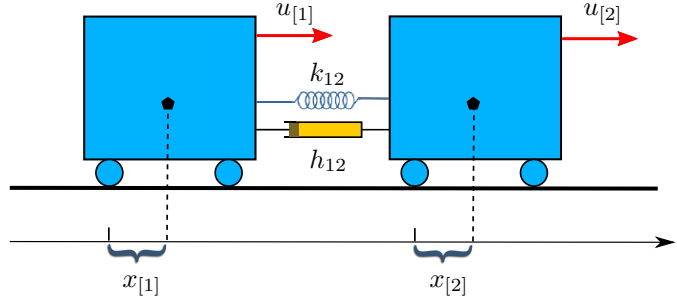
                 \caption{Example system.}
                 \label{fig:exampleCart1D}
               \end{figure}
               The system is composed by two trucks coupled by a spring and a damper. Parameters values are: $m_1=2$, $m_2=4$, $k_{12}=0.4$ and $h_{12}=0.3$. Each truck $i\in\MM=\{1,2\}$, is a subsystem with state variables $\subss x i=(\subss x {i,1},\subss x {i,2})$ and input $\subss u i$, where $\subss x {i,1}$ is the displacement of truck $i$ with respect to a given equilibrium position, $\subss x {i,2}$ is the velocity of the truck $i$ and $100\subss u i$ is a force applied to truck $i$. Subsystems are equipped with the state constraints $\abs{\subss x {i,1}}\leq 4.5$, $\abs{\subss x {i,2}}\leq 2$, $i\in\MM$ and with the input constraints $\abs{\subss u {i}}\leq 1.5$, $i\in\{1,2\}$. We obtain models $\subss\Sigma i$ by discretizing the second order continuous-time system representing each truck with $0.1~$sec sampling time, using exact discretization and treating $\subss u i$ and $\subss x j,~j\in\NN_i$ as exogenous signals. We synthesized controllers $\subss\CC i$, $i\in\MM$ using Algorithm \ref{alg:distrisynt}. At time $t$ we compute the control action $\subss u i$ and apply it to the continuous-time system, keeping the value constant between time $t$ and $t+1$. We assume $\subss x 1(0)=(0,0)$ and $\subss x 2(0)=(3,0)$.

               In Figure \ref{fig:resultsCart1D} we show the results obtained using controllers $\subss\CC i$ and $\subss\CC i^{dis}$ in the time interval from $0$ to $0.3~$sec. We note that for the controller $\subss\CC 1$, since $\subss x 1(0)=0$ one has $\subss u 1(0)=0$. Indeed the solutions to optimization problems \eqref{eq:decMPCProblem} and \eqref{eq:kappainv} are $\subss v 1(0)=0$ and $\kappa_1(\subss z 1(0))=0$. For the second truck the control action is $\subss u 2(0)=-0.76$ because $\subss x 2(0)\neq 0$. However, one has $\subss x 1(1)\neq 0$ because of coupling. Using the distributed controller ${\subss\CC 1}^{dis}$, since $\subss x 1(0)=0$ and $\subss x 2(0)\neq 0$, one has $\subss u 1(0)=-0.012$. Indeed the solution to the LP problem \eqref{eq:kappainv}, with \eqref{eq:pzInvproblem} replaced by \eqref{eq:pzInvproblemDis}, gives $\bkappa_1(\subss z 1)\neq0$. Figure \ref{fig:stateCart1D} shows the position of each truck: we note that using controllers ${\subss\CC i}^{dis}$, the position of the first truck does not change significantly because the controller tries to counteract in advance coupling with system $2$. This shows the benefits of a distributed implementation of local controllers. The state and input trajectories of the second truck are almost identical when using controllers $\subss\CC i$ and ${\subss\CC i}^{dis}$ because the state of the first truck is approximately zero.
               \begin{figure}[!ht]
                 \centering
                 \subfigure[\label{fig:stateCart1D}Displacement of truck $i$ controlled by $\subss\CC i$ (dashed line) and ${\subss\CC i}^{dis}$ (bold line).]{\includegraphics[scale=0.38]{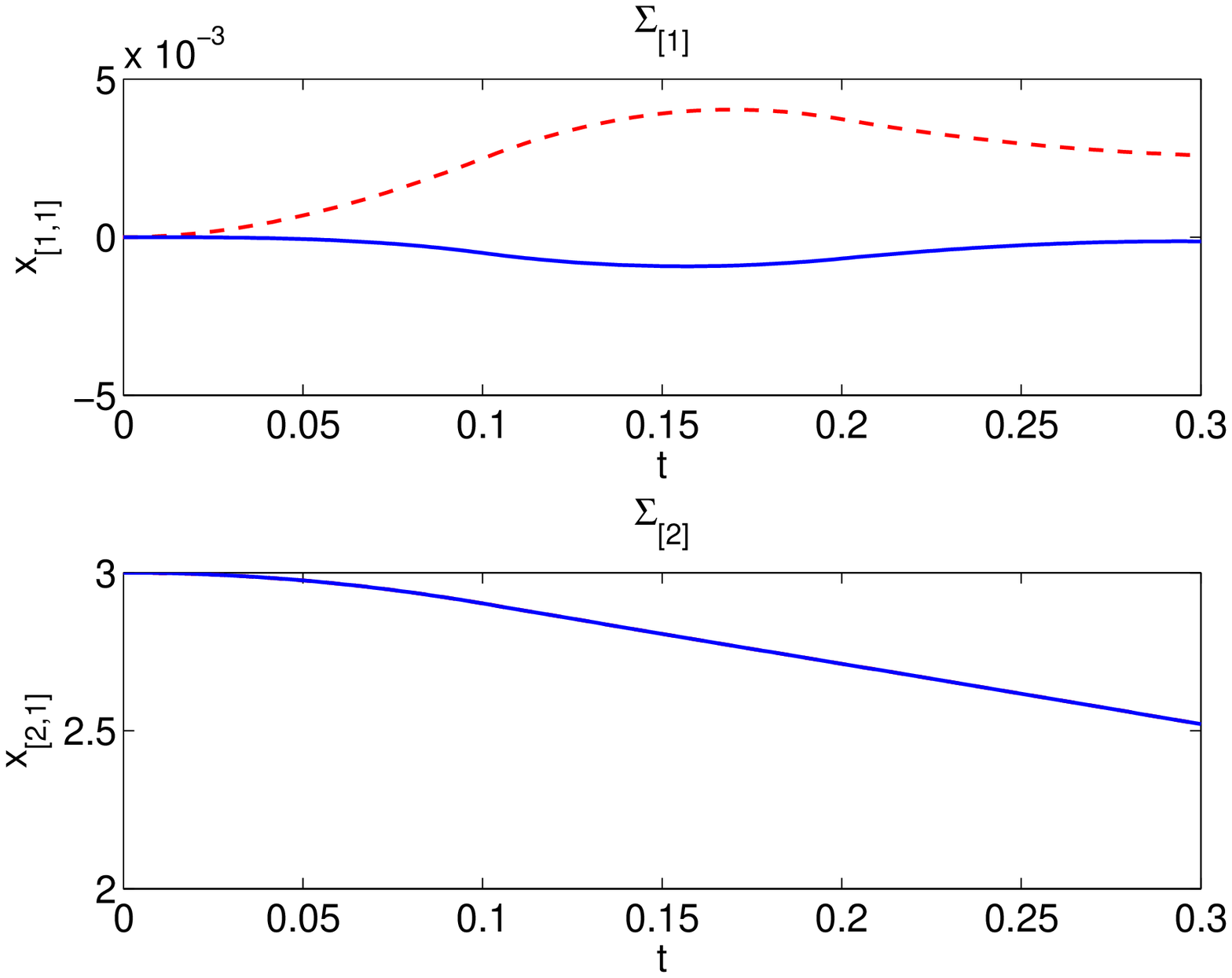}}~~~~~~
                 \subfigure[\label{fig:inputCart1D}Control law computed by using $\subss\CC i$ (dashed line) and ${\subss\CC i}^{dis}$ (bold line).]{\includegraphics[scale=0.38]{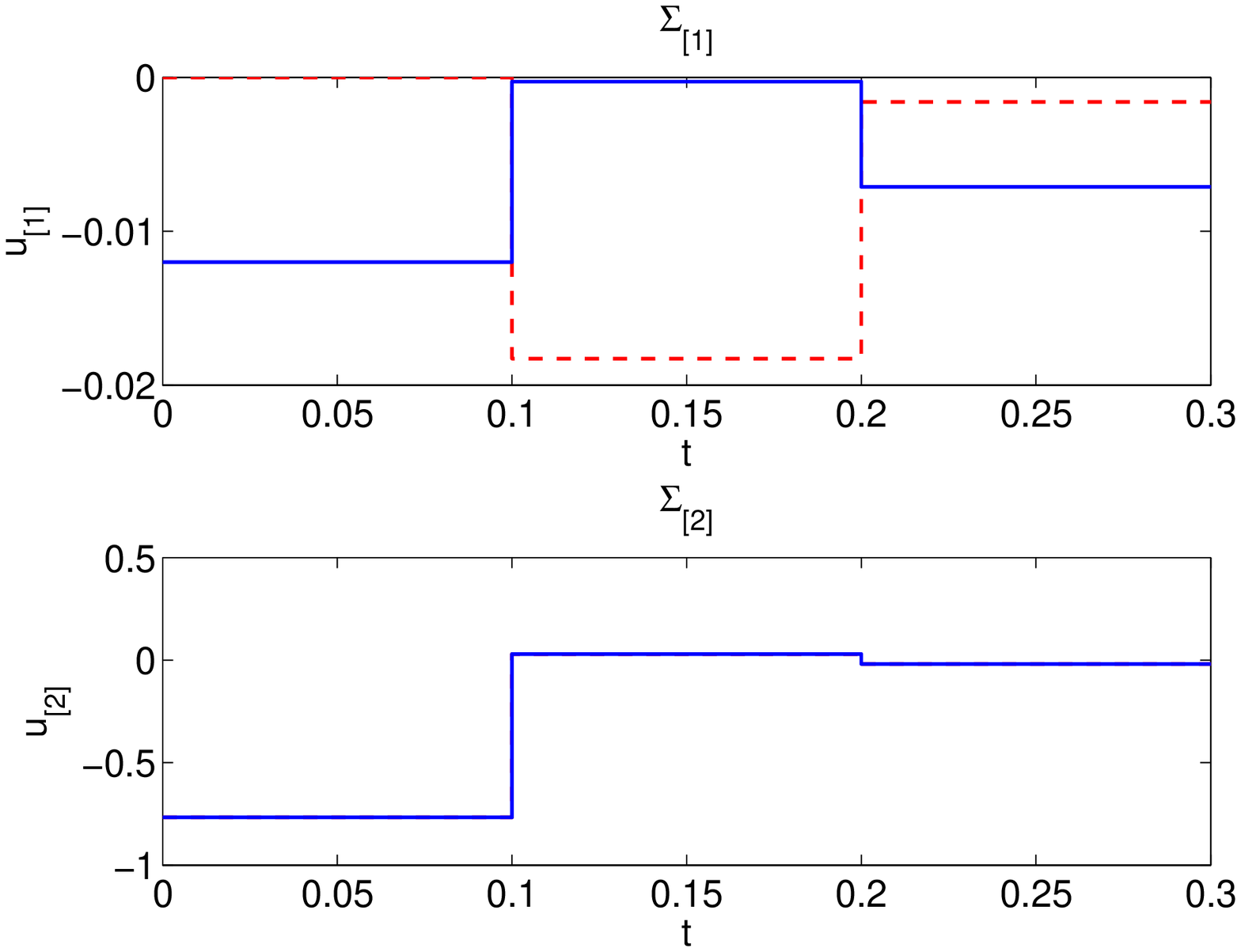}}
                 \caption{Simulation in the first three time-instants with initial state $\mbf x = (0,0,3,0)$.}
                 \label{fig:resultsCart1D}
               \end{figure}

               \subsubsection{Remarks}
                    The example proposed in Figure \ref{fig:exampleCart1D} is of particular interest for decentralized control. We will explain the reason in the continuous-time domain for clarity and simplicity. However, similar considerations apply also to discrete-time versions of the system.

                    The continuous-time system in Figure \ref{fig:exampleCart1D} is described by the following dynamics
                                        \begin{equation}
                      \label{eq:contmass}
                      {\matr{\dot{\subss x {1,1}} \\ \dot{\subss x {1,2}} \\ \dot{\subss x {2,1}} \\ \dot{\subss x {2,2}}}}=\matr{0 & 1 & 0 & 0 \\ -\frac{k_{12}}{m_1} & -\frac{h_{12}}{m_1} & \frac{k_{12}}{m_1} & \frac{h_{12}}{m_1} \\ 0 & 0 & 0 & 1 \\ \frac{k_{12}}{m_2} & \frac{h_{12}}{m_2} & -\frac{k_{12}}{m_2} & -\frac{h_{12}}{m_2} }\matr{\subss x {1,1} \\ \subss x {1,2} \\ \subss x {2,1} \\ \subss x {2,2}}+\matr{ 0 & 0 \\ \frac{100}{m_1}  & 0 \\ 0 & 0 \\ 0 &  \frac{100}{m_2} }\matr{\subss u {1,1} \\ \subss u {2,1}}
                    \end{equation}
                    where $k_{12}>0$, $h_{12}>0$, $m_1>0$ and $m_2>0$. System \eqref{eq:contmass} can be rewritten as
                    \begin{equation}
                      \label{eq:contmass2}
                      {\matr{\dot {\subss x {1}} \\ \dot{\subss x {2}} }}=\matr{ A_{11} & A_{12} \\ A_{21} & A_{22} }\matr{\subss x {1} \\ \subss x {2} }+\matr{ B_1 & 0 \\ 0 & B_2 }\matr{\subss u 1 \\ \subss u 2}
                    \end{equation}
                    that is the coupling of two subsystems $\subss\Sigma 1$ and $\subss\Sigma 2$ with states $\subss x 1 = (\subss x {1,1},\subss x {1,2})$ and $\subss x 1 = (\subss x {2,1},\subss x {2,2})$, respectively. Note that $A_{11}\in\Rset^{2\times 2}$ and $A_{22}\in\Rset^{2\times 2}$ are asymptotically stable matrices, while the matrix 
                    $$\mathbf{A}=\matr{ A_{11} & A_{12} \\ A_{21} & A_{22} }$$
                    is marginally stable, because no mass is bound to
                    a fixed coordinate frame through
                    springs. Importantly, the latter remarks apply
                    independently of (positive) values of parameters $h_{12}$ and $k_{12}$, i.e. for all possible coupling magnitudes.
                    
                    In order to design a decentralized auxiliary
                    control law, with the objective of stabilizing the
                    local systems without accounting for the coupling
                    terms, one could set $K_1=K_2=0$. This would
                    result in asymptotically stable ``local''
                    subsystems (i.e., $A_{11}+B_{1}K_1=A_{11}$ and
                    $A_{22}+B_2K_2=A_{22}$ are both Hurwitz stable
                    matrices), but a marginally stable global system (i.e., $\mathbf{A}+\mathbf{B}\mathbf{K}=\mathbf{A}$ is marginally stable).\\
                    An alternative choice could be to design $K_1$ and $K_2$ using linear quadratic regulators. Next, we show that, if this is done without accounting for couplings, asymptotically stability might not hold. Indeed, this is the case if weights $Q_i = \diag (0,q_i)$ and $R_i=r_i$, $q_i>0$, $r_i>0$, are used. The corresponding controllers are in the form $K_1 = \matr{ 0 & \sigma_{1} }$ and $K_2 = \matr{ 0 & \sigma_{2} }$ which, for all possible values of $q_i>0$, $r_i>0$, are not able to stabilize the global system $\mathbf{A}+\mathbf{B}\mathbf{K}$.
                    
                    These considerations show that, accounting for
                    couplings in the control design phase is
                    fundamental even in simple case studies like the
                    one analyzed in this Section. Again, it is worth
                    remarking that these considerations hold both for
                    small and large coupling terms, since they depend upon structural/physical considerations.

          \subsection{Power Network System}
               In this section, we apply the proposed DeMPC scheme to a power network system composed by several power generation areas coupled through tie-lines. We aim at designing the AGC layer with the goals of
               \begin{itemize}
               \item keeping the frequency approximately at the nominal value, at least asymptotically;
               \item controlling the tie-line powers in order to reduce power exchanges between areas. In the asymptotic regime each area should compensate for local load steps and produce the required power.
               \end{itemize}
               In particular we will show advantages brought about by PnP-DeMPC when generation areas are connected/disconnected to/from an existing network.

               The dynamics of an area equipped with primary control and linearized around equilibrium value for all variables can be described by the following continuous-time LTI model \cite{Saadat2002}
               \begin{equation}
                 \label{eq:ltipower}
                 \subss{\Sigma}{i}^C:\quad\subss{\dot{x}}{i} = A_{ii}\subss x i + B_{i}\subss u i + L_{i}\Delta P_{L_i} + \sum_{j\in\NN_i}A_{ij}\subss x j
               \end{equation}
               where $\subss x i=(\Delta\theta_i,~\Delta\omega_i,~\Delta P_{m_i},~\Delta P_{v_i})$ is the state, $\subss u i = \Delta P_{ref_i}$ is the control input of each area, $\Delta P_{L}$ is the local power load and $\NN_i$ is the sets of predecessor areas, i.e. areas directly connected to $\subss\Sigma i^C$ through tie-lines. The matrices of system \eqref{eq:ltipower}, the parameters values and the state and input constraints are provided in \cite{Riverso2012f}. For each scenario, discrete-time models $\subss\Sigma i$ with $T_s = 1~$sec sampling time are obtained from $\subss\Sigma i^C$ using discretization system-by-system, i.e. exact discretization for each area treating $\subss u i$, $\Delta P_{L_i}$ and $\subss x j,~j\in\NN_i$ as exogenous inputs. We note that the proposed discretization preserves the input-decoupled structure of $\subss\Sigma i^C$.\\
               In the following we first design the AGC layer for a power network composed by four areas (Scenario 1 in \cite{Riverso2012f}) and then we show how in presence of connection/disconnection of an area (Scenario 2 and 3 in \cite{Riverso2012f}, respectively) the AGC can be redesigned via plugging in and unplugging operations.

               \subsubsection{Control experiments}
                    Different control schemes will be compared with the centralized MPC schemes controller described in \cite{Riverso2012f}. For a given Scenario, for each area, at time $t$ control variables $\subss u i$ are obtained through \eqref{eq:tubecontrol} where $\subss v i=\kappa_i(\subss x i)$ and $\subss\hx i=\eta_i(\subss x i)$ are computed at each time $t$ solving the optimization problem \eqref{eq:decMPCProblem} and replacing \eqref{eq:costMPCProblem} with the following cost function depending upon $\subss x i^O=(0,~0,~\Delta P_{L_i},~\Delta P_{L_i})$ and $\subss u i^O=\Delta P_{L_i}$
                    \begin{equation}
                      \label{eq:decMPCproblemPower}
                      \Pset_i^N(\subss x i(t)) = \min_{\substack{\subss\hx i(t)\\\subss v i(t:t+N_i-1)}}~\sum_{k=t}^{t+N_i-1}(\norme{\subss\hx i(k)-\subss x i^O}{Q_i}+\norme{\subss v i(k)-\subss u i^O}{R_i})+\norme{\subss\hx i(t+N_i)-\subss x i^O}{S_i}.
                    \end{equation}
                    As described in \cite{Riverso2012f}, this modification is necessary for achieving compensation of local power load. In the cost function \eqref{eq:costMPCProblem} we set $N_i=15$, $Q_i = \diag(500,0.01,0.01,10)$ and $R_i = 10$. Weights $Q_i$ and $R_i$ have been chosen in order to penalize the angular displacement $\Delta\theta_i$ and to penalize slow reactions to power load steps. Since the power transfer between areas $i$ and $j$ is given by
                  \begin{equation}
                    \label{eq:powerexchanged}
                    \Delta P_{{tie}_{ij}}(k) = P_{ij}(\Delta\theta_i(k)-\Delta\theta_j(k))
                  \end{equation}
                  the first requirement also penalizes huge power transfers. For centralized MPC we consider the overall system composed by the four areas, use the cost function $\sum_{i\in\MM}\Pset_i^N(\subss x i(t))$ and impose the collective constraints \eqref{eq:constraints}. In order to guarantee the stability of the closed loop system, we design the matrix $S_i$ and the terminal constraint set $\hXset_{f,i}$ in two different ways.
                  \begin{itemize}
                  \item\emph{$S_i$ is full ($MPCdiag$)}: we compute the symmetric positive-definite matrix $S_i$ and the static state-feedback auxiliary control law $K_i^{aux}\subss x i$, by maximizing the volume of the ellipsoid described by $S_i$ inside the state constraints while fulfilling the matrix inequality $(A_{ii}+B_iK_i^{aux})'S_i(A_{ii}+B_iK_i^{aux})-S_i\leq-Q_i-K_i^{aux~'}R_iK_i^{aux}$. In order to compare centralized, decentralized and distributed controllers, for the centralized MPC problem we compute the decentralized symmetric positive-definite matrix $\mbf S$ and the decentralized static state-feedback auxiliary control law $\mbf{K^{aux}x}$, $\mbf{K^{aux}}=\diag(K_{1},\ldots,K_{M})$ by maximizing the volume of the ellipsoid described by $\mbf S$ inside the state constraints while fulfilling the matrix inequality $(\mbf{A+BK^{aux}})'\mbf S(\mbf{A+BK^{aux}})-\mbf S\leq-\mbf Q-\mbf{K^{aux~'}RK^{aux}}$.
                  \item\emph{Zero terminal constraint ($MPCzero$)}: we set $S_i=0$ and $\Xset_{f_i} = \subss{x^O}{i}$.
                  \end{itemize}

                  We propose the following performance criteria for evaluating different control schemes.
                  \begin{itemize}
                  \item\emph{$\eta$-index}
                    \begin{equation}
                      \label{eq:performanceeta}
                      \eta = \frac{1}{T_{sim}}\sum_{k=0}^{T_{sim}-1}\sum_{i=1}^M (\norme{\subss x i(k)-\subss x i^O(k)}{Q_i}+\norme{\subss u i(k)-\subss u i^O(k)}{R_i})
                    \end{equation}
                    where $T_{sim}$ is the time of the simulation. From \eqref{eq:performanceeta}, $\eta$ is a weighted average of the error between the real state and the equilibrium state and between the real input and the equilibrium input.
                  \item \emph{$\Phi$-index}
                    \begin{equation}
                      \label{eq:performancePhi}
                      \Phi = \frac{1}{T_{sim}}\sum_{k=0}^{T_{sim}-1}\sum_{i=1}^M\sum_{j\in\NN_i}\abs{\Delta P_{{tie}_{ij}}(k)}T_s
                    \end{equation}
                    where $T_{sim}$ is the time of the simulation and $\Delta P_{{tie}_{ij}}$ is the power transfer between areas $i$ and $j$ defined in \eqref{eq:powerexchanged}. This index gives the average power transferred between areas. In particular, if the $\eta$-index is equal for two regulators, the best controller is the one that has the lower value of $\Phi$.
                  \end{itemize}

               \subsubsection{Scenario 1}
                    \label{sec:scenario1}
                    We consider four areas interconnected as in Scenario 1 in Figure~\ref{fig:scenario1}.
                    \begin{figure}[!ht]
                      \centering
                      \includegraphics[scale=0.75]{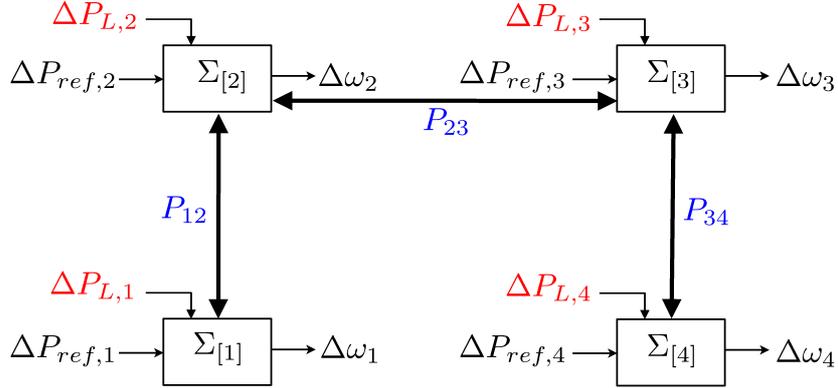}
                      \caption{Power network system of Scenario 1}
                      \label{fig:scenario1}
                    \end{figure}

                    For each system $\subss\Sigma i$ we synthesize the controller $\subss\CC i$, $i\in\MM$ using Algorithm \ref{alg:distrisynt}. Note that in Step \ref{enu:feasProb} of Algorithm \ref{alg:distrisynt} only the feasibility of LP problem is required. Therefore the synthesis of controllers $\subss\CC i$ is computationally more efficient than the nonlinear procedure proposed in \cite{Riverso2012a,Riverso2013c}.
                    \begin{figure}[!ht]
                      \centering
                      \subfigure[\label{fig:simulationscen1freq}Frequency deviation in each area controlled by the proposed DeMPC (bold line) and centralized MPC (dashed line).]{\includegraphics[scale=0.5]{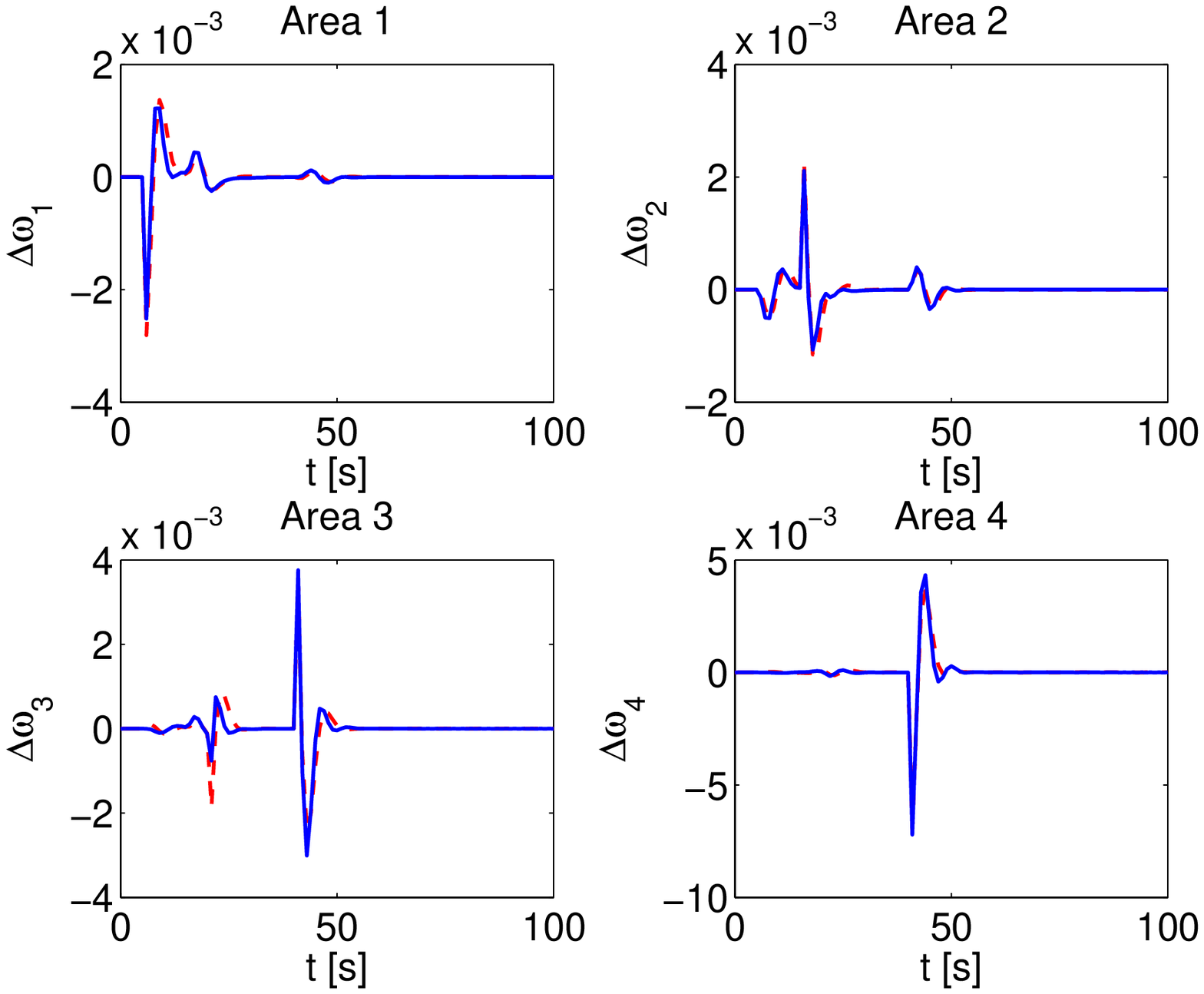}}\\
                      \subfigure[\label{fig:simulationscen1ref}Load reference set-point in each area controlled by the proposed DeMPC (bold line) and centralized MPC (dashed line).]{\includegraphics[scale=0.5]{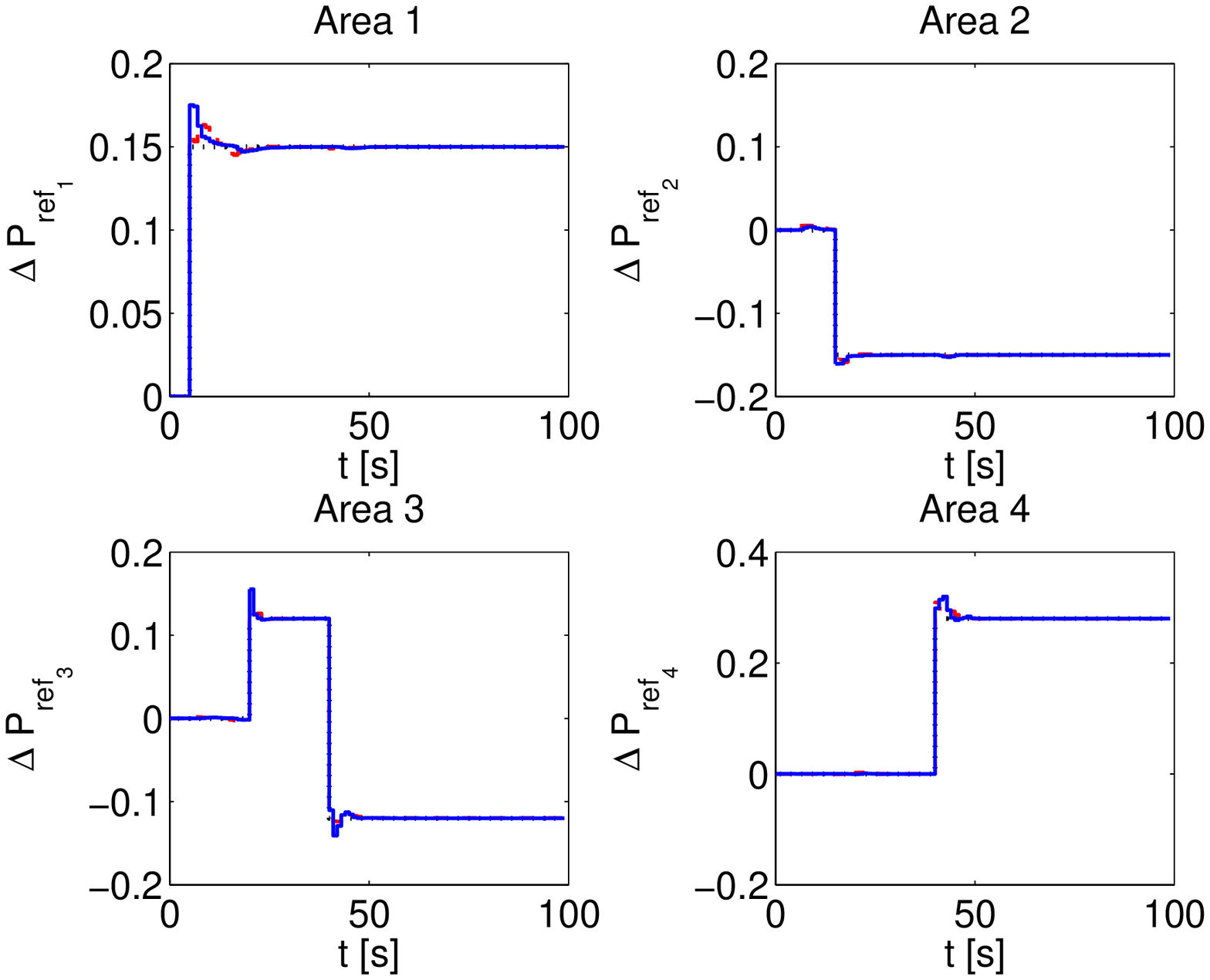}}
                      \caption{Simulation Scenario 1: \ref{fig:simulationscen1freq} Frequency deviation and \ref{fig:simulationscen1ref} Load reference in each area.}
                      \label{fig:simulationscen1}
                    \end{figure}
                    \begin{figure}[!ht]
                      \centering
                      \includegraphics[scale=0.5]{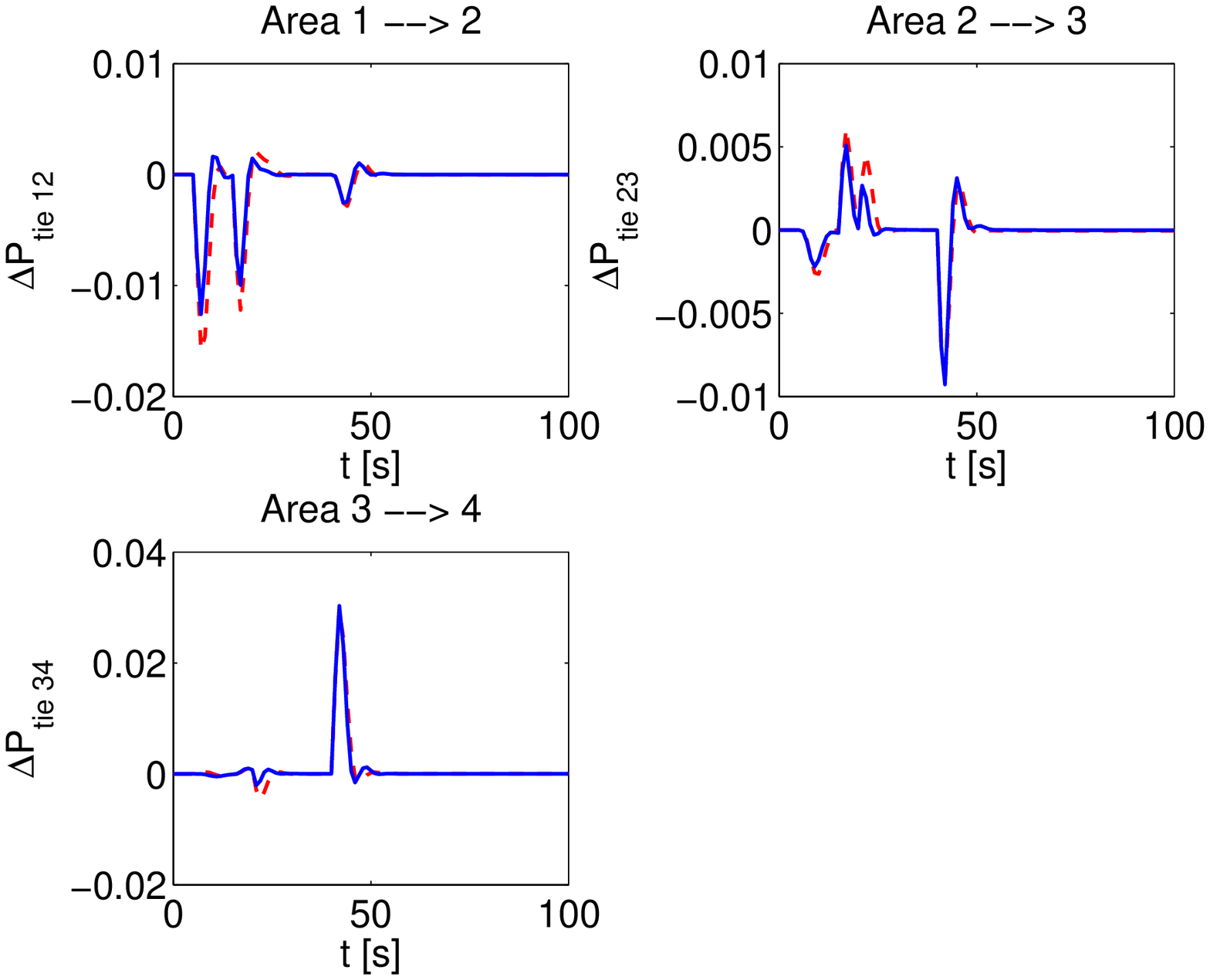}
                      \caption{Simulation Scenario 1: tie-line power between each area controlled by the proposed DeMPC (bold line) and centralized MPC (dashed line).}
                      \label{fig:simulationscen1tiepower}
                    \end{figure}

                    In Figure~\ref{fig:simulationscen1} we compare the performance of the proposed DeMPC scheme with the performance of the centralized MPC controller described in \cite{Riverso2012f}. In the control experiment, step power loads $\Delta P_{L_i}$ are specified in Table 3 of \cite{Riverso2012f} and they account for the step-like changes of the control variables in Figure \ref{fig:simulationscen1}. We highlight that the performance of decentralized and centralized MPC are totally comparable, in terms of frequency deviation (Figure~\ref{fig:simulationscen1freq}), control variables (Figure~\ref{fig:simulationscen1ref}) and power transfers $\Delta P_{{tie}_{ij}}$ (Figure \ref{fig:simulationscen1tiepower}). The values of performance parameter $\eta$ and $\Phi$ using different controllers are reported in Table \ref{tab:simulationsEta} and Table \ref{tab:simulationsPhi}, respectively. In terms of parameter $\eta$, plug-and-play controllers with decentralized and distributed online implementation are equivalent to centralized controller, however the performance of PnP-DeMPC are such that each area can absorb the local loads by producing more power locally ($\Delta P_{ref,i}$) instead of receiving power from predecessor areas: for this reason, PnP-DiMPC has performance more similar to centralized MPC. Compared with plug-and-play controllers proposed in \cite{Riverso2012a,Riverso2013c} (called P\&PMPC), PnP-DeMPC has better performances: we reduce the value of parameter $\eta$ (PnP-DeMPC $0.0256$, P\&PMPC $0.0263$) and especially the value of parameter $\Phi$ (PnP-DeMPC $0.0022$, P\&PMPC $0.0039$). This means that the proposed PnP-DeMPC scheme reduces tie-line powers.

               \subsubsection{Scenario 2}
                    \label{sec:scenario2}
                    We consider the power network proposed in Scenario 1 and we add a fifth area connected as in Figure \ref{fig:scenario2}. Therefore, the set of successors to system $5$ is $\SSS_5=\{2,4\}$.
                    \begin{figure}[!ht]
                      \centering
                      \includegraphics[scale=0.75]{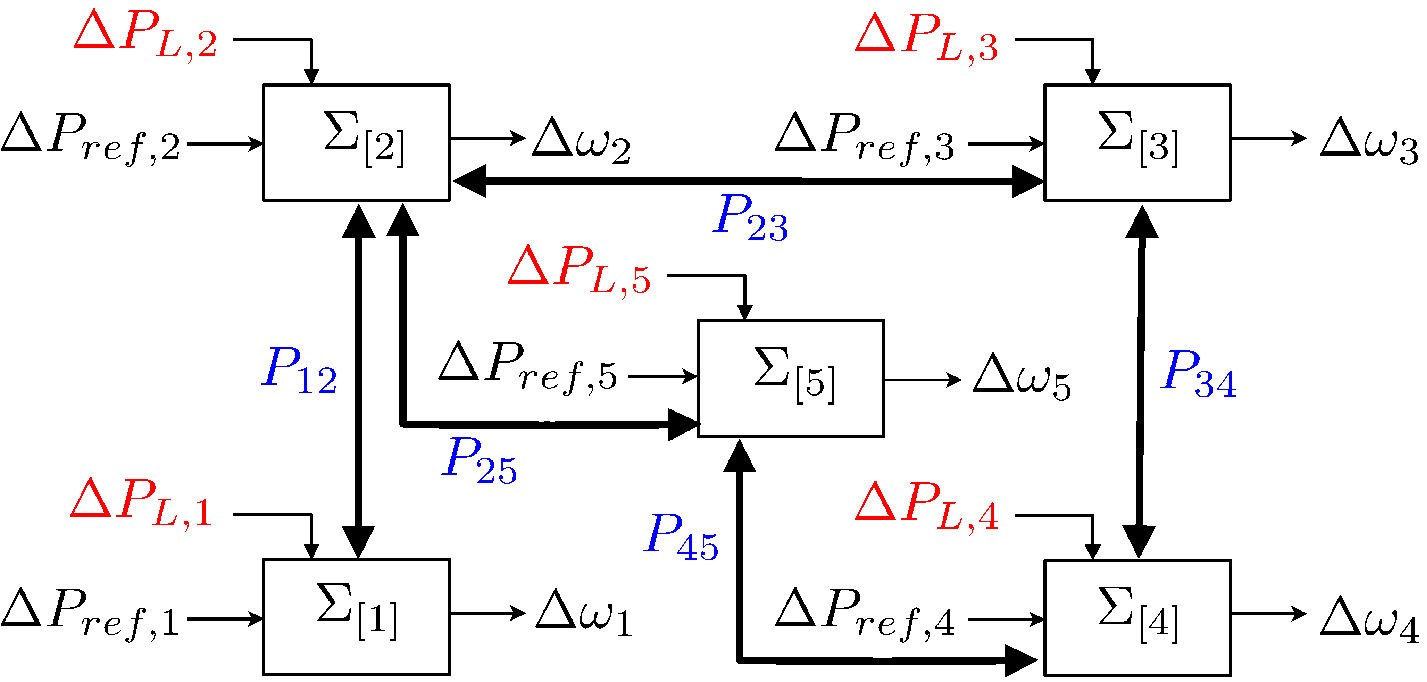}
                      \caption{Power network system of Scenario 2}
                      \label{fig:scenario2}
                    \end{figure}

                    Since systems $\subss\Sigma j$, $j\in\SSS_5$ depend on a parameter related to the added system $\subss\Sigma 5$, a retuning of their controllers is needed. We highlight that our framework, as also the plug-and-play method proposed in \cite{Riverso2012a,Riverso2013c}, allows for subsystems with parameters that depend upon their predecessors. In this case, as discussed in \cite{Riverso2012a,Riverso2013c}, even in the unplugging operation the successors systems have to retune their controllers to guarantee overall asymptotic stability and constraints satisfaction. The controllers $\subss\CC j$, $j\in\{5\}\bigcup\SSS_5$ are tuned using Algorithm \ref{alg:distrisynt}. We highlight that no retuning of controllers $\subss\CC 1$ and $\subss\CC 3$ is needed since $\subss\Sigma 1$ and $\subss\Sigma 3$ are not predecessors of system $\subss\Sigma 5$.
                    \begin{figure}[!ht]
                      \centering
                      \subfigure[\label{fig:simulationscen2freq}Frequency deviation in each area controlled by the proposed DeMPC (bold line) and centralized MPC (dashed line).]{
                        \includegraphics[scale=0.5]{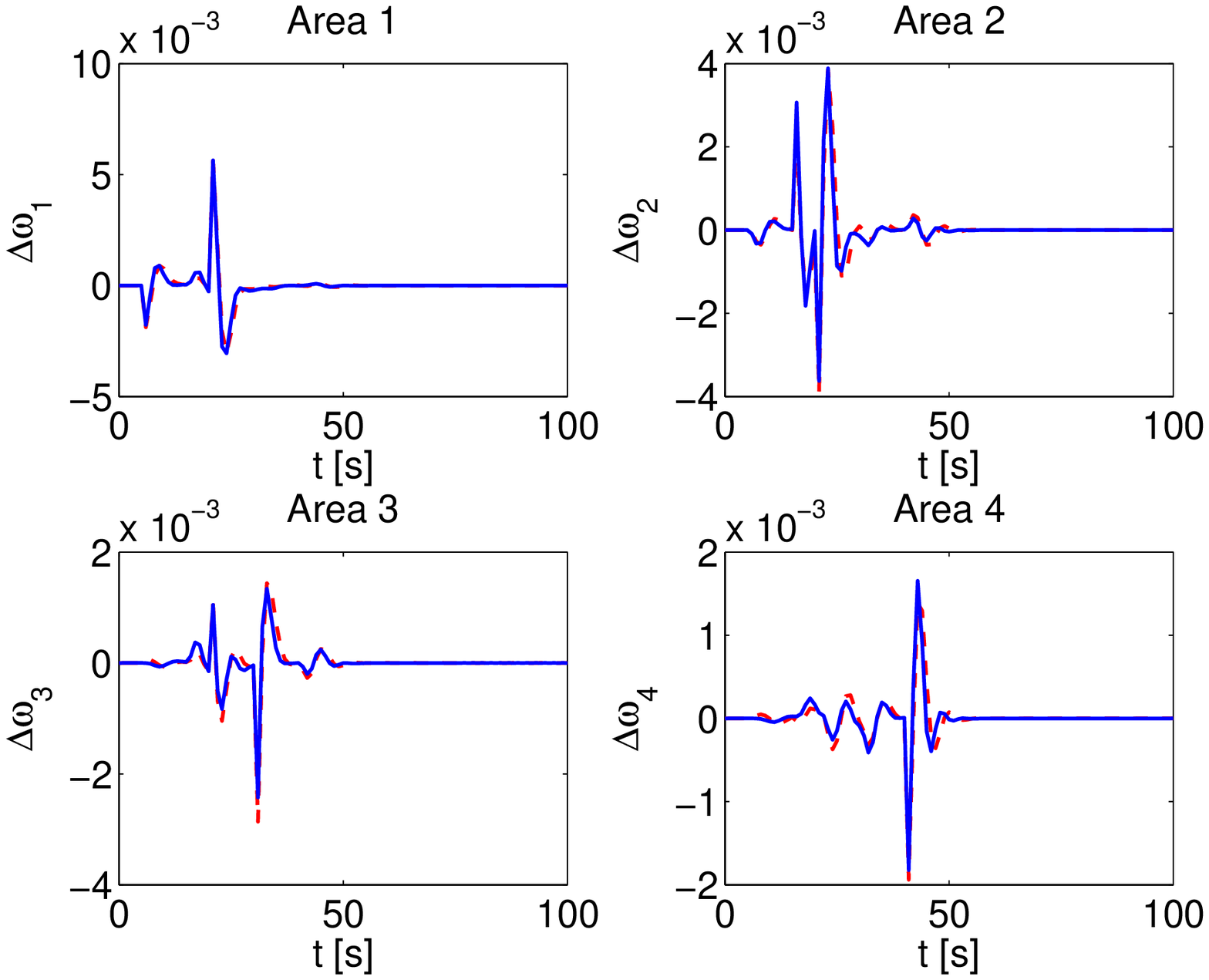}
                        \includegraphics[scale=0.475,width=4.5cm]{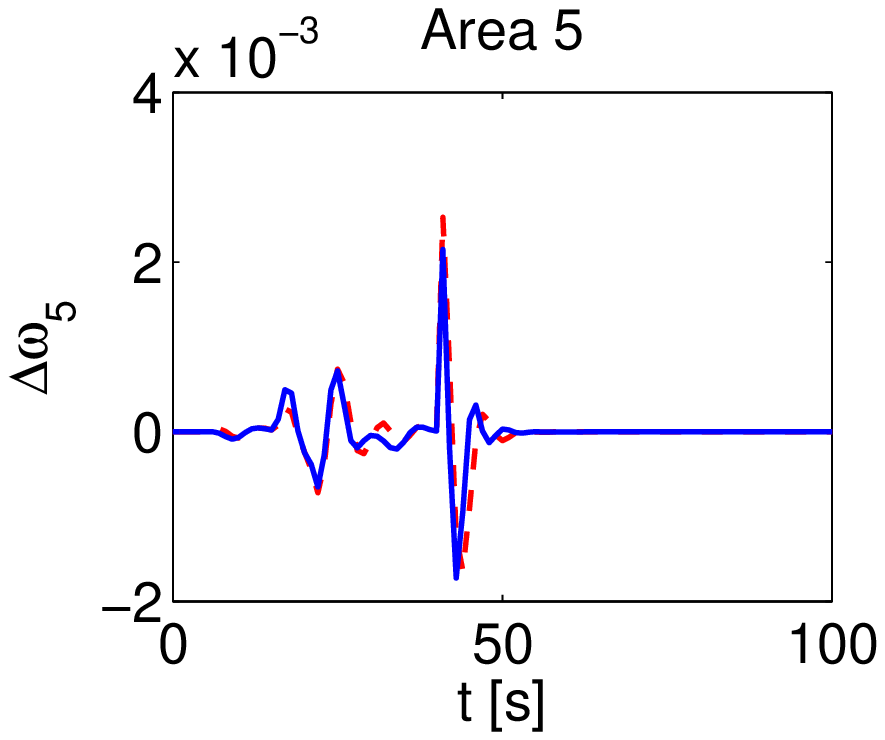}
                      }\\
                      \subfigure[\label{fig:simulationscen2ref}Load reference set-point in each area controlled by the proposed DeMPC (bold line) and centralized MPC (dashed line).]{
                        \includegraphics[scale=0.5]{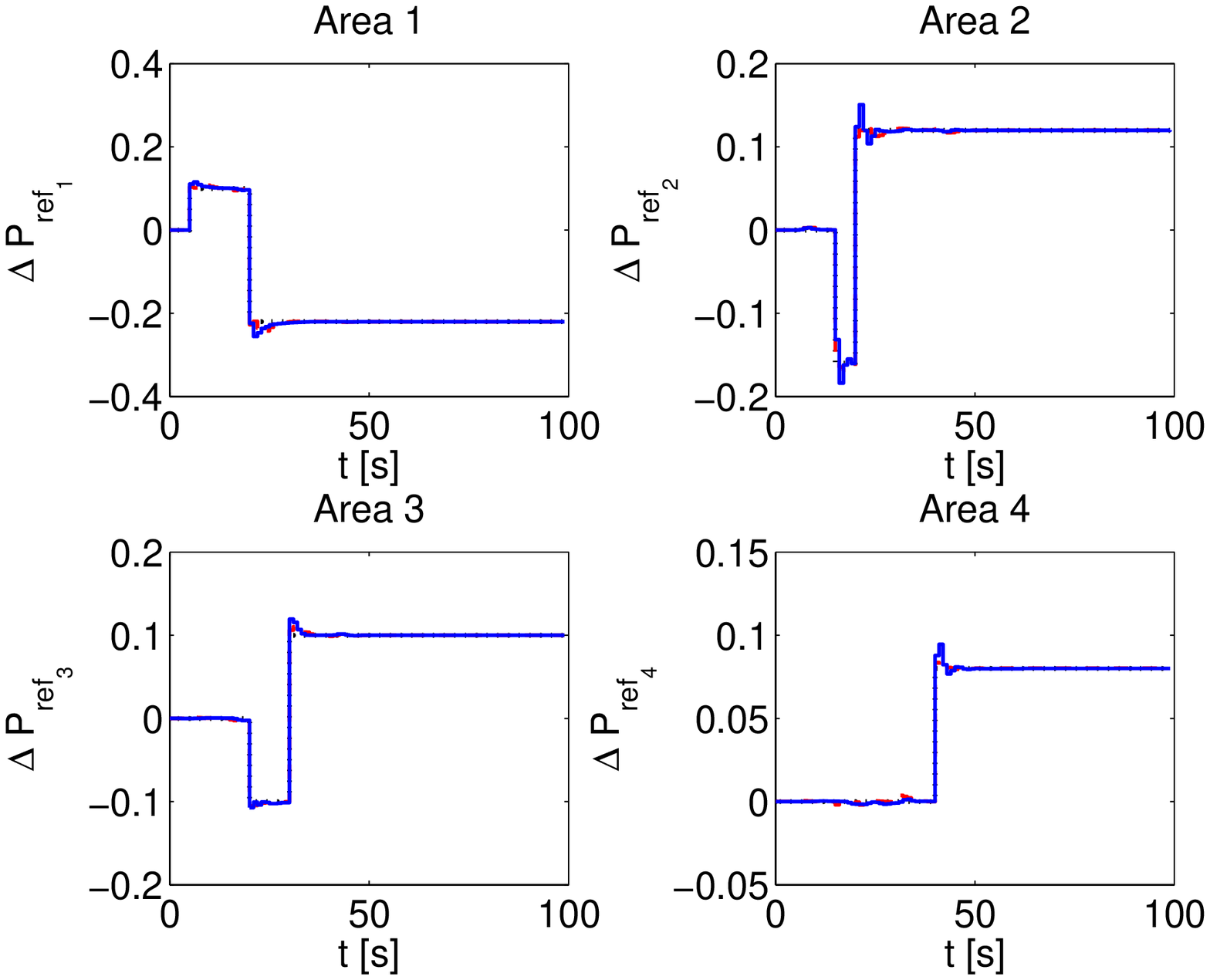}
                        \includegraphics[width=4.5cm,height=3.8cm]{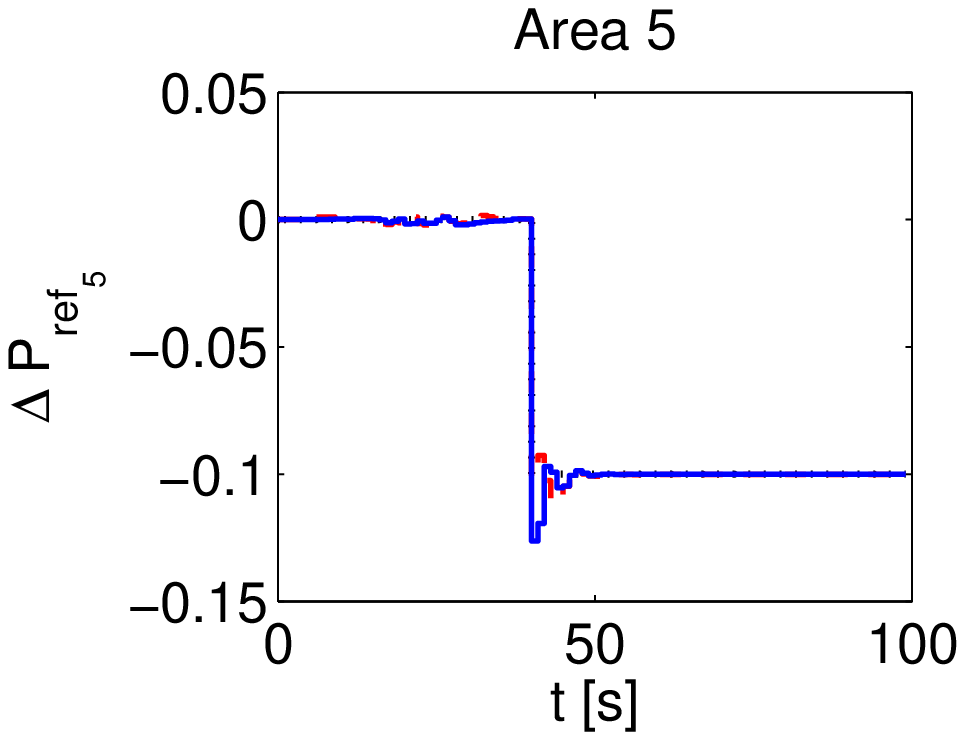}
                      }
                      \caption{Simulation Scenario 2: \ref{fig:simulationscen2freq} Frequency deviation and \ref{fig:simulationscen2ref} Load reference in each area.}
                      \label{fig:simulationscen2}
                    \end{figure}
                    \begin{figure}[!ht]
                      \centering
                      \includegraphics[scale=0.5]{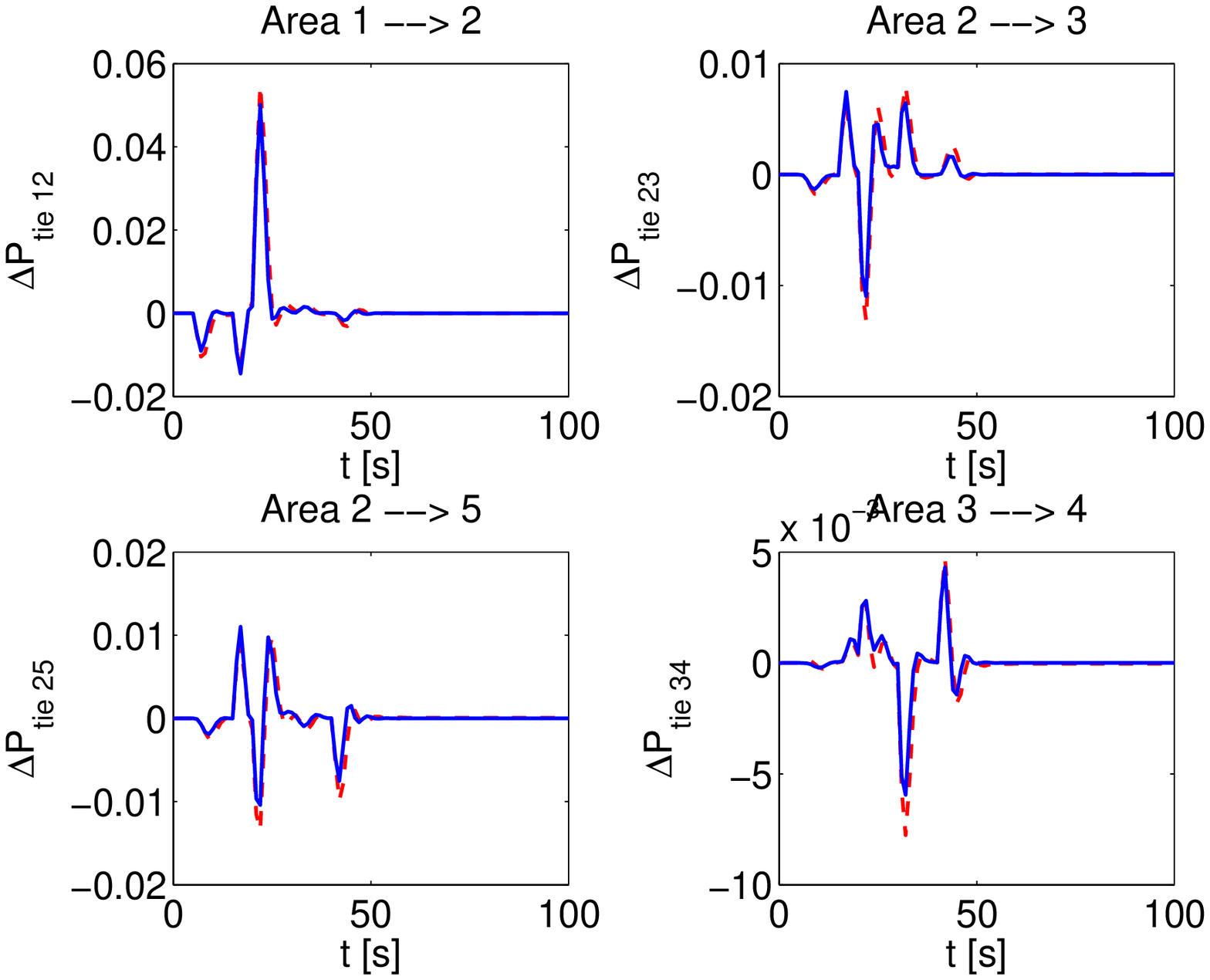}
                      \includegraphics[width=4.25cm,height=3.7cm]{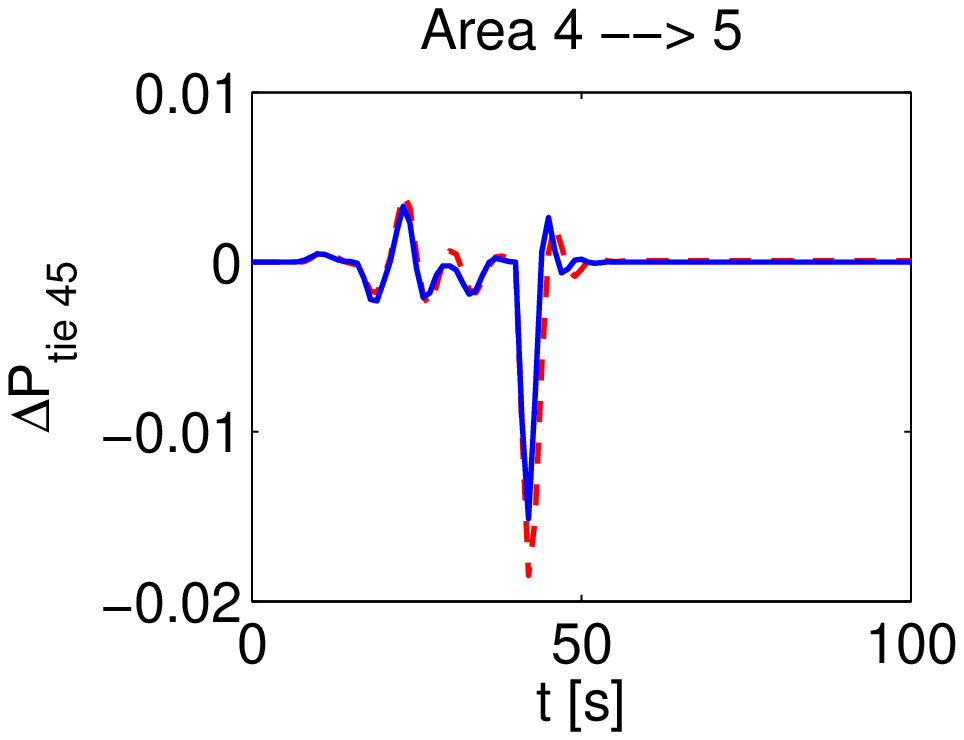}
                      \caption{Simulation Scenario 2: tie-line power between each area controlled by the proposed DeMPC (bold line) and centralized MPC (dashed line).}
                      \label{fig:simulationscen2tiepower}
                    \end{figure}

                    In Figure~\ref{fig:simulationscen2} we compare the performance of proposed DeMPC with the performance of centralized MPC. In the control experiment, step power loads $\Delta P_{L_i}$ specified in Table 4 of \cite{Riverso2012f} have been used and they account for the step-like changes of the control variables in Figure \ref{fig:simulationscen2}. We highlight that the performance of decentralized and centralized MPC are totally comparable, in terms of frequency deviation (Figure~\ref{fig:simulationscen2freq}), control variables (Figure~\ref{fig:simulationscen2ref}) and power transfers $\Delta P_{{tie}_{ij}}$ (Figure \ref{fig:simulationscen2tiepower}). The values of performance parameter $\eta$ and $\Phi$ using different controllers are reported in Table \ref{tab:simulationsEta} and Table \ref{tab:simulationsPhi}, respectively. In terms of parameter $\eta$, plug-and-play controllers with decentralized and distributed online implementation are equivalent to centralized controller, however, as in Scenario 1, the performance of PnP-DeMPC are such that each area can absorb the local loads by producing more power locally ($\Delta P_{ref,i}$) instead of receiving power from predecessor areas: for this reason, PnP-DiMPC has performance more similar to centralized MPC. Compared with P\&PMPC controllers proposed in \cite{Riverso2012a,Riverso2013c}, PnP-DeMPC has better performances in terms of parameter $\Phi$: this corresponds to a reduction of the exchanged power at the price of slightly worse tracking capabilities ($\eta$ increases).

               \subsubsection{Scenario 3}
                    \label{sec:scenario3}
                    We consider the power network described in Scenario 2 and disconnect the area $4$, hence obtaining the areas connected as in Figure \ref{fig:scenario3}. The set of successors to system 4 is $\SSS_4=\{3,5\}$.
                    \begin{figure}[!ht]
                      \centering
                      \includegraphics[scale=0.75]{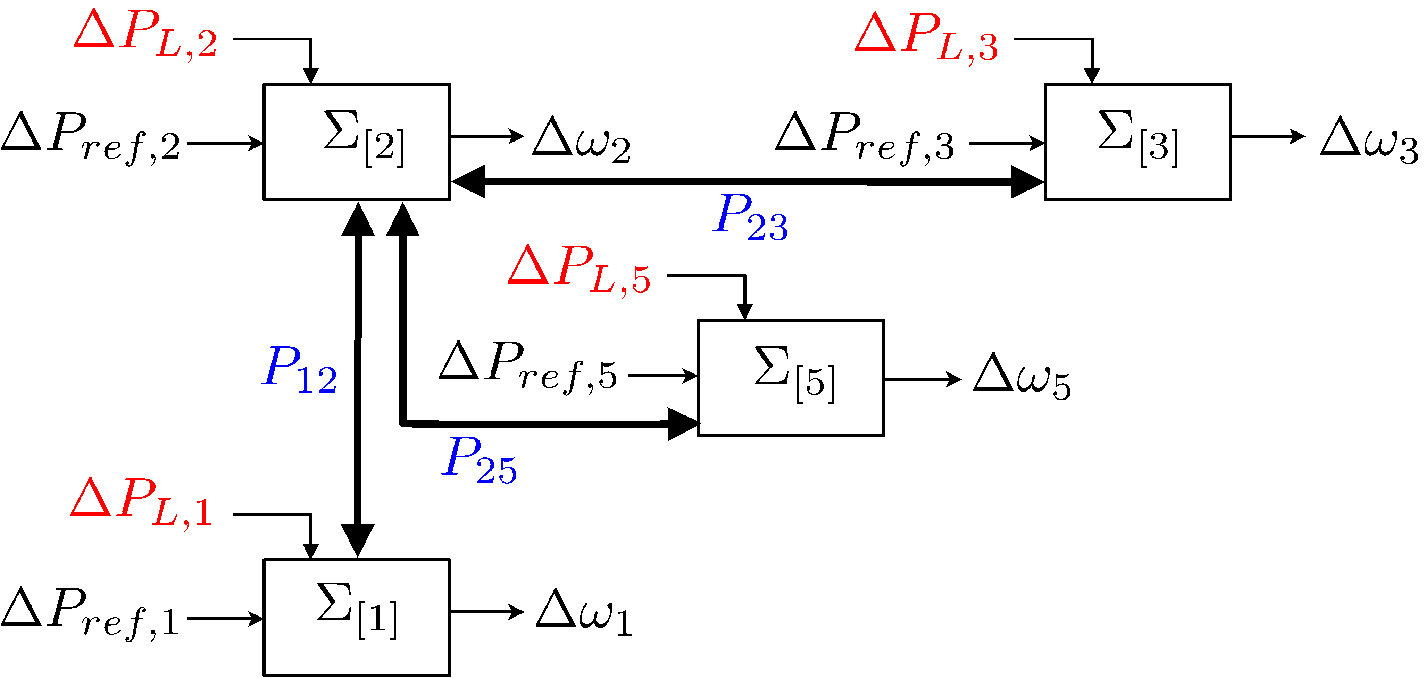}
                      \caption{Power network system of Scenario 3}
                      \label{fig:scenario3}
                    \end{figure}

                    Because of disconnection, systems $\subss\Sigma j$, $j\in\SSS_4$ change their predecessors and local dynamics $A_{jj}$. Then, as explained in Section \ref{sec:scenario2}, the retuning of controllers of successor systems is needed. We highlight that retuning of controllers $\subss\CC 1$ and $\subss\CC 2$ is not needed since systems $\subss\Sigma 1$ and $\subss\Sigma 2$ are not successors of system $\subss\Sigma 4$.
                    \begin{figure}[!ht]
                      \centering
                      \subfigure[\label{fig:simulationscen3freq}Frequency deviation in each area controlled by the proposed DeMPC (bold line) and centralized MPC (dashed line).]{\includegraphics[scale=0.5]{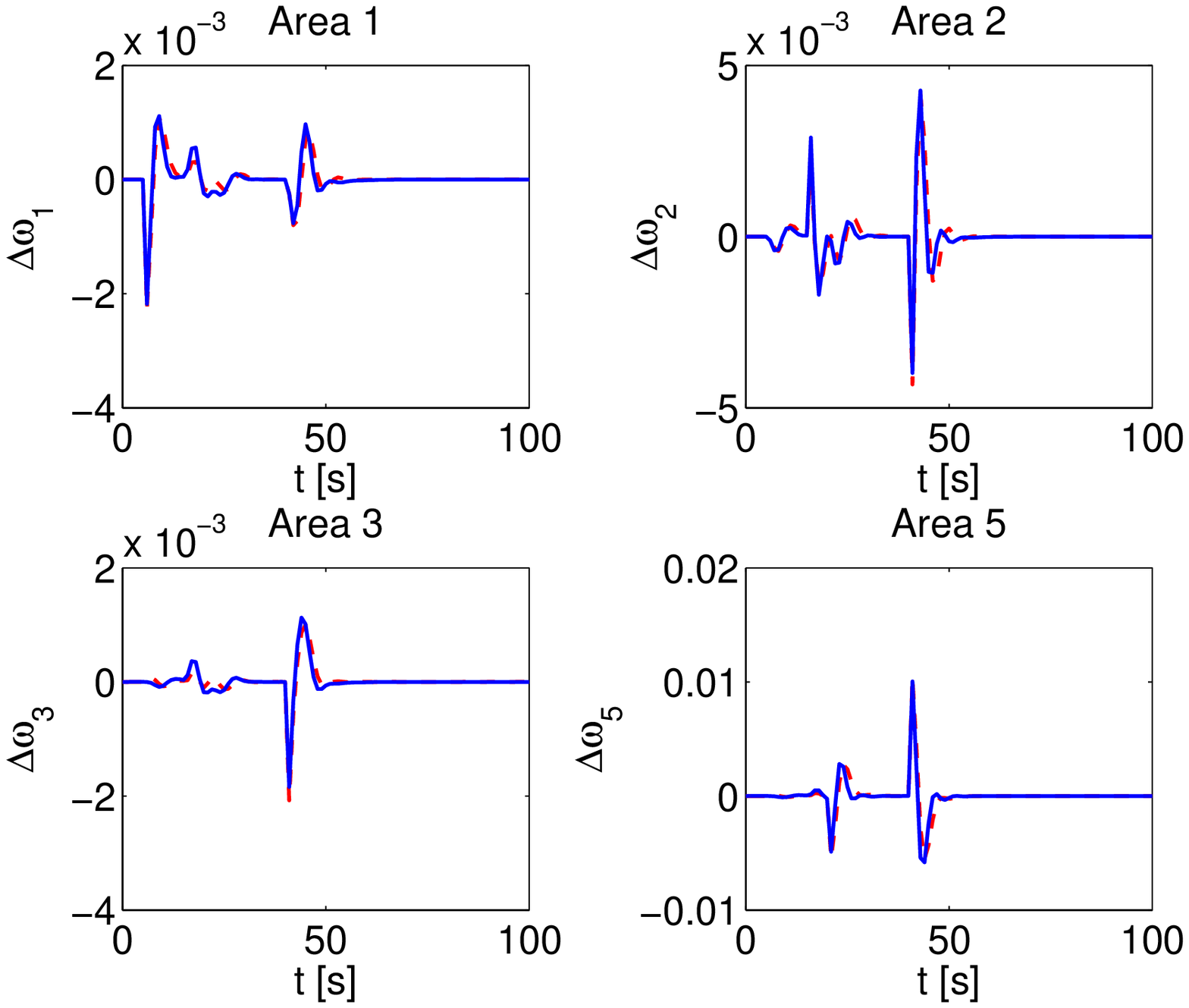}}\\
                      \subfigure[\label{fig:simulationscen3ref}Load reference set-point in each area controlled by the proposed DeMPC (bold line) and centralized MPC (dashed line).]{\includegraphics[scale=0.5]{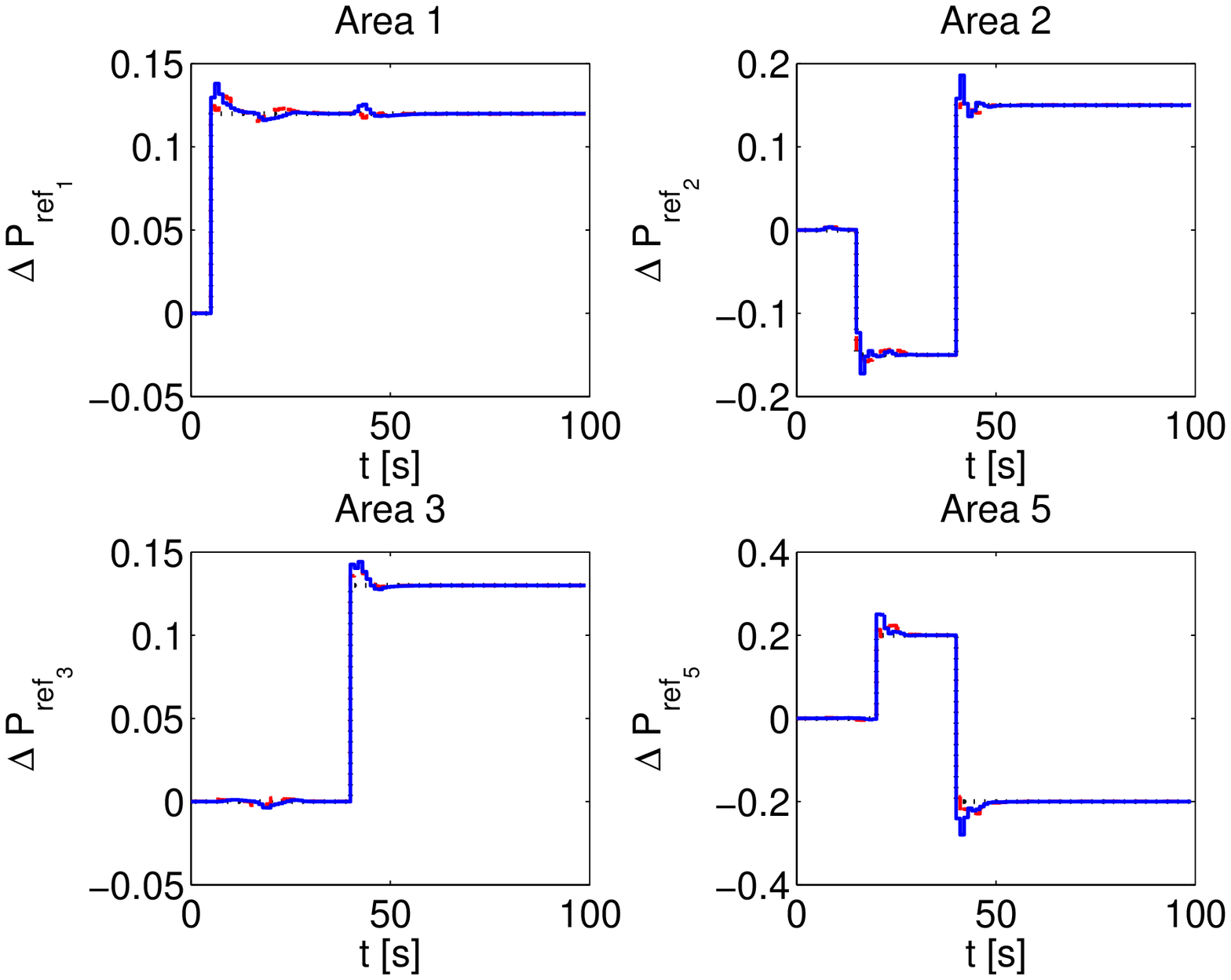}}
                      \caption{Simulation Scenario 3: \ref{fig:simulationscen3freq} Frequency deviation and \ref{fig:simulationscen3ref} Load reference in each area.}
                      \label{fig:simulationscen3}
                    \end{figure}
                    \begin{figure}[!ht]
                      \centering
                      \includegraphics[scale=0.5]{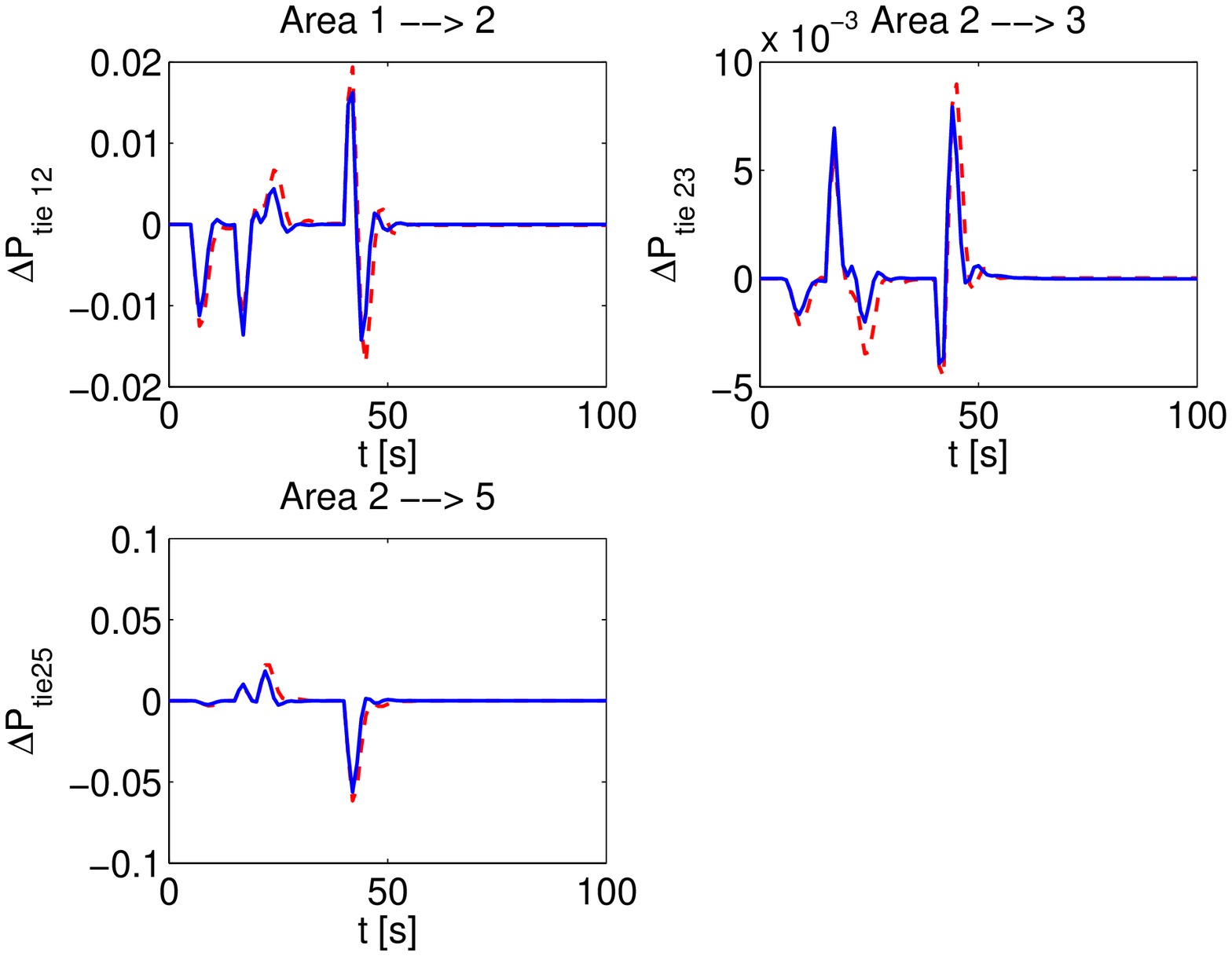}
                      \caption{Simulation Scenario 3: tie-line power between each area controlled by the proposed DeMPC (bold line) and centralized MPC (dashed line).}
                      \label{fig:simulationscen3tiepower}
                    \end{figure}

                    In Figure~\ref{fig:simulationscen3} we compare the performance of proposed DeMPC with the performance of the centralized MPC described in \cite{Riverso2012f}. In the control experiment, step power loads $\Delta P_{L_i}$ specified in Table 5 of \cite{Riverso2012f} have been used. We highlight that the performance of decentralized and centralized MPC are totally comparable in terms of frequency deviation (Figure~\ref{fig:simulationscen3freq}), control variables (Figure~\ref{fig:simulationscen3ref}) and power transfers $\Delta P_{{tie}_{ij}}$ (Figure \ref{fig:simulationscen3tiepower}). The values of performance parameter $\eta$ and $\Phi$ using different controllers are reported in Table \ref{tab:simulationsEta} and Table \ref{tab:simulationsPhi}, respectively. In terms of parameter $\eta$, plug-and-play controllers with decentralized and distributed online implementation are equivalent to centralized controller, however the performance of PnP-DeMPC are such that each area can absorb the local loads by producing more power locally ($\Delta P_{ref,i}$) instead of receiving power from predecessor areas: for this reason, PnP-DiMPC has performance more similar to centralized MPC and PnP-DeMPC reduces of $~40\%$ the performance parameter $\Phi$. Similarly to Scenario 2, compared with P\&PMPC controllers proposed in \cite{Riverso2012a,Riverso2013c}, PnP-DeMPC has better performances in terms of parameter $\Phi$ but worse tracking capabilities ($\eta$ increases).

                    \begin{table}[!ht]
                      \centering
                      \begin{tabular}{|c||c|c||c|c||c|c|}
                        \hline
                        & \multicolumn{2}{|c||}{Scenario 1} & \multicolumn{2}{|c||}{Scenario 2} & \multicolumn{2}{|c|}{Scenario 3} \\
                        \hline
                        &          \emph{MPCdiag}          &    \emph{MPCzero}       &          \emph{MPCdiag}          &    \emph{MPCzero}        &          \emph{MPCdiag}          &    \emph{MPCzero}         \\
                        \hline
                        Cen-MPC    &       0.0249           &   0.0249        &    0.0346            &   0.0347        &    0.0510           &      0.0511       \\
                        \hline
                        De-PnPMPC            &    $\mbf{+2.81\%}$        &   $\mbf{+2.81\%}$  &    $+5.02\%$      &   $+4.90\%$        &   $+7.65\%$            &   $+7.65\%$            \\
                        \hline
                        Di-PnPMPC             &    $+3.61\%$        &   $+3.61\%$  &   $\mbf{+2.31\%}$       &  $\mbf{+2.31\%}$  &   $\mbf{+2.15\%}$            &   $\mbf{+2.15\%}$            \\
                        \hline
                        P\&PMPC \cite{Riverso2012a,Riverso2013c}                   &    $+5.62\%$        &   $+5.62\%$  &   $+4.62\%$       &   $+4.62\%$  &    $+5.69\%$          &    $+5.69\%$         \\
                        \hline
                      \end{tabular}
                      \caption{Value of the performance parameter $\eta$ for centralized MPC (first line) and percentage change using decentralized and distributed MPC schemes for the AGC layer. Best values for PnP controllers are in bold.}
                      \label{tab:simulationsEta}
                    \end{table}

                    \begin{table}[!ht]
                      \centering
                      \begin{tabular}{|c||c|c||c|c||c|c|}
                        \hline
                        & \multicolumn{2}{|c||}{Scenario 1} & \multicolumn{2}{|c||}{Scenario 2} & \multicolumn{2}{|c|}{Scenario 3} \\
                        \hline
                        &          \emph{MPCdiag}          &    \emph{MPCzero}        &         \emph{MPCdiag}         &    \emph{MPCzero}        &          \emph{MPCdiag}          &    \emph{MPCzero}         \\
                        \hline
                        Cen-MPC \cite{Riverso2012f}      &       0.0030           &    0.0029        &    0.0063              &    0.0061        &     0.0060           &      0.0058       \\
                        \hline
                        De-PnPMPC            &    $\mbf{-26.66\%}$        &   $\mbf{-24.14\%}$  &    $\mbf{-31.25\%}$      &   $\mbf{-27.08\%}$        &   $\mbf{-42.85\%}$            &   $\mbf{-38.09\%}$            \\
                        \hline
                        Di-PnPMPC            &    $+0.00\%$        &   $+3.44\%$  &    $-8.62\%$      &   $-5.17\%$        &   $-7.14\%$            &   $-3.57\%$            \\
                        \hline
                        P\&PMPC \cite{Riverso2012a,Riverso2013c}            &    $+16.66\%$     &   $+16.66\%$  &    $+19.05\%$      &   $+22.95\%$        &   $-11.66\%$            &   $-11.66\%$            \\
                        \hline
                      \end{tabular}
                      \caption{Value of the performance parameter $\Phi$ for centralized MPC (first line) and percentage change using decentralized and distributed MPC schemes for the AGC layer. Best values for PnP controllers are in bold.}
                      \label{tab:simulationsPhi}
                    \end{table}

          \subsection{Large-scale System}
               \label{sec:lssexample}
               We consider a large-scale system composed by $1024$ masses coupled as in Figure \ref{fig:lastFrame} where the four edges connected to a point correspond to springs and dampers arranged as in Figure \ref{fig:exampleCart2D}.
               \begin{figure}[htb]
                 \centering
                 \def\svgwidth{250pt}
                 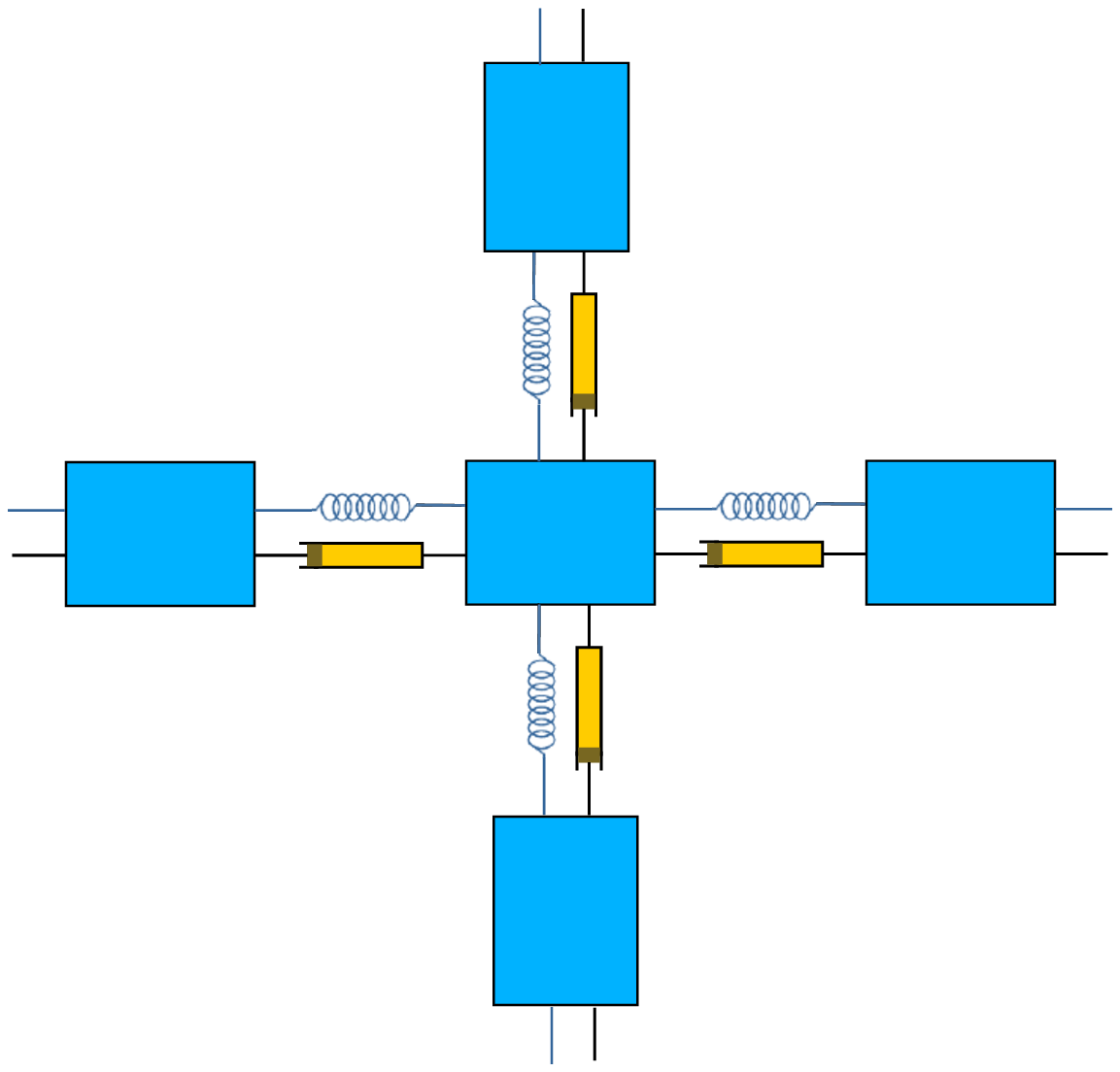
                 \caption{Array of masses: details of interconnections.}
                 \label{fig:exampleCart2D}
               \end{figure}
               Each mass $i\in\MM=1:1024$, is a subsystem with state variables $\subss x i=(\subss x {i,1},\subss x {i,2},\subss x {i,3},\subss x {i,4})$ and input $\subss u i=(\subss u {i,1},\subss u {i,2})$, where $\subss x {i,1}$ and $\subss x {i,3}$ are the displacements of mass $i$ with respect to a given equilibrium position on the plane (equilibria lie on a regular grid), $\subss x {i,2}$ and  $\subss x {i,4}$ are the horizontal and vertical velocity of the mass $i$, respectively, and $100\subss u {i,1}$ (respectively $100\subss u {i,2}$) is the force applied to mass $i$ in the horizontal (respectively, vertical) direction. The values of $m_i$ have been extracted randomly in the interval $[5,10]$ while spring constants and damping coefficients are identical and equal to $0.5$. Subsystems are equipped with the state constraints $\norme{\subss x {i,j}}{\infty}\leq 1.5$, $j=1,3$, $\norme{\subss x {i,l}}{\infty}\leq 0.8$, $i\in\MM$, $l=2,4$ and with the input constraints $\norme{\subss u {i}}{\infty}\leq \beta_i$, where $\beta_i$ have been randomly chosen in the interval $[1,1.5]$. We obtain models $\subss\Sigma i$ by discretizing continuous-time models with $0.2~$sec sampling time, using zero-order hold discretization for the local dynamics and treating $\subss x j,~j\in\NN_i$ as exogenous signals. We synthesized controllers $\subss\CC i$, $i\in\MM$ using Algorithm \ref{alg:distrisynt}. In the worst case the time required to solve Step \ref{enu:feasProb} of Algorithm \ref{alg:distrisynt} is $0.2598~$sec (best case $0.0140~$sec). Note also that the use of a centralized MPC is prohibitive since for the overall system one has $\mbf x\in\Rset^{4096}$, $\mbf u\in\Rset^{2048}$ and the overall state and input constraints are composed respectively by $8192$ and $4096$ affine constraints. In Figure \ref{fig:frameCart1024} we show a simulation where, at time $t=0$, the masses are placed as in Figure \ref{fig:firstFrame}. At all time steps $t$, the control action $\subss u i(t)$ computed by the controller $\CC_i$, for all $i\in\MM$, is kept constant during the sampling interval and applied to the continuous-time system. In the worst case, the solution of the MPC-$i$ problem \eqref{eq:decMPCProblem} requires $0.1047~$sec on a processor Intel Core i7-2600 3.4 GHz, Ram 8GB 1.33 GHz  running Matlab $r2011b$. Convergence is obtained for all masses to their equilibrium position while fulfilling input and state constraints. State and input variables are depicted in Figure \ref{fig:variablesCart1024}. From Figure \ref{fig:variablesCart1024}, the settling time at $95\%$ is $5.37~$sec. For this large-scale system, we also have considered the use of PnP-DeMPC controllers proposed in \cite{Riverso2012a,Riverso2013c}, but since the design of local controllers requires the solution to nonlinear optimization problem and also the number of local constraints is large, the design procedure in \cite{Riverso2012a,Riverso2013c} did not give conclusive results after several hours of computation.

               \begin{figure}[!ht]
                 \centering
                 \subfigure[\label{fig:firstFrame}Position of the masses at initial time.]{\includegraphics[scale=0.52]{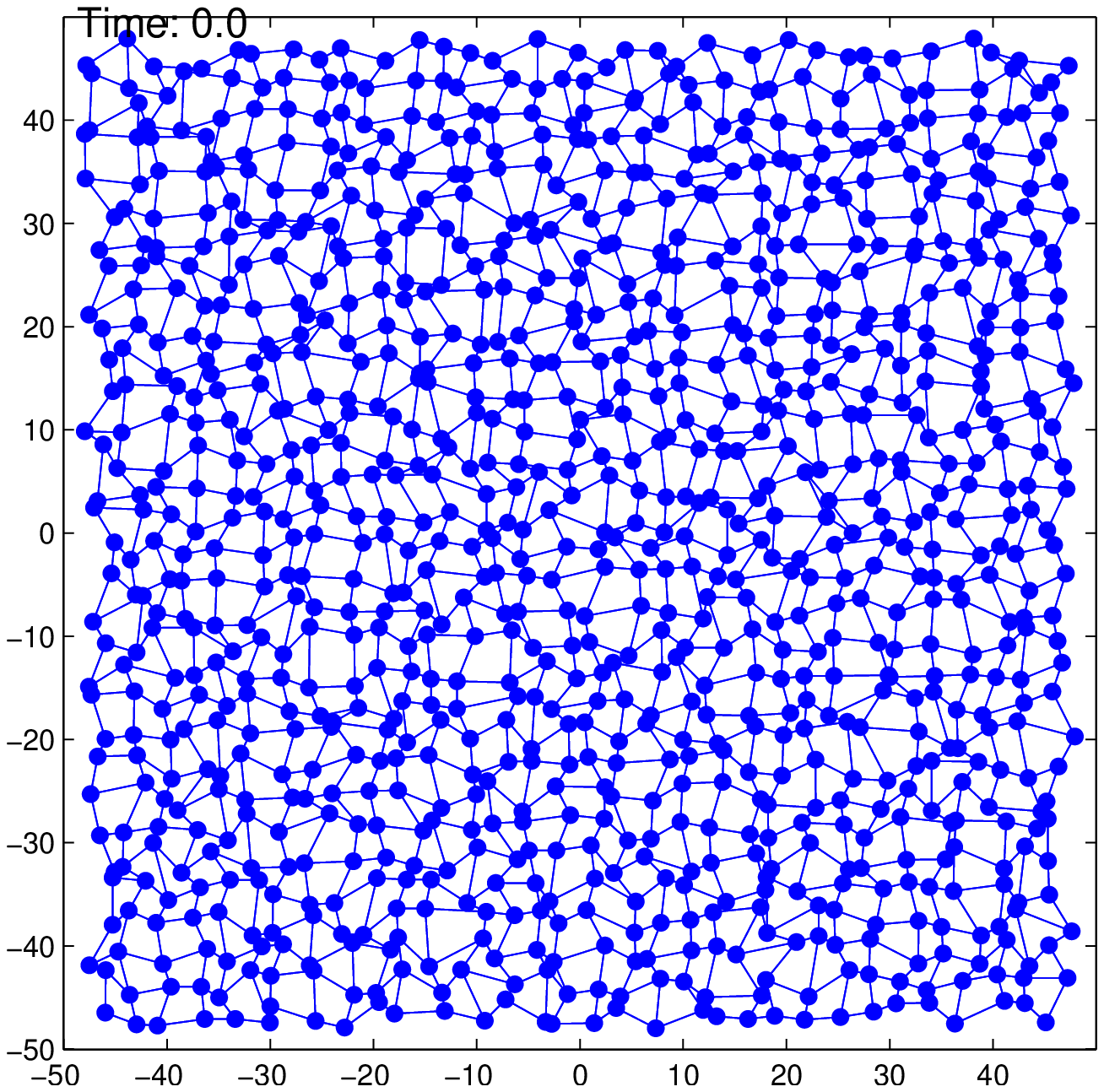}}~~~
                 \subfigure[\label{fig:lastFrame}Position of the masses at time $10~$s ($100$ instant times).]{\includegraphics[scale=0.52]{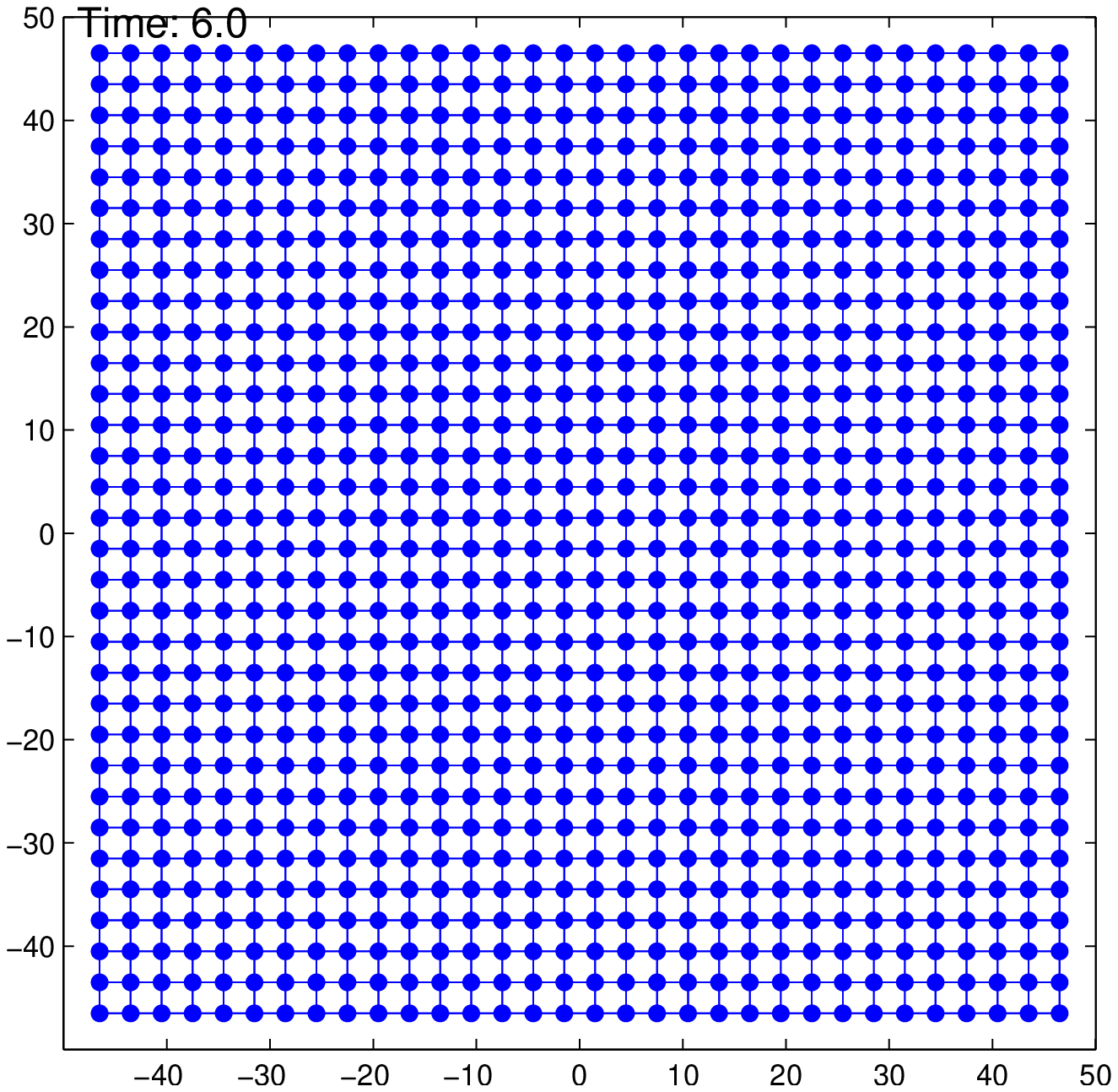}}
                 \caption{Position of the $1024$ trucks on the plane.}
                 \label{fig:frameCart1024}
               \end{figure}
               \begin{figure}[!ht]
                 \centering
                 \subfigure[\label{fig:position1024}Displacements of the masses with respect their equilibrium positions.]{\includegraphics[scale=0.35]{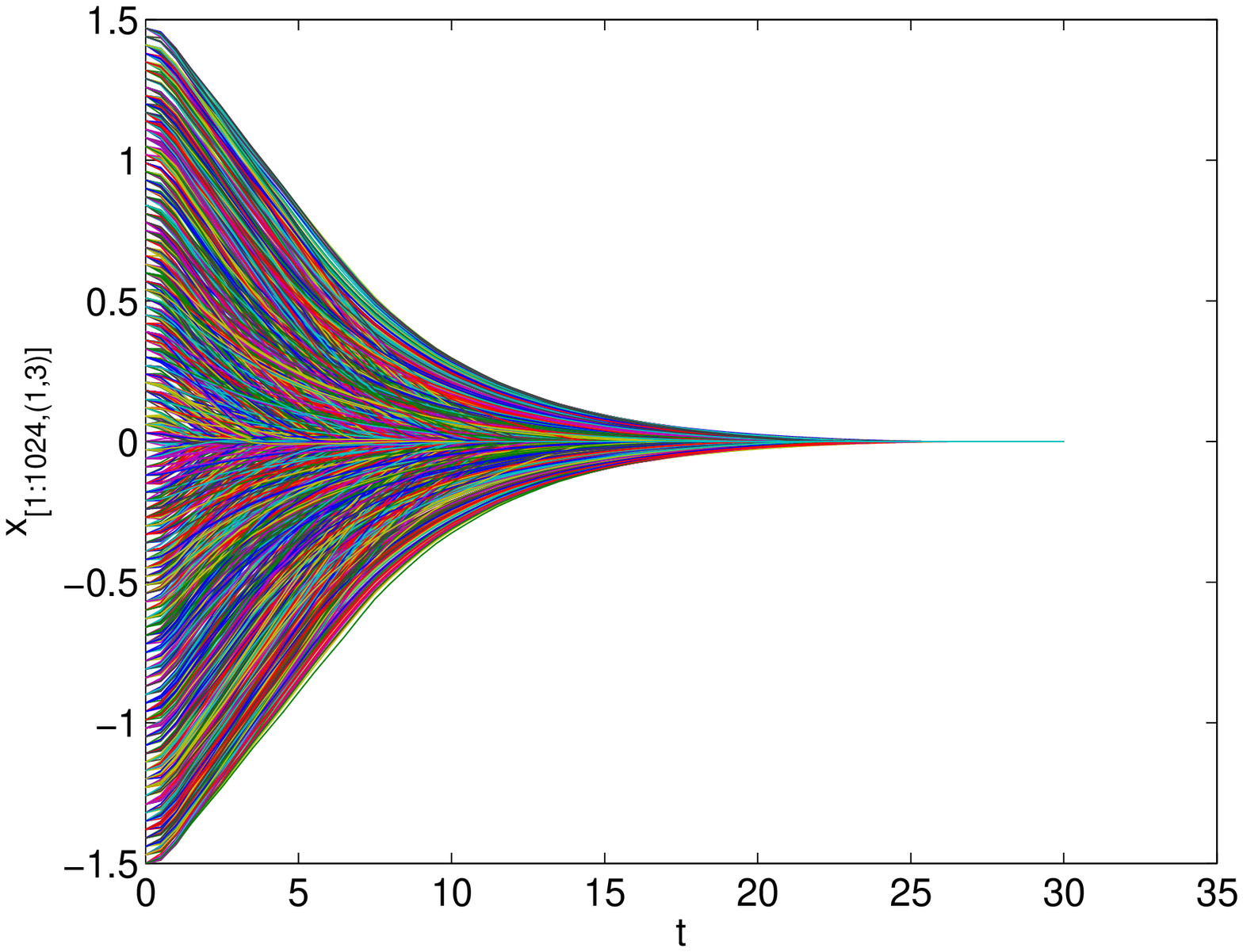}}~~~
                 \subfigure[\label{fig:velocity1024}Velocities, i.e. states $\subss{x}{i,(2,4)}$, $i\in\MM$.]{\includegraphics[scale=0.35]{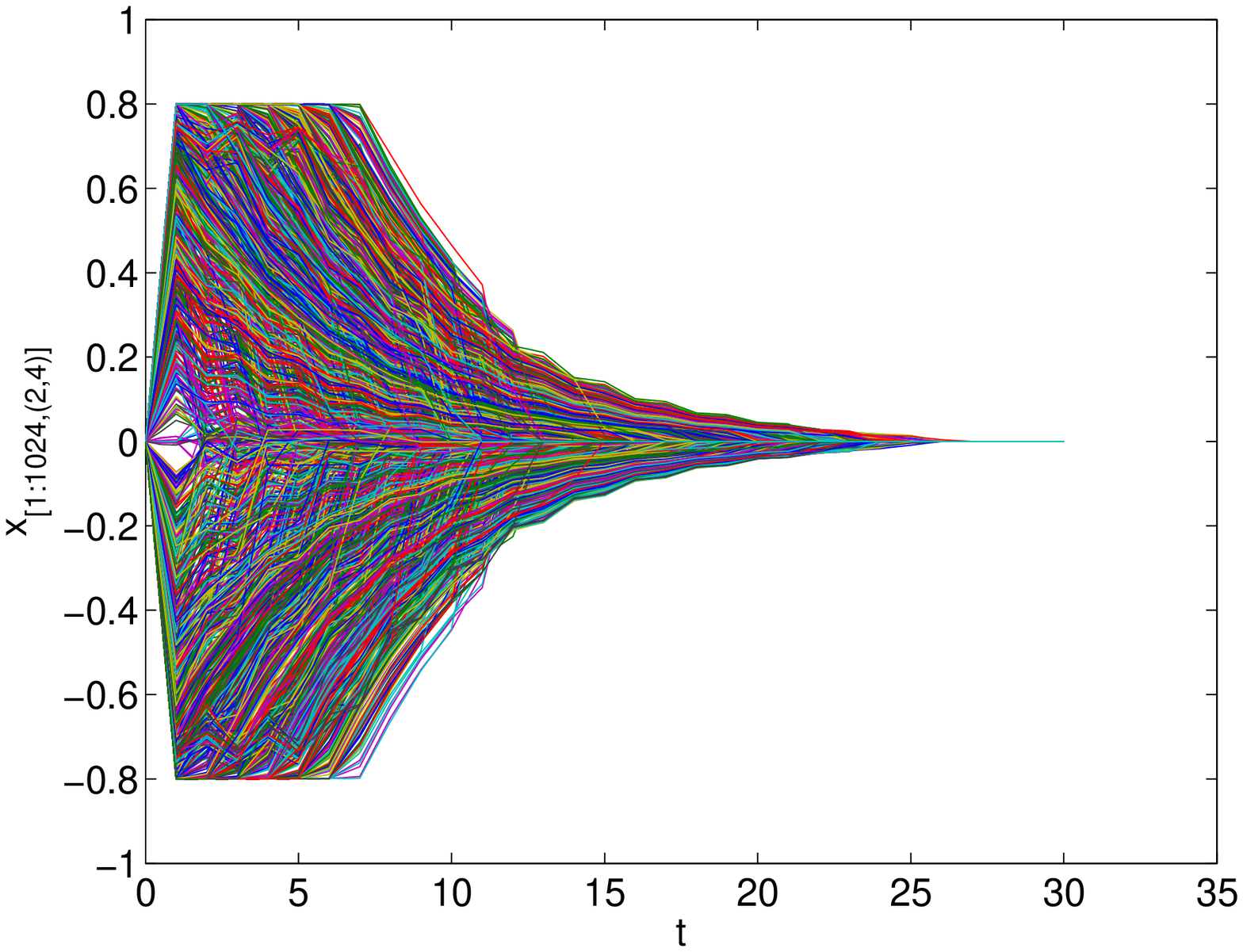}}\\
                 \subfigure[\label{fig:inputH1024}Inputs $\subss{u}{i,1}$, $i\in\MM$.]{\includegraphics[scale=0.35]{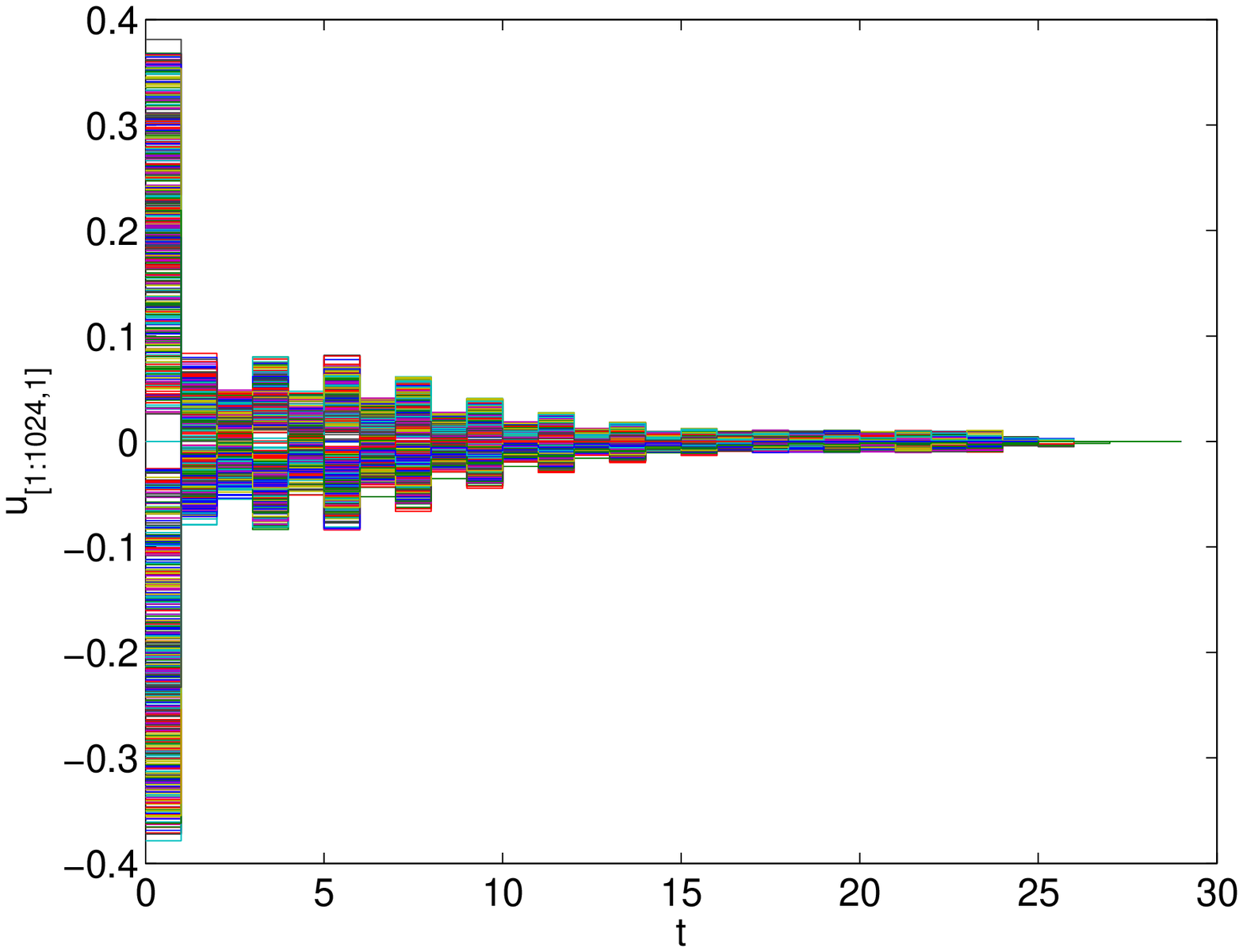}}~~~
                 \subfigure[\label{fig:inputV1024}Inputs $\subss{u}{i,2}$, $i\in\MM$.]{\includegraphics[scale=0.35]{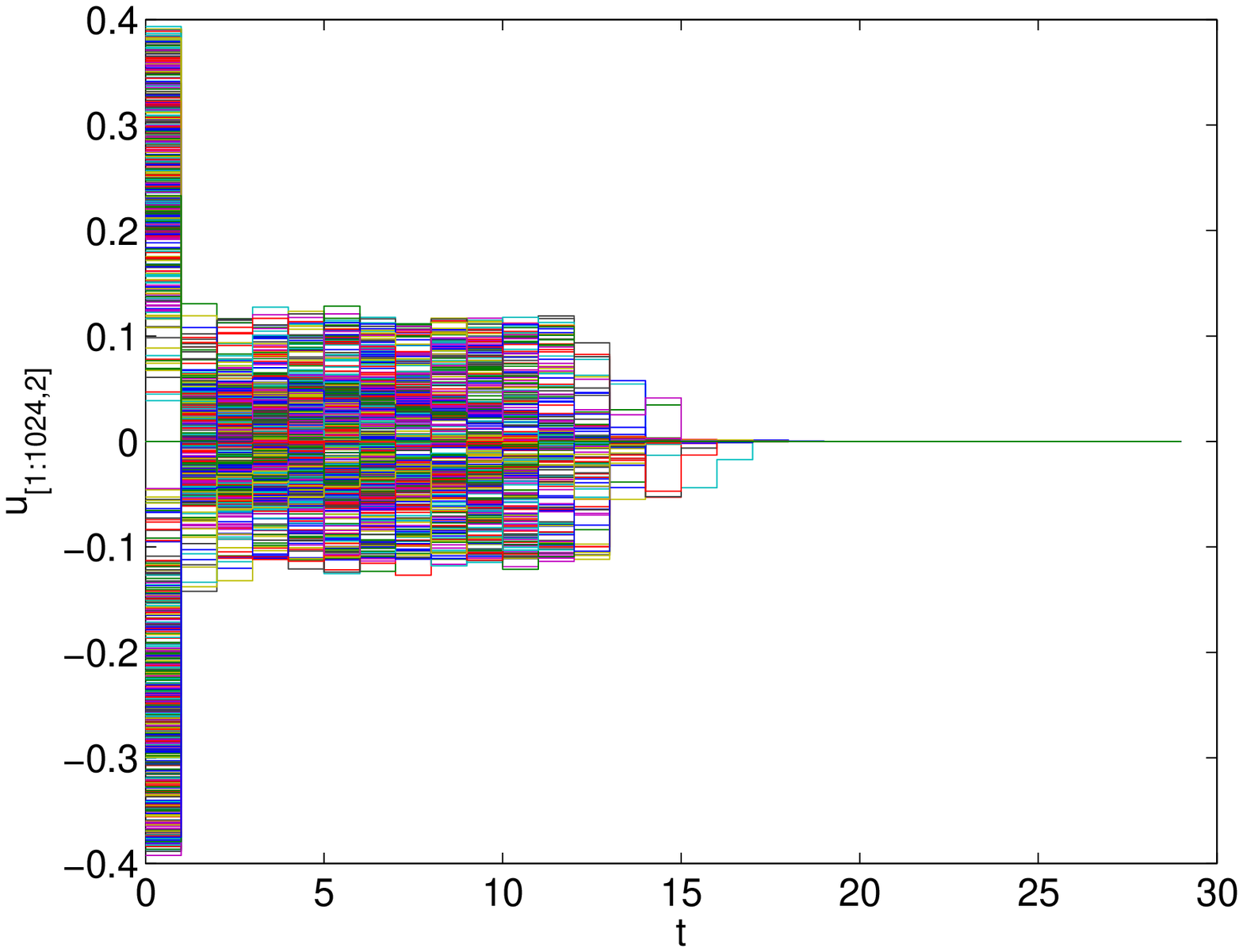}}
                 \caption{State and input trajectories of the $1024$ masses with $\mbf x(0)$ as in Figure \ref{fig:firstFrame}.}
                 \label{fig:variablesCart1024}
               \end{figure}

               \subsubsection{Remarks}
                    The array of masses in Figures \ref{fig:exampleCart2D} and \ref{fig:frameCart1024} is a marginally stable system irrespectively of the magnitude of coupling terms. Indeed, no mass is bound to a fixed reference frame through springs. Next we show that designing local MPC controllers neglecting coupling might compromise feasibility at all times. We used the same MPC settings as for the controllers proposed in Section \ref{sec:lssexample}. Therefore for each subsystem we solve the following MPC problem
                    \begin{subequations}
                      \label{eq:decMPCProblem2}
                      \begin{align}
                        &\label{eq:costMPCProblem2}\Pset_i^N(\subss x i(t)) = \min_{\substack{\subss u i(0:N_i-1|t)}}\sum_{k=0}^{N_i-1} \subss x i(k|t)^T\matr{10&0&0&0\\0&10&0&0\\0&0&10&0\\0&0&0&10}\subss x i(k|t) +\subss u i(k|t)^T\matr{ 1 & 0 \\ 0 & 1}\subss u i(k|t) &\\
                        &\label{eq:dynproblem2}\subss x i(k+1|t)=A_{ii}\subss x i(k|t)+B_i\subss u i(k|t),\qquad k\in0:N_i-1 \\
                        &\label{eq:inhXproblem2}\subss x i(k|t)\in\Xset_i,\qquad\qquad\qquad\qquad\qquad\qquad\quad      k\in1:N_i-1 \\
                        &\label{eq:inVproblem2}\subss u i(k|t)\in\Uset_i,\qquad\qquad\qquad\qquad\qquad\qquad\quad     k\in0:N_i-1 \\
                        &\label{eq:inTerminalSet2}\subss x i(N_i|t)=0
                      \end{align}
                    \end{subequations}
                    where $N_i=20$, $i\in\MM$. Note that MPC problems \eqref{eq:decMPCProblem2} do not depend on coupling terms, both for the online computation of control inputs $\subss u i$ and for the design of the cost function and terminal region.

                    Consider the initial state $\subss x
                    1(0)=(1.5,0.8,0,0)$, $\subss x 2(0)=(1.5,0,0,0)$,
                    and $\subss x l(0)=(0,0,0,0)$, $l\in 3:1024$, that
                    is feasible for the MPC problems
                    \eqref{eq:decMPCProblem2}. For subsystem
                    $\subss\Sigma 1$, we obtain $\subss u
                    1(0|0)=(-0.6304,0.0000)$ and hence, setting $\subss u
                    1(0)=\subss u 1(0|0)$, one gets
                    \begin{equation*}
                      \begin{aligned}
                        \subss x 1(1) = \matr{ 0.9987 & 0.1987 &  0	& 0 \\ -0.01245 & 	 0.9863 &	0 & 0 \\  0 & 0 & 0.9987 & 0.1987 \\  0 & 0 & -0.0125 & 0.9863}\subss x 1(0) +\matr{ 0.2497 & 0 \\ 2.4909 & 0 \\ 0 & 0.2497 \\ 0 & 2.4909 }\subss u 1(0) +\\+ \matr{ 0.0012 & 0.0012 & 0 & 0 \\ 0.0124 & 0.0124 & 0 & 0 \\ 0 & 0 & 0 & 0 \\ 0 & 0 & 0 & 0 }\subss x 2(0)
                      \end{aligned}
                    \end{equation*}
                   Therefore,
                    $$
                    \subss x 1(1) = \matr{ 1.5015 \\ -0.7813 \\ 0 \\ 0 }\notin\Xset_1.
                    $$
                    showing that at time $t=1$ the MPC problem
                    \eqref{eq:decMPCProblem2} for subsystem $\subss\Sigma 1$ is infeasible.


     \section{Conclusions}
          \label{sec:conclusions}
          In this paper we proposed a decentralized MPC scheme based on the notion of tube-based control \cite{Rakovic2005b}. The proposed scheme guarantees asymptotic stability of the closed-loop system and constraints satisfaction at each time instant. Our design procedures allows plug-and-play operations: if a new subsystem enters the network we can design a local controller using information from predecessor systems. Differently from \cite{Riverso2012a,Riverso2013c} we can compute local controllers solving LP optimization problems only. Moreover, no assumption involves quantities of the overall system. Future research will consider generalizations of PnP-MPC design to output feedback schemes and tracking problems.

     \appendix
          \section{Proof of Theorem \ref{thm:mainclosedloop}}
               \label{sec:proofthmmain}
               We start introducing a definition and a few Lemmas that will be used in the proof of Theorem \ref{thm:mainclosedloop}.\\
               \begin{defi}
                 A convex body is a nonempty convex and compact set.
               \end{defi}
               \begin{lem}\cite{Schneider1993}
                 \label{lem:SFlem}
                 Let $A$ and $B$ be convex bodies. The support function of $A$ is $h_A(x)=\sup_{y\in A} x'y$ and it has the following properties: $h_A$ is sublinear (i.e. $h_A(\alpha x)=\alpha h_A(x)$, $\forall\alpha\geq 0$ and $h_A$, $h_A(x+y)\leq h_A(x)+h_A(y)$), $h_{\lambda A}=\lambda h_A$, $\forall\lambda>0$, $h_{A\oplus B}=h_A+h_B$ and $A\subseteq B\Leftrightarrow h_A\leq h_B$.
               \end{lem}
               \begin{lem}
                 \label{lem:outerbound}
                 Let $A=\{x\in\Rset^n:\AAA x\leq \One\}$, $\AAA=[a_1,\ldots,a_p]^T\in\Rset^{p\times n}$ and assume that $A\ominus\ball{\beta}(0)$ strictly contains the origin in its interior. Then,
                 \begin{enumerate}[a)]
                 \item\label{enu:outerbounda} $A\ominus\ball{\beta}(0)=\{x\in\Rset^n:\AAA x\leq \One-\beta(\norme{a_1}{},\ldots,\norme{a_p}{})\}$
                 \item\label{enu:outerboundb} $\beta\norme{a_i}{}<1$, $\forall i\in 1:p$
                 \item\label{enu:outerboundc} defining $\psi=\min_{i\in 1:p}\norme{a_i}{}$ one has $A\ominus\ball{\beta}(0)\subseteq (1-\beta\psi)A$
                 \end{enumerate}
                 \begin{proof}
                   Proceeding as is the proof of point 8 of Proposition 3.28 of \cite{Blanchini2008} one gets
                   $$
                   A\ominus\ball{\beta}(0)=\{x\in\Rset^n:\AAA x\leq \tilde g\}
                   $$
                   where $\tilde g=(\tilde g_1,\ldots,\tilde g_p)$, $\tilde g_i=1-h_{\ball{\beta}(0)}(a_i)$, and $h_{\ball{\beta}(0)}(\cdot)$ is the support function of $\ball{\beta}(0)$. Point (\ref{enu:outerbounda}) follows from $h_{\ball{\beta}(0)}(x)=\beta\norme{x}{}$ (that can be verified using the definition of support functions).\\
                   From point (\ref{enu:outerbounda}) one has $x\in A\ominus\ball{\beta}(0)$ if and only if
                   \begin{equation}
                     \label{eq:outerboundeq1}
                     a_i^Tx\leq 1-\beta\norme{a_i}{},~\forall i\in 1:p
                   \end{equation}
                   and in order to have that all constraints \eqref{eq:outerboundeq1} are fulfilled for $x=0$ and no one is active at $x=0$, point (\ref{enu:outerboundb}) must be verified.\\
                   For point (\ref{enu:outerboundc}), one has
                   \begin{equation}
                     \label{eq:outerboundeq2}
                     A\ominus\ball{\beta}(0)\subseteq\{x\in\Rset^n:\AAA x\leq (1-\beta\psi)\One\}=\{x\in\Rset^n:\AAA (\frac{x}{1-\beta\psi})\leq\One\}
                   \end{equation}
                   where the last equality holds because, from point (\ref{enu:outerboundb}), $\beta\psi<1$. The rightmost set in \eqref{eq:outerboundeq2} is $(1-\beta\psi)A$ and this concludes the proof.
                 \end{proof}
               \end{lem}

               The next result shows that the control law $\bkappa_i(\subss x i)$ defined in \eqref{eq:kappainvcomputed} is homogeneous.
               \begin{lem}
                 \label{lem:homo}
                 For $\subss z i\in\Zset_i$ and $\rho\geq 0$ one has
                 $$
                 \bkappa_i^s(\rho\subss z i)=\rho\bkappa_i^s(\subss z i),~s=0,\ldots,k_i-1
                 $$
                 and hence $\bkappa_i(\rho\subss z i)=\rho\bkappa_i(\subss z i)$.
                 \begin{proof}
                   Let $\bar\beta_i^{(s,f)}$, $f\in 1:q_i$, $s\in 0:k_i-1$ and $\bar\mu$ be the optimizers to $\bar\Pset_i(\subss z i)$. One can easily verify that $\beta_i^{(s,f)}=\rho\bar\beta_i^{(s,f)}$ and $\mu=\rho\bar\mu$ fulfill the constraints \eqref{eq:kappainvBeta}-\eqref{eq:pzInvproblem} for $\bar\Pset_i(\rho\subss z i)$. We show now that these values are also optimal for $\bar\Pset_i(\rho\subss z i)$. By contradiction, assume that $\tilde\beta_i^{(s,f)}$, $\tilde\mu$ are the optimizers to $\bar\Pset_i(\rho\subss z i)$ giving the optimal cost $\tilde\mu<\rho\bar\mu$. One can easily verify that $\beta_i^{(s,f)}=\rho^{-1}\tilde\beta_i^{(s,f)}$ and $\mu=\rho^{-1}\tilde\mu$ verify the constraints \eqref{eq:kappainvBeta}-\eqref{eq:pzInvproblem} for $\bar\Pset_i(\subss z i)$ and yield a cost $\rho^{-1}\tilde\mu<\bar\mu$. This contradicts the optimality of $\bar\mu$ for $\bar\Pset_i(\subss z i)$.
                 \end{proof}
               \end{lem}
               \begin{lem}
                 \label{lem:shrink}
                 If $\subss x i\in\rho_i\Zset_i$ and $\subss{\tilde w} i\in\eta_i\bZset_i^0$ where $\rho_i\geq\eta_i>0$, then $\subss \px i=A_{ii}\subss x i+B_i\bkappa_i(\subss x i)+\subss{\tilde w} i\in\rho_i\Zset_i\ominus(\rho_i-\eta_i)\bZset_i^0$.
                 \begin{proof}
                   Let $\subss\tx i=\frac{\subss x i}{\rho_i}$ and $\subss{\tilde w} i =\frac{\subss w i}{\eta_i}$. From standard arguments in \cite{Rakovic2010} one can write $\subss\tx i\in\Zset_i$ as
                   $$
                   \subss\tx i=(1-\alpha_i)^{-1}\sum_{s=0}^{k_i-1}\sigma_i^s(\subss\tx i)
                   $$
                   where $\sigma_i^s(\subss\tx i)\in\bZset^s_i$ are suitable functions. Then one has
                   $$
                   \subss x i=\rho_i(1-\alpha_i)^{-1}\sum_{s=0}^{k_i-1}\sigma_i^s(\subss\tx i)
                   $$
                   Furthermore, always from \cite{Rakovic2010} one has
                   $$
                   \bkappa_i(\subss\tx i)=(1-\alpha_i)^{-1}\sum_{s=0}^{k_i-1}\bkappa_i^s(\subss\tx i)
                   $$
                   and, from Lemma \ref{lem:homo},
                   $$
                   \bkappa_i(\subss x i)=\rho_i(1-\alpha_i)^{-1}\sum_{s=0}^{k_i-1}\bkappa_i^s(\subss\tx i).
                   $$
                   Computing $\subss\px i$ one gets
                   \begin{subequations}
                     \label{eq:pxshrink}
                     \begin{align}
                       \label{eq:pxshrink1}\subss x i&=\rho_i(1-\alpha_i)^{-1}\sum_{s=0}^{k_i-1}(\underbrace{A_{ii}\sigma_i^s(\subss\tx i)+B_i\bkappa_i(\subss\tx i)}_{t_i^s})+\subss{\tilde w} i\\
                       \label{eq:pxshrink2}                &\in\left[\rho_i(1-\alpha_i)^{-1}\bigoplus_{s=0}^{k_i-1}\bZset_i^{s+1}\right]\oplus\left[\eta_i\bZset_i^0\right]\\
                       \label{eq:pxshrink3}                &=\left[\rho_i(1-\alpha_i)^{-1}\bigoplus_{s=1}^{k_i-1}\bZset_i^{s}\right]\oplus\left[\rho_i(1-\alpha_i)^{-1}\bZset_i^{k_i}\right]\oplus\left[\eta_i\bZset_i^0\right]\\
                       \label{eq:pxshrink4}                &=\left\{\left[\rho_i(1-\alpha_i)^{-1}\bigoplus_{s=1}^{k_i-1}\bZset_i^s\right]\oplus\rho_i\alpha_i(1-\alpha_i)^{-1}\bZset_i^0\oplus\rho_i\bZset_i^0\oplus\eta_i\bZset_i^0\right\}\ominus\rho_i\bZset_i^0\\
                       \label{eq:pxshrink5}                &=\left\{\left[ \rho_i(1-\alpha_i)^{-1}\bigoplus_{s=0}^{k_i-1}\bZset_i^s\right]\oplus\eta_i\bZset_i^0\right\}\ominus\rho_i\bZset_i^0\\
                       \label{eq:pxshrink6}                &=\left\{\rho_i\Zset_i\oplus\eta_i\bZset_i^0\right\}\ominus\rho_i\bZset_i^0\\
                       \label{eq:pxshrink7}                &\subseteq\rho_i\Zset_i\ominus(\rho_i-\eta_i)\Zset_i^0
                     \end{align}
                   \end{subequations}
                   that is the desired result. Note that
                   \begin{itemize}
                   \item in \eqref{eq:pxshrink2} we used $t_i^s\in\bZset_i^{s+1}$ (that holds by construction of sets $\bZset_i^s$);
                   \item in \eqref{eq:pxshrink4} we used the property that $(A\oplus B)\ominus B=A$, if $A\subseteq\Rset^n$ and $B\subseteq\Rset^n$ are convex bodies (see Lemma 3.18 in \cite{Schneider1993}), and the fact that \eqref{eq:Thetai} implies $\bZset_i^{k_i}\subseteq\alpha_i\bZset_i^0$;
                   \item in \eqref{eq:pxshrink5} we used $\rho_i\alpha_i(1-\alpha_i)^{-1}\bZset_i^0\oplus\rho_i\bZset_i^0=\rho_i(\alpha_i(1-\alpha_i)^{-1}+1)\bZset_i^0=\rho_i(1-\alpha_i)^{-1}\bZset_i^0$;
                   \item in \eqref{eq:pxshrink6} we used \eqref{eq:Zrcilambda};
                   \item in \eqref{eq:pxshrink7} we used the property that if $\rho>\eta>0$ and $A$ and $B$ are convex bodies, then
                     \begin{equation}
                       \label{eq:propertyConvexBody}
                       (\rho A\oplus\eta B)\ominus\rho B\subseteq\rho A\ominus(\rho-\eta)B.
                     \end{equation}
                     The inclusion \eqref{eq:propertyConvexBody} can be shown as follows. One has, from the definition of the operator $\ominus$,
                     $$
                     x\in(\rho A\oplus\eta B)\ominus\rho B\Leftrightarrow x\oplus\rho B\subseteq\rho A\oplus\eta B
                     $$
                     and therefore, using Lemma \ref{lem:SFlem}
                     \begin{equation}
                       \label{eq:hsupport}
                       \begin{aligned}
                         h_{\{x\}}+\rho h_B&\leq \rho h_A+\eta h_B\\
                         h_{\{x\}}+(\rho-\eta)h_B&\leq\rho h_A.
                       \end{aligned}
                     \end{equation}
                     Since $(\rho-\eta)>0$, one has that $(\rho-\eta)h_B$ is the support function of $(\rho-\eta)B$. Then \eqref{eq:hsupport} is equivalent to
                     $$
                     x\oplus(\rho-\eta)B\subseteq\rho A
                     $$
                     that is $x\in\rho A\ominus(\rho-\eta)B$.
                   \end{itemize}
                 \end{proof}
               \end{lem}

               \textbf{Proof of Theorem \ref{thm:mainclosedloop}}\\
               The first part of the proof uses arguments similar to the ones adopted for proving Theorem 1 both in \cite{Farina2012} and \cite{Riverso2012a,Riverso2013c}.\\
               We first show recursive feasibility, i.e. that $\subss x i(t)\in\Xset_i^N,~\forall i\in\MM$ implies $\subss x i(t+1)\in\Xset_i^N$.\\
               Assume that, at instant $t$, $\subss x i(t)\in\Xset_i^N$. The optimal nominal input and state sequences obtained by solving each MPC-$i$ problem $\Pset_i^N$ are $\subss v i(0:N_i-1|t)=\{\subss v i(0|t),\ldots,\subss v i(N_i-1|t)\}$ and $\subss \hx i(0:N_i|t)=\{\subss \hx i(0|t),\ldots,\subss \hx i(N_i|t)\}$, respectively. Define $\subss{v}{i}^{aux}(N_i|t)=\kappa_i^{aux}(\subss\hx i(N_i|t))$ and compute $\subss{\hx}{i}^{aux}(N_i+1|t)$ according to \eqref{eq:dynproblem} from $\subss\hx i(N_i|t)$ and $\subss v i(N_i|t)=\subss{v}{i}^{aux}(N_i|t)$. Note that, in view of constraint \eqref{eq:inTerminalSet} and points (\ref{enu:invariantAux}) and (\ref{enu:terminalSetAux}) of Assumption \ref{ass:axiomsMPC}, $\subss{v}{i}^{aux}(N_i|t)\in\Vset_i$ and  $\subss{\hx}{i}^{aux}(N_i+1|t)\in\hXset_{f_i}\subseteq\hXset_i$. We also define the input sequence
               \begin{equation}
                 \label{eq:feasibleInput}
                 \subss\bv i(1:N_i|t)=\{\subss v i(1|t),\ldots,\subss v i(N_i-1|t),\subss{v}{i}^{aux}(N_i|t)\}
               \end{equation}
               and the state sequence produced by \eqref{eq:dynproblem} from the initial condition $\subss\hx i(0|t)$ and the input sequences $\subss\bv i(1:N_i|t)$,
               \begin{equation}
                 \label{eq:feasibleState}
                 \subss\hbx i(1:N_i+1|t+1)=\{\subss\hx i(1|t),\ldots,\subss\hx i(N_i|t),\subss{\hx}{i}^{aux}(N_i+1|t)\}.
               \end{equation}
               In view of the constraints \eqref{eq:decMPCProblem} at time $t$ and recalling that $\Zset_i$ is an RCI set, we have that $\subss x i(t+1)-\subss\hx i(1|t)\in\Zset_i$. Therefore, we can conclude that the state and the input sequences $\subss\hbx i(1:N_i+1|t)$ and $\subss\bv i(1:N_i|t)$ are feasible at $t+1$, since constraints \eqref{eq:inZproblem}-\eqref{eq:inTerminalSet} are satisfied. This proves recursive feasibility.

               We now prove convergence of the optimal cost function.

               We define $\Pset_i^{N,0}(\subss\hx i(0|t))=\min_{\subss v i(0:N_i-1|t)}\sum_{k=0}^{N_i-1}\ell_i(\subss\hx i(k),\subss v i(k))+V_{f_i}(\subss\hx i(N_i))$ subject to the constraints \eqref{eq:dynproblem}-\eqref{eq:inTerminalSet}. By optimality, using the feasible control law \eqref{eq:feasibleInput} and the corresponding state sequence \eqref{eq:feasibleState} one has
               \begin{equation}
                 \label{eq:optimcost}
                 \Pset_i^{N,0}(\subss\hx i(1|t))\leq\sum_{k=1}^{N_i}\ell_i(\subss\hx i(k|t),\subss v i(k|t))+V_{f_i}(\subss{\hx}{i}^{aux}(N_i+1|t+1))
               \end{equation}
               where it has been set $\subss v i(N_i|t)=\subss{v}{i}^{aux}(N_i|t)$. Therefore we have
               \begin{equation}
                 \label{eq:decrease}
                 \begin{aligned}
                   \Pset_i^{N,0}(\subss\hx i(1|t))-\Pset_i^{N,0}(\subss\hx i(0|t))\leq&-\ell_i(\subss\hx i(0|t),\subss v i(0|t))+\ell_i(\subss\hx i(N_i|t),\subss{v}{i}^{aux}(N_i|t))+\\
                   &+V_{f_i}(\subss{\hx}{i}^{aux}(N_i+1|t))-V_{f_i}(\subss{\hx}{i}^{aux}(N_i|t)).
                 \end{aligned}
               \end{equation}
               In view of Assumption \ref{ass:axiomsMPC}-(\ref{enu:decterminal}), from \eqref{eq:decrease} we obtain
               \begin{equation}
                 \label{eq:decreasLyap}
                 \Pset_i^{N,0}(\subss\hx i(1|t))-\Pset_i^{N,0}(\subss\hx i(t))\leq-\ell_i(\subss\hx i(t),\subss v i(t))
               \end{equation}
               and therefore $\subss\hx i(0|t)\rightarrow 0$ and $\subss v i(0|t)\rightarrow 0$ as $t\rightarrow\infty$.

               Next we prove stability of the origin for the closed-loop system. We highlight that this part of the proof differs substantially from the proof of Theorem 1 both in \cite{Farina2012} and \cite{Riverso2012a,Riverso2013c}. For the sake of clarity, the proof is split in three distinct steps.
               \paragraph{Step 1} Prove that if $\mbf x(0)\in\Xset^N$ there is $\tilde T>0$ such that $\mbf x(\tilde T)\in\Zset$.

                    Recalling that the state $\mbf x(t)$ evolves according to the equation \eqref{eq:controlled_model}, we can write
                    \begin{equation}
                      \label{eq:controlled_model_rewrite}
                      \mbf x(t+1)=\mbf{A_Dx}(t)+\mbf B\bkappa(\mbf x(t))+\mbf{A_cx}(t)+\bar\eta(t)
                    \end{equation}
                    where
                    \begin{equation}
                      \label{eq:bareta}
                      \bar\eta(t)=\mbf{B}( \mbf v(t)+\bkappa(\mbf z(t))-\bkappa(\mbf x(t)) )
                    \end{equation}
                    and $\mbf z(t)=\mbf x(t)-\mbf\hx(0|t)$. In particular, if $\mbf x(0)\in\Xset^N$, recursive feasibility shown above implies that \eqref{eq:controlled_model_rewrite} holds for all $t\geq 0$.\\
                    Note that in view of Assumption \ref{ass:Z0}, the LP problem \eqref{eq:kappainv} is feasible for all $\subss z i\in\Rset^{n_i}$. Indeed \eqref{eq:kappainvSum1} and \eqref{eq:pzInvproblem} require that there are $\mu>0$ and $\subss\bz i^s\in\mu\bZset_i^s$, $s\in 0:k_i-1$ such that $\subss z i=\sum_{i=0}^{k_i-1}\subss\bz i^s$ and since $\bZset_i^0\supset\ball{\omega_i}(0)$ (i.e. $\bar\Zset_i^0$ is full dimensional), these quantities always exist. This implies that the function $\bkappa(\mbf x(t))$ in \eqref{eq:controlled_model_rewrite} is always well defined.\\
                    We already proved the asymptotic convergence to zero of the nominal state $\mbf \hx(0|t)$ and the input signal $\mbf v(0|t)$ and hence we proved that
                    \begin{equation}
                      \label{eq:hxvdelta}
                      \forall\delta>0,~\exists T_1>0:~\norme{\mbf \hx(0|t)}{}\leq\delta\mbox{ and }\norme{\mbf v(0|t)}{}\leq\delta,~\forall t\geq T_1.
                    \end{equation}
                    Moreover function $\bkappa(\cdot)$ is piecewise affine and continuous. In view of this, $\bkappa(\cdot)$ is also globally Lipschitz, i.e.
                    \begin{equation}
                      \label{eq:lipschitz}
                      \exists~L>0~:~\norme{\bkappa(\mbf x-\mbf\hx)-\bkappa(\mbf x)}{}\leq L\norme{\mbf\hx}{}
                    \end{equation}
                    for all $(\mbf x,\mbf\hx)$ such that $\mbf x\in\Xset$ and $\mbf x-\mbf\hx\in\Zset$. Using \eqref{eq:lipschitz} one can show that, for all $\epsilon >0$, setting $\delta=\frac{\epsilon}{\norme{B}{}(1+L)}$ the following implication holds
                    \begin{equation}
                      \label{eq:etaepsilon}
                      \norme{\mbf \hx(0|t)}{}\leq\delta\mbox{ and }\norme{\mbf v(0|t)}{}\leq\delta\Rightarrow\norme{\bar\eta(t)}{}\leq\epsilon,~\forall x(t)\in\Xset.
                    \end{equation}
                    Therefore, from \eqref{eq:hxvdelta}
                    \begin{equation}
                      \label{eq:epsilonT}
                      \forall\epsilon>0,~\exists T_1>0:~\norme{\bar\eta(t)}{}\leq\epsilon,~\forall t\geq T_1.
                    \end{equation}
                    Since $\norme{\mbf\hx(0|t)}{}\rightarrow 0$, as $t\rightarrow\infty$, and $\Zset$ contains $\prod_{i=1}^M\ball{\omega_i}(0)$, then
                    \begin{equation}
                      \label{eq:deltaz}
                      \forall\delta_z>0,~\exists T_2>0:~\mbf\hx(0|t)\in\delta_z\Zset,~\forall t\geq T_2
                    \end{equation}
                    and hence, from \eqref{eq:inZproblem},
                    \begin{equation}
                      \label{eq:xafterT2}
                      \mbf x(t)=\mbf \hx(0|t)+(\mbf x(t)-\mbf\hx(0|t))\in(1+\delta_z)\Zset,~\forall t\geq T_2.
                    \end{equation}
                    From \eqref{eq:controlled_model_rewrite} we have, for all $i\in\MM$,
                    \begin{equation}
                      \label{eq:fromCollectivemodel}
                      \subss x i(t+1)=A_{ii}\subss x i(t)+B_i\bkappa_i(\subss x i(t))+\subss{\tilde w}i(t).
                    \end{equation}
                    where $\subss{\tilde w} i=\sum_{j\in\NN_i}A_{ij}\subss x j+\subss{\bar\eta} i$, $\forall i\in\MM$. Setting $\bar T=\max\{T_1,T_2\}$ and using \eqref{eq:epsilonT} and \eqref{eq:xafterT2}, one has, $\forall t\geq\bar T$
                    \begin{equation}
                      \label{eq:tildewin1plusdeltaz}
                      \subss{\tilde w} i\in(1+\delta_z)\bigoplus_{j\in\NN_i}A_{ij}\Zset_j\oplus\ball\epsilon(0).
                    \end{equation}
                    From Assumption \ref{ass:constr_satisf} we have
                    \begin{subequations}
                      \begin{align}
                        \label{eq:setPredecessorsa}\bigoplus_{j\in\NN_i}A_{ij}\Zset_j&\subseteq\bigoplus_{j\in\NN_i}A_{ij}\left[\Xset_j\ominus\ball{\rho_{j,1}}(0)\right]\\
                        \label{eq:setPredecessorsb}                                               &\subseteq\bigoplus_{j\in\NN_i}\left[\left(A_{ij}\Xset_j\right)\ominus\left(A_{ij}\ball{\rho_{j,1}}(0)\right)\right]\\
                        \label{eq:setPredecessorsc}                                               &\subseteq\left(\bigoplus_{j\in\NN_i}A_{ij}\Xset_j\right)\ominus\left(\bigoplus_{j\in\NN_i}A_{ij}\ball{\rho_{j,1}}(0)\right)\\
                        \label{eq:setPredecessorsd}                                               &\subseteq\Wset_i\ominus\left(\bigoplus_{j\in\NN_i}A_{ij}\ball{\rho_{j,1}}(0)\right)\\
                      \end{align}
                    \end{subequations}
                    Manipulations \eqref{eq:setPredecessorsb} and \eqref{eq:setPredecessorsc} are justified as follows. Let $U_1$, $U_2$, $V_1$, $V_2$ be convex bodies in $\Rset^n$. Then \cite{Markov1998,Schneider1993},
                    \begin{itemize}
                    \item \eqref{eq:setPredecessorsb} follows from $G(U\ominus V)\subseteq GU\ominus GV$, where $G\in\Rset^{n\times n}$,
                    \item \eqref{eq:setPredecessorsc} follows from $(U_1\ominus V_1)\oplus(U_2\ominus V_2)\subseteq (U_1\oplus U_2)\ominus(V_1\oplus V_2)$.
                    \end{itemize}
                    Therefore, there is $\xi_i\in[0,1)$ (that does not depend on $\epsilon$ and $\delta_z$) such that
                    \begin{equation}
                      \label{eq:sumZjinepsW}
                      \bigoplus_{j\in\NN_i}A_{ij}\Zset_j\subseteq\xi_i\Wset_i,
                    \end{equation}
                    and then, from \eqref{eq:tildewin1plusdeltaz},
                    \begin{equation}
                      \label{eq:tildewindeltazepsW}
                      \subss{\tilde w}i\in(1+\delta_z)\xi_i\Wset_i\oplus\ball{\epsilon}(0),~\forall t\geq\bar T.
                    \end{equation}
                    Note that in \eqref{eq:epsilonT} the parameter $\epsilon>0$ can be chosen arbitrarily small. Assume that it verifies $\epsilon<(1+\delta_z)\xi_i\omega_i$, $\forall i\in\MM$ where $\omega_i$ are the radii of the balls in Assumption \ref{ass:Z0}. Then, using Assumption \ref{ass:Z0} we get for $t\geq\bar T$
                    \begin{equation}
                      \label{eq:wtildei}
                      \subss{\tilde w}i(t)\in(1+\delta_z)\xi_i(\Wset_i\oplus\ball{\omega_i}(0))\subseteq(1+\delta_z)\xi_i\bZset_i^0.
                    \end{equation}
                    In view of \eqref{eq:xafterT2} and \eqref{eq:wtildei}, Lemma \ref{lem:outerbound} guarantees that
                    \begin{equation}
                      \label{eq:plusxindeltazZ}
                      \subss\px i\in (1+\delta_z)(\Zset_i\ominus(1-\xi_i)\bZset_i^0)
                    \end{equation}
                    From Assumption \ref{ass:Z0}, one has $\Zset_i\ominus (1-\xi_i)\bZset_i^0\subset\Zset_i\ominus\ball{(1-\xi_i)\omega_i}(0)$ and hence, since $\Zset_i$ contains the origin in its interior, there is $\mu_i\in [0,1)$ such that $\Zset_i\ominus (1-\xi_i)\Zset_i^0\subset\mu_i\Zset_i$. From \eqref{eq:plusxindeltazZ} we get
                    \begin{equation}
                      \label{eq:plusxinmuZi}
                      \subss\px i\in (1+\delta_z)\mu_i\Zset_i.
                    \end{equation}
                    If in \eqref{eq:deltaz} we set $\delta_z$ such that $(1+\delta_z)\mu_i<1$, we have shown that for $t=\bar T$ it holds
                    $$
                    \subss x i(\bar T+1)\in\Zset_i
                    $$
                    and the proof of Step 1 is concluded setting $\tilde T=\bar T+1$.

               \paragraph{Step 2} Prove that if $\mbf x(0)\in\Xset^N$, then $\mbf x(t)\rightarrow 0$ as $t\rightarrow +\infty$.

                    Set $t=\tilde T$. Since from Step 1 (that holds if $\mbf x(0)\in\Xset^N$), one has $\subss x i(t)\in\Zset_i$, under Assumption \ref{ass:axiomsMPC}, the optimizers to $\Pset_i^N(\subss x i(t))$ are $\subss v i(0|t)=0$ and $\subss \hx i(0|t)=0$. Hence, from \eqref{eq:bareta} one has $\bar\eta(t)=0$ and \eqref{eq:fromCollectivemodel} holds with $\subss{\tilde w} i(t)=\sum_{j\in\NN_i}A_{ij}\subss x j(t)$. Furthermore, since $\subss u i=\bkappa_i(\subss x i)$ is the control law that makes the set $\Zset_i$ RCI with respect to $\subss{\tilde w} i$, one has $\subss x i(t+1)\in\Zset_i$. The previous arguments can be applied in a recursive fashion showing that, $\forall t\geq\tilde T$
                    \begin{equation}
                      \label{eq:xiinZi}
                      \subss x i(t)\in\Zset_i
                    \end{equation}
                    \begin{equation}
                      \label{eq:tildewiequalAijxj}
                      \subss{\tilde w}i(t)=\sum_{j\in\NN_i}A_{ij}\subss x j(t).
                    \end{equation}
                    From \eqref{eq:xiinZi}, \eqref{eq:tildewiequalAijxj}, \eqref{eq:sumZjinepsW} and Assumption \ref{ass:Z0} we also have
                    \begin{equation}
                      \label{eq:tildewiinxiZseti0}
                      \subss{\tilde w}i(t)\in\bigoplus_{j\in\NN_i}A_{ij}\Zset_j\subseteq\xi_i\Wset_i\subseteq\xi_i\bar\Zset_i^0.
                    \end{equation}
                    Set $t=\tilde T$ and $\lambda_i(t)=1$, $\forall i\in\MM$. From \eqref{eq:xiinZi} and \eqref{eq:tildewiinxiZseti0} it holds
                    \begin{equation}
                      \label{eq:xintildelambdaZseti}
                      \subss x i(t)\in\lambda_i(t)\Zset_i
                    \end{equation}
                    \begin{equation}
                      \label{eq:tildewiinlambdaixibsetZi0}
                      \subss{\tilde w} i(t)\in\lambda_i(t)\xi_i\bZset_i^0.
                    \end{equation}
                    From Lemma \ref{lem:shrink} we have
                    \begin{equation}
                      \label{eq:xiplusinZminuZ0}
                      \subss x i(t+1)\in\lambda_i(t)\Zset_i\ominus(1-\xi_i)\lambda_i(t)\bZset_i^0.
                    \end{equation}
                    From Assumption \ref{ass:Z0}, it holds $\ball{\omega_i}(0)\subseteq\bZset_i^0$. Then
                    \begin{equation}
                      \label{eq:xt1inlambdaZ}
                      \subss x i(t+1)\in\lambda_i(t)\Zset_i\ominus\ball{(1-\xi_i)\lambda_i(t)\omega_i}(0).
                    \end{equation}
                    The next goal is to compute $\tilde\lambda_i(t+1)<\lambda_i(t)$ such that
                    \begin{equation}
                      \label{eq:xt1inlambdatildeZ}
                      \subss x i(t+1)\in\tilde\lambda_i(t+1)\Zset_i.
                    \end{equation}
                    This will be achieved using Lemma \ref{lem:outerbound} with $A=\Zset_i$ and $\beta=(1-\xi_i)\omega_i$. We have to verify that $A\ominus\ball{\beta}(0)$ contains the origin in its interior and since $\ball{\omega_i}(0)\in\bZset_i^0$, it is enough to show that the smaller set $\Zset_i\ominus (1-\xi_i)\bZset_i^0$ contains the origin in its interior.\\
                    From \eqref{eq:Zrcilambda} one has
                    \begin{equation}
                      \label{eq:Ziinclusions}
                      \begin{aligned}
                        \Zset_i\ominus (1-\xi_i)\bZset_i^0&=\left( \left( (1-\alpha_i)^{-1}\bigoplus_{s=1}^{k_i-1}\bar\Zset_i^s \right) \oplus (1-\alpha_i)^{-1}\bZset_i^0 \right)\ominus (1-\xi_i)\bZset_i^0\\
                        &\supseteq\left( (1-\alpha_i)^{-1}\bZset_i^0\ominus (1-\xi_i)\bZset_i^0 \right)\oplus \left((1-\alpha_i)^{-1}\bigoplus_{s=1}^{k_i-1}\Zset_i^s \right)
                      \end{aligned}
                    \end{equation}
                    where the last inclusion follows from the fact that for generic sets $C$, $D$ and $E$ in $\Rset^n$, it holds \cite{Markov1998}
                    $$
                    (C\oplus D)\ominus E\supseteq (D\ominus E)\oplus C.
                    $$
                    Note that the origin is strictly contained in $\bZset_i^0$ (from Assumption \ref{ass:Z0}) and also in sets $\bZset_i^s$, $s\in 1:k_i-1$, (by construction). Since in \eqref{eq:Zrcilambda} $\alpha_i\in [0,1)$ and we have chosen $\xi_i\in [0,1)$, one has $(1-\alpha_i)^{-1}>(1-\xi_i)$ and therefore
                    $$
                    0\in (1-\alpha_i)^{-1}\bZset_i^0\ominus (1-\xi_i)\bZset_i^0.
                    $$
                    Since the origin is strictly contained in all summands appearing in \eqref{eq:Ziinclusions}, it is also strictly contained in $A\ominus\ball{\beta}(0)$. Letting $\bZset_i^0=\{z\in\Rset^{n_i}:\bar\ZZ_i^0 z\leq\One\}$ and $\bar\ZZ_i^0=[\bar z_{i,1}^0,\ldots,\bar z_{i,q_i}^0]^T$, from point (\ref{enu:outerboundc}) of Lemma \ref{lem:outerbound} we get
                    $$
                    \Zset_i^0\ominus\ball{(1-\xi_i)\omega_i}(0)\subseteq (1-(1-\xi_i)\omega_i\psi_i)\Zset_i,~\psi_i=\min_{j\in 1:q_i}\norme{\bar\ZZ_{i,j}^0}{}.
                    $$
                    From \eqref{eq:xt1inlambdaZ}, one obtains that \eqref{eq:xt1inlambdatildeZ} is fulfilled with
                    \begin{equation}
                      \label{eq:dyntildelambda}
                      \left\{
                        \begin{aligned}
                          \tilde\lambda_i(t+1)&=a_i\lambda_i(t)\\
                          a_i&=1-(1-\xi_i)\omega_i\psi_i
                        \end{aligned}
                      \right.
                    \end{equation}
                    and from point (\ref{enu:outerboundb}) of Lemma \ref{lem:outerbound}, it holds $\abs{a_i}<1$.\\
                    From \eqref{eq:tildewiequalAijxj} and \eqref{eq:xt1inlambdatildeZ} we have
                    $$
                    \subss{\tilde w}i(t+1)\in\bigoplus A_{ij}\tilde\lambda_j(t+1)\Zset_j
                    $$
                    and setting
                    \begin{equation}
                      \label{eq:dynlambdamaxtildelambda}
                      \lambda_i(t+1)=\max_{{j\in\NN_i}\cup\{i\}}\tilde\lambda_j(t+1)
                    \end{equation}
                    it holds
                    \begin{equation}
                      \label{eq:tildewitplus1in}
                      \subss{\tilde w}i(t+1)\in\lambda_i(t+1)(\bigoplus_{j\in\NN_i}A_{ij}\Zset_j)\subseteq\lambda_i(t+1)\xi_i\bZset_i^0
                    \end{equation}
                    where we used \eqref{eq:sumZjinepsW} and Assumption \ref{ass:Z0}. Note that \eqref{eq:tildewitplus1in} corresponds to \eqref{eq:tildewiinlambdaixibsetZi0} at time $t+1$. Similarly, from \eqref{eq:xt1inlambdatildeZ} and \eqref{eq:dynlambdamaxtildelambda} we have that \eqref{eq:xintildelambdaZseti} holds at time $t+1$. Furthermore $0\leq\tilde\lambda_i(t+1)<1$, $0\leq\lambda_i (t+1)<1$. Therefore, the previous arguments can be applied in a recursive fashion to prove that \eqref{eq:xintildelambdaZseti}, \eqref{eq:tildewiinlambdaixibsetZi0}, \eqref{eq:dyntildelambda} and \eqref{eq:dynlambdamaxtildelambda} hold for all $t\geq\tilde T$ and $i\in\MM$. \\
                    In order to conclude the proof of Step 2, we show that $\lambda_i(t)$ given by \eqref{eq:dynlambdamaxtildelambda} converge to zero as $t\rightarrow +\infty$. Indeed, from \eqref{eq:xintildelambdaZseti} this implies that $\mbf x(t)\rightarrow 0$ as $t\rightarrow\infty$. Let $\bar{\lambda}(t)=\max_{i\in\MM}\lambda_i(t)$ and $\bar a=\max_{i\in\MM}a_i$. From \eqref{eq:dyntildelambda} and \eqref{eq:dynlambdamaxtildelambda} one has
                    \begin{equation}
                      \label{eq:lambdabartplus1}
                      \bar\lambda (t+1)=\max_{i\in\MM}\lambda_i (t+1)=\max_{i\in\MM}\max_{j\in\NN_i\cup\{i\}} a_i\lambda_i (t)\leq \bar a\max_{i\in\MM}\max_{j\in\NN_i\cup\{i\}}\lambda_i(t)=\bar a\bar\lambda (t)
                    \end{equation}
                    Being $a\in [0,1)$ and $\bar\lambda (\tilde T)=1$, \eqref{eq:lambdabartplus1} implies that $\bar\lambda (t)\rightarrow 0$ as $t\rightarrow\infty$ and this concludes the proof of Step 2.

               \paragraph{Step 3} Prove stability of the origin of the closed-loop system \eqref{eq:controlled_model}.

                    Note that
                    \begin{equation}
                      \label{eq:ZiinXiN}
                      \Zset_i\subseteq\Xset_i^N
                    \end{equation}
                    as shown at the beginning of Step 2. Moreover, from Assumption \ref{ass:Z0} and \eqref{eq:Zrcilambda} (where, from \eqref{eq:Thetai}, $\alpha_i\in [0,1)$)
                    $$
                    \ball{\omega_i}(0)\subseteq\Zset_i
                    $$
                    and then $\Zset$ is a neighborhood of the origin of $\Rset^n$.
For a given $\epsilon>0$ let $\rho>0$ be such that $\rho<1$ and $\rho \Zset\subseteq B_\epsilon(0)$. As shown at the beginning of Step 2, if $\mbf x\in\Zset$ then problems $\Pset_i^N(\subss x i)$, $i\in\MM$ are feasible, the closed-loop dynamics \eqref{eq:controlled_model} reduces to
                    \begin{equation}
                      \label{eq:closedloopinZ}
                      \mbf \px = \mbf{Ax+B\bkappa(x)}
                    \end{equation}
                    and $\Zset$ is an invariant set for \eqref{eq:closedloopinZ}. This implies that if $\mbf x(0)\in\Zset$, then $\mbf x(t)$ is well defined $\forall t\geq 0$. In order to conclude the proof we have to show that also $\rho\Zset$ is invariant for \eqref{eq:closedloopinZ}. If $\tilde{\mbf x}\in\rho\Zset$ then there is $\mbf x\in\Zset$ such that $\tilde{\mbf x}=\rho\mbf x$. From \eqref{eq:closedloopinZ} and Lemma \ref{lem:homo} one has
                    $$
                    \tilde{\mbf x}^+=\rho\mbf{Ax}+\rho\mbf{B\bkappa(x)}=\rho\mbf\px
                    $$
                    and since $\mbf\px\in\Zset$ then $\tilde{\mbf x}^+\in\rho\Zset$.
               \begin{flushright}$\square$\end{flushright}

     \bibliographystyle{IEEEtran}
     \bibliography{Optimized_PnP_MPC_report}

\end{document}